\documentclass[12pt]{article}
\usepackage{amssymb, amsmath, amsthm,verbatim}
\usepackage{lscape} 
\usepackage{mathdots}
\usepackage{multirow}
\usepackage{thmtools}
\usepackage{mathtools}
\usepackage[makeroom]{cancel}
\usepackage{appendix}
\usepackage{bm}
\usepackage{bbm}
\usepackage{caption}
\usepackage{enumerate} 
\usepackage{graphicx}
\usepackage{natbib}
\usepackage{setspace}
\usepackage[a-1b]{pdfx}
\usepackage{hyperref} 
\usepackage{xcolor}
\hypersetup{
    colorlinks,
    linkcolor={red!50!black},
    citecolor={blue!50!black},
    urlcolor={blue!80!black}
}

\usepackage{tocloft}
\usepackage[explicit]{titlesec}
\usepackage{tikz}
\usetikzlibrary{bending}
\usetikzlibrary{shapes,snakes}
\usepackage{mathrsfs}
\usepackage{subcaption}
\usepackage[top=1.25in, bottom=1.25in, left=1.25in, right=1.25in]{geometry}

\pdfoutput=1

\newtheorem{lemma}{Lemma}
\newtheorem{theorem}{Theorem}
\newtheorem{theorem1}{Theorem}
\newtheorem{proposition}{Proposition}

\newtheorem{definition}{Definition}
\newtheorem{assumption}{Assumption}

\newtheorem{condition}{Condition}

\newtheorem{fact}{Fact}

\DeclareMathOperator{\sgn}{sgn}
\DeclareMathOperator{\diag}{diag}

\def\ul{\underline}
\def\bar{\overline}

\begin{document}
\title{Implicit Incentive Provision with Misspecified Learning}
\author{
Federico Echenique\footnote{Department of Economics, University of California, Berkeley, \href{mailto:fede@econ.berkeley.edu}{fede@econ.berkeley.edu}. }
\and Anqi Li\footnote{Department of Economics,  University of Waterloo, \href{mailto:angellianqi@gmail.com}{angellianqi@gmail.com}. We thank Junnan He for illuminating  discussions. }
}
\date{}

\maketitle
\thispagestyle{empty}
\begin{abstract}
We study misspecified Bayesian learning in principal-agent relationships, where an agent is assessed by an evaluator and rewarded by the market. The agent’s outcome depends on their innate ability, costly effort---whose effectiveness is governed by a productivity parameter---and noise. The market infers the agent's ability from observed outcomes and rewards them accordingly. The evaluator conducts costly assessments to reduce outcome noise, which shape the market's inferences and provide implicit incentives for effort. 

Society---including the evaluator and the market---holds dogmatic, inaccurate beliefs about ability, which distort learning about effort productivity and effort choice. This, in turn, shapes the evaluator's choice of assessment. We describe a feedback loop linking misspecified ability, biased learning about effort, and distorted assessment. We characterize outcomes that arise in stable steady states and analyze their robust comparative statics and learning foundations. Applications to education and labor market reveal how stereotypes can reinforce  across domains---sometimes disguised as narrowing or even reversals of outcome gaps---and how policy interventions targeting assessment can help.

\end{abstract}

\newpage
\setcounter{page}{1}

 \section{Introduction}\label{sec_intro}
We study a principal-agent model with misspecified learning that features an agent and an evaluator as strategic players. An observable outcome depends on the agent's innate ability, effort choice---whose effectiveness is governed by a productivity parameter---and noise. Based on the outcome, the market rewards the agent for their perceived ability and possibly their effort. The evaluator conducts costly assessments that reduce outcome noise, thereby shaping the market’s inferences and the agent's effort incentives. All players are Bayesian, but society---the evaluator and the market---holds dogmatic, possibly inaccurate beliefs about the agent's ability, while learning about effort productivity from observed outcomes. We study how such misspecified learning interacts with strategic choices of assessment and effort, and its implications for equilibrium beliefs and behavior.

We shall first describe two situations that illustrate the ideas behind our model and results. 

Two students, Alice and Bob, are identical except for gender. Their school performance depends on innate ability and study effort, whose effectiveness reflects their aptitude for studying. The school serves as an evaluator, conducting costly assessments that reduce the noise in performance. Ultimately, both students are rewarded by the market based on observed school performances. Now, despite being identical, both the school and the market believe that Alice’s ability is lower than Bob’s. This reflects the \emph{gender–brilliance stereotype}, which posits that girls lack the raw talent emphasized in certain fields such as STEM \citep{leslie2015expectations,bian2017gender}. In reality, Alice and Bob are equally capable, so Alice’s strong performance comes as a surprise. In our framework, society does not update its belief about ability; instead, it misattributes Alice’s success to a greater aptitude for studying. 

In a version of the model where the school values effort \emph{per se} (for instance, higher student performance enhances school reputation and resources), the school scrutinizes Alice more intensively to induce greater effort, which in turn feeds back into her performance and subsequent learning.

Our second situation involves two entry-level workers, Chen and David, who are identical except for race and ethnicity. Job performance reflects innate ability but can be enhanced through costly effort, whose effectiveness captures work aptitude. Performance influences the market’s perception of ability and therefore the worker’s future career prospects. The evaluator is a manager who conducts costly performance appraisals to reduce noise and make performance a more informative signal of ability. Society—including the manager and the market—believes that Chen is less capable than David. Indeed, the \emph{model minority stereotype} portrays certain minorities as lacking creativity, leadership, and social skills \citep{lu2024creativity}.\footnote{Other biases and stereotypes also contribute to the underestimation of minorities’ intrinsic abilities. The \emph{ultimate attribution error} refers to the tendency to attribute out-group successes to external factors such as effort, and their failures to internal traits such as low ability \citep{pettigrew1979ultimate}. The \emph{stereotype incongruity theory} posits that when minority performance contradicts prevailing group stereotypes (e.g., East Asians excelling in leadership or vision), it is attributed to effort rather than innate talent \citep{eagly2002role}.} As with Alice, Chen’s strong performance is misattributed to greater work aptitude. The manager—who values effort because it raises profits—responds by intensifying assessment for Chen, thereby generating divergence in effort choice, performance, and future learning.

Our model captures the situations described above and more. It formalizes the interaction between assessment, effort, and learning. Building on this channel, we analyze equilibrium distortions, their comparative statics and learning foundations. In domains such as education, labor, and politics, our results speak to the origins of stereotypes, their cross-domain spillovers, and their welfare consequences. In addition, the model identifies policy levers that can either mitigate or reinforce these effects.

Our analysis builds on the single-period version of Holmstr{\"o}m’s (1999) career-concerns model, in which an agent’s outcome depends on their random innate ability, effort choice, and noise. In our model, however, noise is reduced through a costly assessment conducted by an evaluator, who values effort as part of a broader objective. Assessment shapes the market’s inferences about ability and provides implicit incentives for effort. Second, effort contributes to the outcome through an productivity parameter, so that the mean outcome depends on mean ability and effort productivity. The agent knows these parameters, but society must learn about them. Learning is Bayesian but starts from a misspecified prior: society holds a dogmatic (degenerate) misbelief about mean ability—intended to capture persistent stereotypes—while the prior over effort productivity includes the truth in its support.

In a \emph{Berk--Nash} equilibrium \citep{esponda2016berk}, all players choose optimally given their beliefs, and beliefs are formed via Bayes' rule from the observed \emph{outcome distribution} induced by equilibrium behavior. Equilibrium beliefs must concentrate on parameters in the prior support that minimize the Kullback–Leibler divergence between perceived and true outcome distributions \citep{berk1966limiting}.  Misbeliefs about ability, therefore, distort beliefs about effort productivity. In our setting with asymmetric information and moral hazard, this further distorts beliefs about effort choices. While the agent chooses effort according to their true productivity, society---knowing only the equilibrium effort strategy---infers effort choice from its own (potentially inaccurate) beliefs. This wedge between perceived and true efforts has novel consequences for equilibrium distortions and their dynamic learning foundations.

We study environments where the agent's effort is a complement to both assessment and productivity; thus, the agent works harder when more productive or when facing stronger incentives. Given the effort strategy, we characterize the marginal effect of assessment on the mean outcome and the evaluator’s optimal assessment, both assumed to increase with society’s beliefs about effort productivity. Theorem \ref{thm_direction} analyzes the direction of equilibrium distortions.

Two scenarios arise. When society underestimates the agent's ability, it reconciles its beliefs with the observable outcome by mistakenly attributing the higher-than-expected outcome to effort. It believes that the agent's effort is more productive than it is in reality. This, in turn, leads to intensified assessment on the part of the evaluator, which attenuates the overestimation of effort productivity. Assessment and belief thus act as substitutes, yielding a unique equilibrium. By contrast, when society overestimates ability, it attributes the lower-than-expected outcome to effort productivity being lower than it actually is. Such beliefs reduce the evaluator's assessment, which triggers even more pessimistic beliefs. When the assessment and belief spaces are not overly constrained, multiple equilibria may arise, some of which exhibit large belief distortions even under small ability misspecifications (as in Panel (a) of Figure \ref{fig:combined}).

We focus on \emph{locally asymptotically stable equilibria in the Lyapunov sense}. We define a map from each initial belief to the KL-minimizing posterior induced by the optimal behavior under the initial belief. Equilibrium beliefs are the fixed points of this map. An equilibrium is stable (resp. unstable) if nearby beliefs are attracted to (resp. repelled from) it under the map.  In the single-dimensional case, equilibrium stability is straightforward to determine and follows an alternating pattern, as illustrated in Panel (a) of Figure \ref{fig:combined}. Appendix \ref{sec_learning} provides learning foundations: in a dynamic model where beliefs and behavior co-evolve over time, society’s posterior beliefs almost surely converge to stable equilibria (each reached with positive probability), but to unstable equilibria with zero probability.

\begin{figure}[h!]
    \centering

    \begin{subfigure}[t]{0.47\textwidth}
        \centering
        \begin{tikzpicture}[scale=1.1]
            
            \draw (1.6,1.6) rectangle (6,6);

            \node[above] at (2.5,3.7) {$\tilde{\psi}(\beta)$};

            \draw[thick, black] plot[domain=1.66:6-0.7, samples=1000] 
                (\x+0.7, {max(1.6, 1.2*sqrt(25 - 0.75*(8*(\x^2+8))/(\x^2)))});    
            \draw[thick, black] plot[domain=1.6:2.37, samples=1000] (\x, 1.6);

            \draw[dashed] (1.6,1.6) -- (6,6);

            \node at (1.9,1.7) {$-$};
            \node at (2.35,2.2) {$-$};
            \node at (2.7,2.85) {$+$};
            \node at (4.4,4.6) {$+$};
            \node at (5.2,5) {$-$};

            \draw[dashed] (2.5,2.5) -- (2.5,1.6);
            \draw[dashed] (4.8,4.8) -- (4.8,1.6);

            \node[below] at (2.5,1.6) {$\hat{\beta}_2$};
            \node[below] at (4.8,1.6) {$\hat{\beta}_1$};
            \node[below] at (1.6,1.6) {$\underline{\beta}$};
            \node[below] at (6,1.6)   {$\overline{\beta}$};
            \node[left]  at (1.6,6)   {$\overline{\beta}$};
        \end{tikzpicture}
        \caption{$\beta$ denotes effort productivity, $\tilde{\psi}$ the belief map, and $\sgn(\tilde{\psi}(\beta)-\beta)$ the direction of adjustment. Eqm beliefs $\hat{\beta}_1$, $\hat{\beta}_2$,  $\underline{\beta}$  follow the pattern of stable $\to$ unstable $\to$ stable.}

    \end{subfigure}
    \hfill
    
    \begin{subfigure}[t]{0.47\textwidth}
        \centering
        \begin{tikzpicture}[scale=1.1]
            
            \draw (1.6,1.6) rectangle (6,6);

            \node[above]      at (3.5,5)    {$\tilde{\psi}(\beta)$};
            \node[above, blue] at (3.8,4.45) {$\tilde{\psi}'(\beta)$};

            \draw[thick, blue] plot[domain=1.66:6-0.9, samples=1000] 
                (\x+0.9, {max(1.6, 1.2*sqrt(25 - 0.75*(8*(\x^2+8))/(\x^2)))});    
            \draw[thick, blue] plot[domain=1.6:2.58, samples=1000] (\x, 1.6);

            \draw[dashed] (1.6,1.6) -- (6,6);

            \draw[thick, black] plot[domain=1.6:6-0.7, samples=1000] 
                (\x+0.7, {max(1.6, 0.6+1.2*sqrt(25 - 0.75*(8*(\x^2+8))/(\x^2)))});

            \node[blue] at (1.9,1.7) {$-$};
            \node[blue] at (2.6,2.5) {$-$};
            \node[blue] at (3.05,3.2) {$+$};
            \node[blue] at (4.4,4.56) {$+$};
            \node[blue] at (5.1,4.95) {$-$};

            \node[below] at (2.78,1.6) {$\hat{\beta}'_2$};
            \node[below] at (4.75,1.6) {$\hat{\beta}'_1$};
            \node[below] at (2.4,1.6)  {$\hat{\beta}_2$};
            \node[below] at (5.5,1.6)  {$\hat{\beta}_1$};

            \draw[dashed] (2.4,1.6)  -- (2.4,2.4);
            \draw[dashed] (4.75,1.6) -- (4.75,4.75);
            \draw[dashed] (2.78,1.6) -- (2.78,2.78);
            \draw[dashed] (5.55,1.6) -- (5.55,5.55);

            \draw[<-, black, thick] plot[domain=4.9:4.95,  samples=100] (\x, 1.6);
            \draw[<-, black, thick] plot[domain=5.1:5.15,  samples=100] (\x, 1.6);
            \draw[<-, black, thick] plot[domain=5.3:5.35,  samples=100] (\x, 1.6);
            \draw[<-, black, thick] plot[domain=1.9:2.0,   samples=100] (\x, 1.6);
            \draw[<-, black, thick] plot[domain=2.05:2.15, samples=100] (\x, 1.6);
            \draw[<-, black, thick] plot[domain=2.2:2.3,   samples=100] (\x, 1.6);

            \node[below] at (1.6,1.6) {$\underline{\beta}$};
            \node[below] at (6,1.6)   {$\overline{\beta}$};
            \node[left]  at (1.6,6)   {$\overline{\beta}$};
        \end{tikzpicture}
        \caption{A shock that lowers assessment intensity shifts the belief map from black to blue. This triggers downward adjustment of beliefs from $\hat{\beta}_1$ to $\hat{\beta}_1'$, and from $\hat{\beta}_2$ to $\ul\beta$. }
    \end{subfigure}

    \caption{The case where society overestimates the agent's ability.}
    \label{fig:combined}
\end{figure}
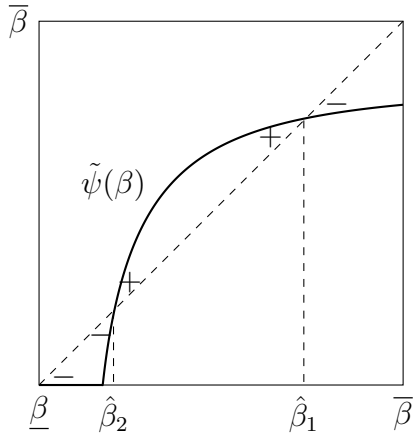
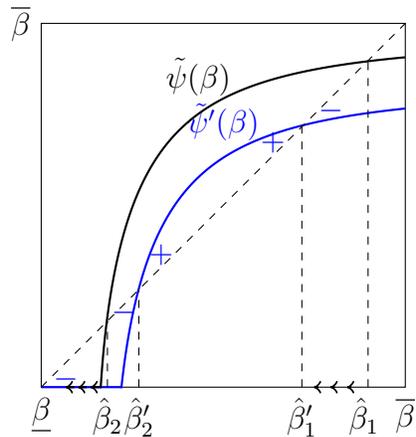

Because multiple stable equilibria may arise, Theorem \ref{thm_magnitude} analyzes their robust comparative statics, showing in particular that intensifying assessment mitigates distortions in the weak set order. Panel (b) of Figure \ref{fig:combined} illustrates an assessment-reducing shock across all productivity values.  When society overestimates ability, the outcome deficit is explained by further lowering beliefs, thereby shifting the belief map downward. Starting from an equilibrium before the shock, the updated belief lies above its original level if the equilibrium is stable (as with $\hat{\beta}_1$) and below if unstable (as with $\hat{\beta}_2$). In both cases, this gap triggers downward adjustment. If the equilibrium is stable, the adjustment settles at a nearby stable point $\hat{\beta}_1'$; if unstable, it continues until the next stable point $\underline{\beta}$ is reached. By contrast, convergence to a nearby unstable point would require an upward adjustment ($\hat{\beta}_2 \uparrow \hat{\beta}_2'$). The implication is clear:  lower assessments amplify distortions at stable equilibria, while comparative statics at unstable equilibria lack a learning foundation.

Section \ref{sec_implication} explores the implications of our results for discrimination. Our first finding concerns the cross-domain spillover of stereotypes. In equilibrium, society overestimates minorities' effort productivity and scrutinizes them more harshly, prompting them to exert greater effort than majorities. Yet the perceived effort gap between the groups is even larger, since inaccurate beliefs about effort productivity are used to infer effort choices. Thus, stereotypes about ability spill over into the domain of effort, reinforcing the \emph{ability–effort tradeoff} in psychology \citep{wang2018systematic}. Section \ref{sec_implication} reviews supporting evidence from education and labor economics, showing that stereotypes about minorities’ ability often induce them to exert exceptional effort in schools, workplaces, and academic or admissions settings \citep{rice2015perfectionism, moscatello2024inventing, dupas2021gender, hengel2022publishing}. Yet this very effort, rather than counteracting bias, may instead entrench it.

The welfare consequences of distortions depend on whether the market rewards effort. If only ability is rewarded—as in traditional career-concerns models—minorities exert greater effort yet earn less than majorities, leading to lower objective welfare. When effort is also rewarded—through curriculum alignments or general human-capital investments that are valued by the broader market  \citep{deming2022four,cullen2013can}—minorities are compensated for their higher perceived effort. This may narrow or even reverse the observed outcome gap, consistent with recent trends in STEM education and labor markets \citep{ceci2014women, bertrand2018coase, bertrand2020gender}, although such convergence may reflect an entrenchment of stereotypes rather than genuine equity, with ambiguous welfare consequences.

Our analysis reveals new policy levers to address discrimination, particularly through intensifying assessment. This can be achieved either by increasing the evaluator’s valuation of effort, or by reducing the cost of effort or assessment. The analysis thus uncovers an unexpected bias-correcting role for policies such as pay-for-performance schemes for teachers and technological innovations that lower assessment costs. It also reveals the bias-amplifying effects of certain educational policies, such as multiple test-retake opportunities or admissions practices that favor affluent students \citep{proulx2024should,chettydiversifying}. These policies implicitly raise the relative effort cost for disadvantaged students and, through the mechanism we propose, ultimately exacerbate distortions between rich and poor.

Appendix \ref{sec_learning} analyzes convergence to stable equilibria in a dynamic learning model. The main challenge arises from strategic uncertainty about the agent’s effort. During the learning process, effort productivity and, hence, effort choice remain uncertain to society, conditional on public histories. Thus even with linear–Gaussian primitives, posteriors may not retain parametric structures. To address this challenge, we approximate posteriors by truncated normal distributions using Laplace’s method, bounding the approximation errors while accounting for the feedback between beliefs and strategic behavior. We then characterize the evolution of posteriors with a stochastic difference equation and establish almost-sure convergence to stable equilibria using stochastic approximation techniques. 

Our model can be extended to analyze the consequences of affirmative action policies. In Section \ref{sec_aa}, we show that enforcing color-sighted but impartial assessment across groups imposes a universal lower bound on equilibrium assessment intensity. This mechansim effectively constrains the assessment space and prevents belief cascades. Meanwhile, allowing for heterogeneous beliefs across groups turns the belief space multidimensional and makes almost-sure convergence to equilibria more challenging to establish. Nevertheless, we demonstrate that when ability misspecifications are small, a unique equilibrium exists, is Lyapunov locally asymptotically stable, and can be reached with positive probability from any full-support prior.

\section{Model}\label{sec_model}
\subsection{Setup}\label{sec_setup}
We study a game between two strategic players: an agent (A) and an evaluator (E).  
A chooses an effort $a \geq 0$ at cost $c(a)$ to influence an outcome $X$. E conducts a costly assessment, choosing an intensity $h \in [0,1]$ at cost $\kappa(h)$. The outcome depends on A's random innate ability $\eta \sim N(\mu,1)$, as well as the effect $r(a,\beta)$ of their effort, parameterized by $\beta \in \mathbb{R}$ referred to as \emph{effort productivity}. Given $(\eta, a, \beta, h)$, 
\[
X \;=\; \underbrace{\eta}_\text{ability} + \underbrace{r(a,\beta)}_\text{effective effort} + \underbrace{\varepsilon}_\text{noise}.
\]
The term \( r(a,\beta) \) captures the contribution of effort  to the mean outcome. 
The noise term \( \varepsilon \sim N(0,\, h^{-1} - 1) \) is independent of ability, and its variance is reduced by E's assessment intensity \( h \): 
zero assessment corresponds to infinite noise, whereas unit assessment yields zero noise. 
Greater assessment intensity increases the precision of the outcome as a signal of ability.

A is rewarded by a third party---the market (M)---who is purely passive. M values A as $\eta + v_{\mathrm{M}}(a,\beta)$, where the function $v_{\mathrm{M}}$ captures M's valuation of effort and productivity.  The special case $v_{\mathrm{M}}= 0$---in which only ability is rewarded but effort is not---has been the focus of the career-concerns literature \citep{holmstrom1999managerial}. Here, we allow $v_{\mathrm{M}}$ to take general forms. 

A's payoff equals the market reward net of effort cost. 

M observes only the outcome but not its individual components (ability, effort, productivity, or noise). Based on the outcome, M forms an expectation of A's value. E's choice of assessment influences M's inferences, thereby providing implicit effort incentives to A (more on this in the next section).

E values effort and productivity at $v_{\mathrm{E}}(a,\beta)$. E may or may not value ability.

We can equivalently express the outcome as  
\begin{align*}
X \;=\; \mu + r(a,\beta) + \frac{\epsilon}{\sqrt{h}},
\quad \text{where } 
\epsilon \coloneqq \sqrt{h}\left(\eta - \mu + \varepsilon\right) \sim N(0,1).
\end{align*}
This makes it clear that $X$ is Gaussian with mean $\mu+r(a,\beta)$ and variance $h^{-1}$.  
We parameterize the mean outcome by $\omega \coloneqq (\mu,\beta)$, referred to as the \emph{fundamental}. The true fundamental is $\omega^*=(\mu^*,\beta^*)$.

Players are Bayesian. The true fundamental $\omega^*$ is known the agent.  Society (S)---comprising E and M---knows this structure but is uncertain about the fundamental value.  Its prior belief is supported on $\Omega=\{\hat{\mu}\}\times[\ul\beta,\bar\beta]$, 
with $\hat{\mu}\neq\mu^*$ and $0<\ul\beta<\beta^*<\bar\beta$. In words, S's prior about A's mean ability is \emph{misspecified}, i.e., it is degenerate and inaccurate. The prior is not misspecified when it comes to effort productivity. The support accommodates a range of productivity, including, in particular, the truth $\beta^*$.

The remainder of the game evolves as follows: 
\begin{enumerate}
    \item E publicly commits to an assessment intensity $h$. 

    \item A privately chooses an effort level $a$.
    \item Ability $\eta \sim N(\mu^*,1)$ and noise $\varepsilon \sim N(0,h^{-1}-1)$ are realized independently and determine the outcome as 
 \[X=\eta+r(a,\beta^*)+\varepsilon.\]
Only the outcome is publicly realized. 
     
    \item M forms an expectation about $\eta+v_{\mathrm{M}}(a,\beta)$ and rewards A.

\end{enumerate}

 We state the (standard) assumptions on model primitives in Appendix~\ref{sec_assumption1}. They are omitted here to streamline the presentation. Section \ref{sec_discussion} discusses the case where E and M may have different priors.

\subsection{Equilibrium} 
We analyze the \emph{Berk--Nash equilibria} of the game \citep{esponda2016berk}. Such an equilibrium captures a steady state in which players optimize given their beliefs about the fundamental and the implied beliefs about each other's behavior.  In turn, beliefs are formed via Bayesian updating from the \emph{outcome distribution} induced by equilibrium behavior. 

In this section, we first elaborate on these equilibrium components and then present the formal definition of equilibrium. In Section \ref{sec_preliminary}, we impose regularity conditions on these components and link them to primitives.  

\vspace{-10pt} \paragraph{Optimal effort.}  
Fix an assessment intensity $h$ and a degenerate belief about the fundamental,  $\delta_{\omega}$ with $\omega \in \Omega$. In equilibrium, these determine M's perception about A's effort, denoted by $\hat{a}$ (calculated below). Then $X$ is believed to be Gaussian with mean $\mu+ r(\hat{a},\beta)$ and variance $h^{-1}$. Conditional on $X$, M's expectation of A's value is
\begin{align*}
\mathbb{E}_{\delta_{(\omega,\hat{a})}}
\!\left[\eta+v_{\mathrm{M}}(a,\beta) \mid X\right]
&=\hat{\mu}+ \frac{\sigma_{\eta}^2}{\sigma_X^2}
   \bigl[X-\hat{\mu}-r(\hat{a},\beta)\bigr]
   +v_{\mathrm{M}}(\hat{a},\beta) \\
&= hX + \text{terms independent of the true effort}.
\end{align*}
Among the terms on the right-hand side, only (the mean of) $hX$ can be influenced  by A's actual choice of  effort. One can interpret $h \in [0,1]$ as the share of the outcome received by A; increasing $h$ strengthens the implicit incentives provided by E's assessment via market reward. 

Knowing the true fundamental $\omega^*$, A solves
\begin{equation}\label{eqn_aproblem}
    \max_{a \geq 0} \; hr(a,\beta^*) - c(a).
\end{equation}
We denote by $a(\cdot,\cdot)$ the optimal effort strategy as a function of assessment $h$ and productivity $\beta$. The true effort is $a(h,\beta^*)$. However, S perceives effort as  $\hat{a}=a(h,\beta)$ when it believes that productivity is $\beta$.

\vspace{-10pt}
\paragraph{Optimal assessment.} E does not observe A's effort, but he knows A's strategy  $a(\cdot,\cdot)$ and can reason by backward induction.  The resulting inference regarding effort may be inaccurate if the belief about productivity is wrong. However, it is rational given what E knows about A and the environment. 

Fix again a degenerate belief about the fundamental, $\delta_{\omega}$. Anticipating that A chooses according to $a(\cdot,\cdot)$, 
E selects an assessment intensity $h \in [0,1]$ to solve 
\begin{equation}\label{eqn_eproblem}
\max_{h \in [0,1]} \; V_{\mathrm{E}}\!\left(h,\beta\right) - \kappa(h), \quad \text{where } V_{\mathrm{E}}(h,\beta)\coloneqq v_{\mathrm{E}}\left(a\left(h,\beta\right),\beta\right). 
\end{equation}
In E's objective function, $\beta$ enters directly through the fundamental channel and indirectly through perceived effort. We denote the solution to E's problem by $h(\beta)$.

\vspace{-10pt}

\paragraph{Beliefs.} The behavior of E and A determines the true outcome distribution, observed by all parties in a steady state. Based on it, S updates beliefs about the fundamental in a Bayesian fashion. The statistical foundation for such beliefs were developed by \cite{berk1966limiting}: posterior beliefs must concentrate on those elements of the prior's support that  minimize the Kullback–Leibler (KL) divergence between the perceived and true outcome distributions. 

In our setting, 
\[R(h,\beta)\coloneqq r\left(a\left(h,\beta\right),\beta\right)\]
is \emph{optimal effective effort} of an agent with productivity $\beta$ who faces  assessment intensity $h$. Given $h$, the true outcome distribution is $N\left(\mu^* + R(h,\beta^*), h^{-1}\right)$, while the perceived outcome distribution is $N\left(\mu + R(h,\beta), h^{-1}\right)$ for each $\omega \in \Omega$.  
The KL minimization problem is
\begin{equation}\label{eqn_dkl}
\min_{\omega \in \Omega} \;
\frac{1}{2h^{-1}}
\left[
   \hat{\mu} + R(h,\beta) 
   - \bigl(\mu^* + R(h,\beta^*)\bigr)
\right]^{2},
\end{equation}
where the objective is the KL divergence. 

Under our assumption that S holds dogmatic beliefs about A's mean ability,  $\hat{\mu}$ is fixed. The problem thus reduces to choosing $\beta \in [\ul\beta,\bar\beta]$ to minimize the objective. 
We denote the solution by $\tilde{\beta}(h)$ to emphasize its dependence on $h$.

\begin{definition}
A \emph{Berk--Nash equilibrium (Berk--NE)} is a triple $\bigl(a(\cdot,\cdot), \hat{h}, \hat{\beta}\bigr)$ 
such that: 
\begin{itemize}
    \item For any $h$ and $\beta$, $a(h,\beta)$ is the effort choice. The true effort is $a(h,\beta^*)$. 
    
    \item $\hat{h} = h(\hat{\beta})$: assessment is optimal given E's belief about the fundamental 
    and the implied belief about A's effort choice.
    
    \item $\hat{\beta} = \tilde{\beta}(\hat{h})$: belief results from Bayesian updating on the outcome 
    distribution induced by $\hat{h}$ and $a(\hat{h},\beta^*)$.\footnote{Throughout, we use $\hat{\cdot}$ to indicate equilibrium objects, and $\tilde{\cdot}$ to denote the projection of an object to $[\ul\beta,\bar\beta]$, i.e.,  $\tilde{x}=x$ if $x \in [\ul \beta,\bar\beta]$, $\bar\beta$ if $x>\bar\beta$, and $\ul\beta$ if $x<\ul\beta$. We refer to $\hat{\beta}$ as an equilibrium belief about productivity, although strictly speaking the belief is the degenerate distribution $\delta_{\hat{\beta}}$. }
\end{itemize}

If KL divergence equals zero, so that $\hat{\beta}$ perfectly fits the equilibrium outcome distribution, then the equilibrium is a \emph{self-confirming equilibrium (SCE)}.\footnote{The concept of self-confirming equilibrium originated in \cite{fudenberg1993self} and \cite{battigalli1}. The emergence of such equilibria as long-run outcomes of misspecified learning with regime shifts has recently been studied by \cite{ba2021robust} and \cite{lanzani2025dynamic}. Similar consistency requirements arise in other solution concepts for games with bounded rationality, such as the analogy-based expectation equilibrium \citep{jehiel2005analogy} and the personal equilibrium of misspecified causal models \citep{spiegler2020behavioral}.}
\end{definition}

\subsection{Examples}\label{sec_example}
We phrase the two examples from the introduction as special cases of our model.
\vspace{-10pt}
\paragraph{Education.} 
The evaluator is a school, and the agent is a student. The student's school performance is $X=\eta+r(a,\beta)+\varepsilon$, where $\eta$ is innate ability,  $r(a,\beta)$ is study output parameterized by study aptitude $\beta$, and $\varepsilon$ is noise whose variance is reduced by the school's assessment.  

The market rewards ability and  a fraction \(\delta \in [0,1]\) of study output,  i.e.
$
v_{\mathrm{M}}(a,\beta) = \delta\, r(a,\beta).
$ 
A higher \(\delta\) may result from reforms that better align school curricula with market needs. The school values student performance (e.g., better performance enhances school reputation and resources) and possibly student well-being, with weights \(\lambda_{\mathrm{E}}>0\) and \(\lambda_{\mathrm{A}}\geq 0\).  Its payoff can be expressed as 
\[
v_{\mathrm{E}}(a,\beta) \;=\; \lambda_1 \, r(a,\beta) - \lambda_2 \, c(a)
\]
\[
\text{where} \quad 
\lambda_1 = \lambda_{\mathrm{E}} + \delta \lambda_{\mathrm{A}}>0, 
\qquad 
\lambda_2 = \lambda_{\mathrm{A}} \geq 0,
\]
up to terms independent of effort or assessment.\footnote{Given $X$, the market rewards the student by 
\[\mathbb{E}_{\delta_{(\omega, \hat{a})}}\left[\eta + \delta r(\hat{a},\beta)\mid X\right]=\hat{\mu}+h \left[X-\hat{\mu}-r(\hat{a},\beta)\right]+\delta r(\hat{a},\beta).\]
From society's perspective, $X$ has expected value $\hat{\mu}+r\left(\hat{a},\beta\right)$; hence the second term on the right-hand side has a zero mean. It follows that society perceives the the expected market reward as $\hat{\mu}+\delta r(\hat{a},\beta)$ and the student's expected payoff as $\hat{\mu}+\delta r(\hat{a},\beta)-c(\hat{a})$. Taking the weighted sum of these considerata yields the above expression.  }  

The assumption that the student knows $\omega^*$ reflects familiarity with their own talent and study approach. The assumption that $\hat{\mu}\neq \mu^*$ reflects the educational stereotypes outlined in the Introduction, which are stubborn and resistant to corrections. They influence how society learns about $\beta$, shaping and being shaped by the strategic choices of assessment and effort. 

\vspace{-10pt}

\paragraph{Labor markets.} The evaluator is the employer, and the agent is a worker. The outcome $X = \eta + r(a,\beta) + \varepsilon$ represents the worker's performance on the job, where $r(a,\beta)$ denotes the contribution of effort, parameterized by work aptitude~$\beta$. 

The market values the worker by $\eta + \delta r(a,\beta)$: $\delta = 0$ corresponds to situations where effort generates only employer-specific human capital; $\delta > 0$ if effort produces general human capital valued by the broader market. The employer values worker performance because it generates revenues, and he may also value worker  well-being; the weights are $\lambda_{\mathrm{E}}>0$ and $\lambda_{\mathrm{A}} \geq 0$. 

Sources of misspecifications of worker ability---such as gender and racial stereotypes, attribution errors, and stereotype incongruity---are extensively documented in the empirical literature and briely mentioned in the Introduction. 
\vspace{-10pt}

\paragraph{Linear-quadratic case. } When $r(a,\beta)=\beta a$, $c(a)=ca^2/2$, and $\kappa(h)=\kappa h^2/2$, solving players' optimal behavior yields
\[
a(h,\beta) = \frac{h\beta}{c}, \quad R(h,\beta) = \frac{h\beta^2}{c}, \quad h(\beta) = \frac{\lambda_1 \beta^2}{\lambda_2 \beta^2 + \kappa c},
\]
where $\lambda_1=\lambda_{\text{E}}+\delta \lambda_{\mathrm{A}}$ and $\lambda_2=\lambda_{\mathrm{A}}$. For comparative statics, let $\{\lambda_{\mathrm{E}},\,\delta,\,-c,\,-\kappa\}$ be the relevant parameter set, with a typical element denoted by $\zeta$.

\section{Results}\label{sec_result}

\subsection{Preliminaries}\label{sec_preliminary}
This section introduces preliminary concepts and assumptions for later analysis.

To begin, recall the  effort, effective effort, and assessment maps: 
\[
\begin{array}{rcl}
a : [0,1]\times [\ul\beta,\bar\beta] &\to& \mathbb{R}_+, \quad (h,\beta) \mapsto a(h,\beta), \\[4pt]
R : [0,1]\times [\ul\beta,\bar\beta] &\to& \mathbb{R}_+, \quad (h,\beta) \mapsto R(h,\beta), \\[4pt]
h : [\ul\beta,\bar\beta] &\to& [0,1], \quad \beta \mapsto h(\beta).
\end{array}
\]
These are assumed to be smooth functions. For analytical convenience, we extend their $\beta$-argument to $\mathbb{R}_+$ and impose the following conditions:\footnote{This extension is only used to establish Theorem \ref{thm_direction} Part (iii) Case B.}

\begin{assumption}\label{assm_regularity1}
\begin{enumerate}[(i)] 
\item $a(\cdot,\cdot)$ satisfies \(a(0,\beta) \equiv a(h,0) \equiv 0\),  is strictly increasing in \( h \) for all \( \beta > 0 \), and is strictly increasing in \( \beta \) for all \( h \in (0,1] \).

\item $R(\cdot,\cdot)$ satisfies the properties assumed of $a(\cdot,\cdot)$ in Part (i), and, additionally, strict increasing differences in $(h,\beta)$. 
\item $h(\cdot)$ satisfies $h(0)=0$ and $h(\beta) \in (0,1)$ for all $\beta>0$, and is strictly increasing in $\beta$. 
\end{enumerate}
\end{assumption}

Part (i) of Assumption \ref{assm_regularity1} states that the optimal effort is zero when the agent is either unproductive or unassessed. Otherwise it strictly increases with productivity and assessment. Part (ii) states that the marginal contribution of assessment to the optimal effective effort increases with the agent's productivity. Part (iii) states that unproductive agents should not be assessed, whereas productive agents should be assessed at a positive, interior level. As productivity rises,  assessment should increase to strengthen incentives. 

 Appendix \ref{sec_assumption1} provides sufficient conditions on primitives that ensure the validity of Assumption \ref{assm_regularity1}. It is clear that the assumption holds in the linear-quadratic example.

\subsection{Equilibria and stability}

We develop a fixed-point representation of Berk--NE, beginning by rewriting the KL minimization problem as
\begin{equation}\label{eqn_minkl}
\min_{\beta\in[\ul\beta,\bar\beta]}
\; \bigl|\Delta_{\mu}+R(h,\beta)-R(h,\beta^*)\bigr|,
\qquad
\text{where }\;\Delta_{\mu}\coloneqq\hat{\mu}-\mu^*,
\end{equation}
In the objective, $\Delta_{\mu}$ is the misspecification of ability, while the gap
$R(h,\beta)-R(h,\beta^*)=r\left(a\left(h,\beta\right),\beta\right)-r\left(a\left(h,\beta^*\right), \beta^*\right)$
captures misbeliefs about effective effort. Misbeliefs about productivity ($\beta\neq\beta^*$) enters this gap directly through the fundamental channel and indirectly through perceived effort.

Since \( R(h, \beta) \) is strictly increasing in $\beta$ for all $h>0$, \eqref{eqn_minkl} admits a unique solution, $\tilde{\beta}(h)$, for any such $h$. In particular, uniqueness holds over the range $h([\ul\beta,\bar\beta])$, as Assumption 1 
(iii)  and $\ul\beta>0$ together imply that $h(\ul\beta)>0$.\footnote{Since $R(0,\beta)\equiv 0$,  \eqref{eqn_minkl} admits the entire interval $[\ul\beta, \bar\beta]$ as its solutions when $h=0$. The assumption $\ul\beta>0$ is made solely to rule out this under-identification scenario.   \label{fn_lowerbound}} 

Based on $h(\cdot)$ and $\tilde{\beta}(\cdot)$, define a function 
\[
\tilde{\psi}: [\ul\beta,\bar\beta] \to [\ul\beta,\bar\beta], 
\quad \beta \mapsto (\tilde{\beta} \circ h)(\beta)=\tilde{\psi}(\beta).
\]
For each $\beta$, $\tilde{\psi}(\beta)$ is the KL-minimizing parameter for the outcome distribution induced by the optimal assessment at $\beta$.  The fixed points of $\tilde{\psi}$---whose existence follows from the Brouwer fixed point theorem---fully characterize Berk--NE beliefs about  productivity. If a fixed point $\hat{\beta}$ lies in $(\ul\beta,\bar\beta)$, it corresponds to an SCE. This is because $\hat{\beta}$ already solves the KL minimization problem without the constraint $\beta\in[\ul\beta,\bar\beta]$, and so it must perfectly fit the outcome distribution induced by $h(\hat{\beta})$.

In the linear-quadratic example, $\tilde{\psi}$ has a closed-form solution: 
\[\tilde{\psi}(\beta)=\sqrt{\max\left\{(\beta^*)^2-\dfrac{\Delta_{\mu} c}{h(\beta)}, (\ul\beta)^2\right\}}.\]
Panel (a) of Figure \ref{fig:combined} plots $\tilde{\psi}$ and its fixed points for a small, positive $\Delta_{\mu}$. Multiple fixed points may arise, corresponding to distinct Berk--NE

We deal with equilibrium multiplicity by using out-of-equilibrium dynamics to select stable equilibria. Analysis of such dynamics is deferred to Appendix \ref{sec_learning}. Here, we introduce the relevant notion of stability and briefly discuss its motivation.

To begin, we make the following assumption about $\tilde{\psi}$:
\begin{assumption}\label{assm_regularity2}
$\tilde{\psi}$ strictly crosses the 45 degree line at its fixed points. 
\end{assumption}
Assumption~\ref{assm_regularity2} essentially says that all roots of the function $\tilde{\psi}(\beta)-\beta$ are regular. It can be ensured to hold generically by an application of Sard's lemma (see Appendix~\ref{sec_assumption2}).

Under Assumption \ref{assm_regularity2}, we can classify Berk--NE into two categories: 

\begin{definition}\label{defn_stability}
    A Berk--NE with belief $\hat{\beta}$ is \emph{locally asymptotically stable in the Lyapunov sense} if $\tilde{\psi}$ crosses the 45-degree line strictly from above at $\hat{\beta}$. It is \emph{unstable} if crossing occurs strictly from below. 
\end{definition}

The concepts of stability and instability take on their natural meaning under a dynamic model of out-of-equilibrium learning under misspecification.  Appendix~\ref{sec_learning} analyzes such a model over an infinite horizon, in which society begins with a misspecified prior and updates beliefs via Bayes' rule based on observed outcomes. Unlike in the steady-state model presented in the main text, there assessment and effort evolve accordingly and, in turn, feed back into belief updating. These dynamics---occurring away from the steady state and featuring strong dependence between beliefs and strategic behavior---render the standard law of large numbers inapplicable.
Nevertheless, we show that society's posterior beliefs almost surely converge to the set of the locally asymptotically stable equilibrium values of $\hat{\beta}$, and to every such $\hat{\beta}$ with positive probability. In contrast, convergence to unstable equilibrium values of $\hat{\beta}$ occurs with zero probability.   We therefore adopt local asymptotic stability as the equilibrium selection criterion.

The core argument presented in Appendix \ref{sec_learning} is that, asymptotically, society's posterior beliefs become almost degenerate, and their evolution can be approximated by a deterministic ordinary differential equation (ODE), whose rate of change has the same sign as that of \(\tilde{\psi}(\beta) - \beta\). When $\tilde{\psi}(\beta)=\beta$, the ODE reaches a steady state, which corresponds to a Berk--NE belief. Starting from there, small perturbations trigger adjustment dynamics that restore beliefs to their original level if the equilibrium is locally asymptotically stable.  If the equilibrium is unstable, then small perturbations push beliefs further away. These are the adjustment dynamics that underpin our notion of stability. See Panel (a) of Figure \ref{fig:combined} for a graphical illustration.

\subsection{Main results}\label{sec_main}
We present two main results. The first result is presented as two separate statements, dealing with the cases $\Delta_{\mu}>0$ and $\Delta_{\mu}<0$. It addresses equilibrium uniqueness, stability, and the direction of distortions. The second result derives the comparative statics of the distortions associated with all locally asymptotically stable equilibria. 

\begin{theorem}[Part A]\label{thm_direction}  
Suppose Assumptions \ref{assm_regularity1} and \ref{assm_regularity2} hold. When society underestimates the agent's ability, i.e., $\Delta_{\mu}<0$: 
    \begin{enumerate}[(i)]
        \item The Berk--NE is unique and locally asymptotically stable.
        \item In this equilibrium, society's belief about effort productivity is above the truth: $\hat{\beta}>\beta^*$.
        \item   For $|\Delta_{\mu}|$ sufficiently small, the equilibrium is self-confirming, with $|\hat{\beta}-\beta^*|= O(\Delta_{\mu})$.\footnote{We use the Landau notation: $O(\cdot)$ means at most on the order of, $\Theta(\cdot)$ means exactly on the order of, and $o(\cdot)$ means slower than the order of. }
    \end{enumerate}  
\end{theorem}

\begin{theorem1}[Part B]
When society overestimates the agent's ability, i.e., $\Delta_{\mu}>0$: 
\begin{enumerate}[(i)]
    \item There is a finite and odd number of Berk--NE. The one with the highest belief of effort productivity is locally asymptotically stable, the second highest belief is unstable, and so on.
    \item In all equilibria, beliefs of effort productivity are below the truth: $\hat{\beta}<\beta^*$. 
    \item All equilibria are self-confirming, except possibly the one with the lowest belief of effort productivity. For any $\Delta_{\mu}>0$, there exists 
    $\underline{\underline{\beta}}>0$ such that whenever $\ul\beta \in (0, \ul{\ul{\beta}})$, the lowest equilibrium belief is $\ul\beta$, and the corresponding equilibrium is \emph{not} self-confirming. 
\end{enumerate}
\end{theorem1}

To state the second result, we parameterize the environment by $\zeta \in \mathbb{R}$ and write the optimal assessment as $h(\beta,\zeta)$. We assume that $h(\beta,\zeta)$ strictly increases with $\zeta$ for all $\beta>0$, so that shocks to the environment---as captured by $\zeta$---strictly increase the optimal assessment for all positive productivity values. In the linear–quadratic example, this holds for any $\zeta \in \{\lambda_{\mathrm{E}},\,\delta,\,-c,\,-\kappa\}$.

\begin{theorem}\label{thm_magnitude}  
Under Assumptions \ref{assm_regularity1} and \ref{assm_regularity2}, the distortions in the belief of effort productivity, \( |\hat{\beta} - \beta^*| \), in locally asymptotically stable equilibria increase in the weak set order under either of the following changes:
\begin{enumerate}[(i)]  
    \item A slight increase in the degree of ability misspecification \( |\Delta_{\mu}| \).  
    \item A slight shock to the environment, decrease in \( \zeta \), which lowers the optimal assessment \( h(\beta; \zeta) \) for all \( \beta>0 \).  
\end{enumerate}  
Furthermore, the above changes result in distinct sets of SCEs. 
\end{theorem}

\subsection{Proofs}\label{sec_analysis}
\paragraph{Proof of Theorem \ref{thm_direction}.} We begin by unpacking the definition of $\tilde{\psi}$. For each $\beta \in [\ul\beta,\bar\beta]$, the equation in $x$
\begin{equation}\label{eqn_perfectperception}
R(h(\beta), x)-R(h(\beta),\beta^*)=-\Delta_{\mu}
\end{equation}
admits at most one root in $\mathbb{R}_+$, because $R(h,\beta)$ is increasing in $\beta$ for all $h>0$. 
In this section,  we assume that a root exists and denote it by $\psi(\beta)$. In Appendix \ref{sec_op}, we complete the proof for the case in which a root may not exist. 

In words, $\psi(\beta)$ is the unconstrained best fit of the outcome distribution induced by the optimal assessment for $\beta$. Projecting  $\psi(\beta)$ onto $[\ul\beta,\bar\beta]$ yields the constrained best fit,  $\tilde{\psi}(\beta)$. A Berk--NE is self-confirming if the equilibrium belief $\hat{\beta}$ satisfies $\psi(\hat{\beta})=\hat{\beta}$.

The remainder of the proof proceeds by distinguishing two cases: $\Delta_{\mu}<0$ and $\Delta_{\mu}>0$. 

 \vspace{-10pt}

\subparagraph{Case A: $\Delta_{\mu}<0$.} In this case, the right-hand side of \eqref{eqn_perfectperception} is positive. Therefore, $\psi(\beta)>\beta^*$ always holds. As we raise $\beta$, $h(\beta)$ increases. Since $R(\cdot,\cdot)$ has strictly increasing differences, the left-hand side of \eqref{eqn_perfectperception} increases for all fixed $x>\beta^*$. To solve \eqref{eqn_perfectperception}, $\psi(\beta)$ must decrease with $\beta$. After truncation, $\tilde{\psi}(\beta)$ remains decreasing, and the decrease is strict whenever $\tilde{\psi}(\beta)<\bar\beta$.

As depicted in Panel (a) of Figure 
\ref{fig:twopanels}, the equilibrium is unique and locally asymptotically stable. For  $|\Delta_{\mu}|\approx 0$,  we have $|\hat{\beta}-\beta^*| = O(|\Delta_{\mu}|)$, so the equilibrium is self-confirming. 

\begin{figure}[htbp]
\centering

\begin{minipage}[t]{0.48\textwidth}
\centering
\begin{tikzpicture}[scale=1.1]
    
    \draw (1.6,1.6) rectangle (6,6);

    \node[above] at (4,5) {$\tilde{\psi}(\beta)$};

    \draw[thick, black] plot[domain=1.6:6, samples=1000] 
        (\x, {min(6, sqrt(25 + 0.6* (8*(\x^2 + 8))/(\x^2)))});

    \draw[dashed] (1.6,1.6) -- (6,6);

    \draw[dashed] (5,1.6) -- (5,6);
    \draw[dashed] (1.6,5) -- (6,5);

    \node[circle, fill=black, inner sep=1.5pt] at (5.572,5.572) {};

    \node[below] at (1.6,1.6) {$\underline{\beta}$};
    \node[below] at (5,1.6) {$\beta^*$};
    \node[below] at (6,1.6) {$\overline{\beta}$};

    \node[left] at (1.6,5) {$\beta^*$};
    \node[left] at (1.6,6) {$\overline{\beta}$};
\end{tikzpicture}
\caption*{Case A: $\Delta_{\mu}<0$}
\end{minipage}
\hfill
\begin{minipage}[t]{0.48\textwidth}
\centering
\begin{tikzpicture}[scale=1.1]
    
    \draw (1.6,1.6) rectangle (6,6);

    \node[above] at (3,3.7) {$\tilde{\psi}(\beta)$};

    \draw[thick, black] plot[domain=1.6:5.8, samples=1000] 
        (\x+0.2, {max(1.6, sqrt(25 - 0.75* (8*(\x^2 + 8))/(\x^2)))});    
    \draw[thick, black] plot[domain=1.6:1.907, samples=1000] 
        (\x, 1.6);

    \draw[dashed] (1.6,1.6) -- (6,6);

    \draw[dashed] (5,1.6) -- (5,6);
    \draw[dashed] (1.6,5) -- (6,5);

    \node[diamond, fill=black, inner sep=1.5pt] at (1.6,1.6) {};
    \node[circle, fill=black, inner sep=1.5pt] at (2,2) {};
    \node[circle, fill=black, inner sep=1.5pt] at (3.98,3.98) {};

    \node[below] at (1.6,1.6) {$\underline{\beta}$};
    \node[below] at (5,1.6) {$\beta^*$};
    \node[below] at (6,1.6) {$\overline{\beta}$};

    \node[left] at (1.6,5) {$\beta^*$};
    \node[left] at (1.6,6) {$\overline{\beta}$};

    \begin{scope}[shift={(6.5,4.5)}]
        \node[circle, fill=black, inner sep=1.5pt] at (0,1) {};
        \node[right] at (0.2,1) {SCE};
        \node[diamond, fill=black, inner sep=1.5pt] at (0,0.5) {};
        \node[right] at (0.2,0.5) {Berk--NE};
    \end{scope}
\end{tikzpicture}
\caption*{Case B: $\Delta_{\mu}>0$}
\end{minipage}
\caption{Equilibrium uniqueness, stability, and distortion.}
\label{fig:twopanels}
\end{figure}

\vspace{-10pt}
\subparagraph{Case B: $\Delta_{\mu}>0$. } In this case, the right-hand side of \eqref{eqn_perfectperception} is negative, so $\psi(\beta)<\beta^*$ always holds. Increasing $\beta$ raises $h(\beta)$, which in turn makes the left-hand side of \eqref{eqn_perfectperception} more negative for all fixed $x<\beta^*$. Solving  \eqref{eqn_perfectperception} thus requires that we increase $\psi(\beta)$, which, after truncation, remains increasing in $\beta$, with the decreasing being strict whenever $\tilde{\psi}(\beta) >\ul\beta$.

As depicted in Panel (b) of Figure \ref{fig:twopanels}, multiple fixed points may arise, among which the largest is locally asymptotically stable, the second largest is unstable, and so on. All fixed points---except possibly the smallest---lie in the interior of $[\underline{\beta}, \overline{\beta}]$, thereby corresponding to SCE beliefs.

When the smallest fixed point equals  $\ul\beta$, the corresponding equilibrium is not self-confirming. By Theorem \ref{thm_direction}, this occurs for any $\Delta_{\mu}>0$, provided that $\ul\beta$ is sufficiently close to zero. Recall the assumption that an unproductive agent shouldn't be assessed, i.e., $h(0)=0$, and that an unassessed agent exerts no (effective) effort, i.e., $R(0,\beta)=0$. Thus for any $\Delta_{\mu}>0$, there exists $\ul{\ul\beta} \in (0,\beta^*)$ such that 
\[R(h(\beta),\beta)-R(h(\beta),\beta^*)>-\Delta_{\mu} \quad  \forall \beta \in [0, \underline{\underline{\beta}}).\]
Now, pick any $\ul{\beta} \in (0,\ul{\ul{\beta}})$ and suppose that $\psi(\ul\beta)$ exists (again, we dispense with this assumption in Appendix \ref{sec_op}). Since $\psi(\ul\beta)$ solves \eqref{eqn_perfectperception}, it must satisfy $\psi(\ul\beta)<\ul\beta$, so $\tilde{\psi
}(\ul\beta)=\ul\beta$. This establishes that $\ul\beta$ is a Berk--NE belief but not an SCE belief.

In the argument above, a key assumption is that the assessment space is effectively unbounded below: recall that $h([\underline{\beta},\bar{\beta}]) \subset (0,1)$. As beliefs decline, assessment decreases accordingly. When the belief space is not overly constrained below, i.e., when $\underline{\beta}$ is chosen as above, this feedback loop can trigger a cascade that drives beliefs to the lower bound $\ul\beta$.

\vspace{-10pt}

\paragraph{Proof of Theorem \ref{thm_magnitude}.} We focus on the case when $\Delta_{\mu}>0$. The proof for the case when $\Delta_{\mu}<0$ is analogous. 

An increase in $\Delta_{\mu}$ makes the right-hand side of \eqref{eqn_perfectperception} more negative. The $x$ that solves the equation must then be lower, which means that the function $\psi(\cdot)$ is pointwise smaller. 

The same effect obtains when $\zeta$ decreases, which reduces $h(\beta;\zeta)$ for all $\beta>0$, Since $R(\cdot,\cdot)$ has strictly increasing differences, the $x$ that solves equation \eqref{eqn_perfectperception} is lower. Hence  $\psi(\cdot)$ is pointwise smaller. 

The resulting downward shift in $\psi$, and hence in $\tilde{\psi}$, is illustrated in Panel (b) of Figure~\ref{fig:combined}.

The effect of the pointwise change in $\tilde{\psi}$ on its stable equilibrium follows from the correspondence principle (see, for example, \citealt{echen02}). What is new here are our notion of stability and learning dynamics: we show that the correspondence principle can be applied to study the comparative statics of the Lyapunov locally asymptotically stable equilibria that arise as the long-run outcomes of misspecified Bayesian learning.

\section{Implications}\label{sec_implication}
This section examines the implications of our results for discrimination.

Consider two agents: \( m \) and \( w \), who are, in reality, identical and as described in Section \ref{sec_model}. However, society overestimates $m$'s innate ability and underestimates $w$'s. That is, $\Delta_{m}>0>\Delta_{w}$, where $\Delta_{i}$ denotes the misspecification of agent $i$'s ability.

We examine the equilibrium disparities in the \emph{objective} outcomes between groups. These outcomes include the perceived effort productivity, assessment, perceived and true efforts, expected market reward, and objective welfare. The last two objects are evaluated under the true equilibrium efforts and  outcome distributions.

The next two propositions follow immediately from Theorems \ref{thm_direction} and \ref{thm_magnitude}:

\begin{proposition}\label{prop_disparity}
The following statements are true in any SCE.
\begin{enumerate}[(i)]
\item Society underestimates $m$’s effort productivity and overestimates $w$’s:  
$\hat{\beta}_w > \beta^* > \hat{\beta}_m.$

\item Agent $w$ faces more intensive assessment than $m$: $h(\hat{\beta}_w) > h(\hat{\beta}_m).$

\item While $w$'s actual effort is greater than $m$'s, their perceived effort gap is even larger:  
\[
a(h(\hat{\beta}_w), \hat{\beta}_w)  
> a(h(\hat{\beta}_w), \beta^*)  
> a(h(\hat{\beta}_m), \beta^*)  
> a(h(\hat{\beta}_m), \hat{\beta}_m).
\]
\item The expected market rewards to group $i$ equals $\hat{\mu}_i+v_{\mathrm{M}}(a(h(\hat{\beta}_i), \hat{\beta}_i),\hat{\beta}_i).$ The difference between $m$ and $w$ is 
\[
\underbrace{\hat{\mu}_m - \hat{\mu}_w}_{>0}  
+ \underbrace{v_{\mathrm{M}}(a(h(\hat{\beta}_m), \hat{\beta}_m), \hat{\beta}_m)  
- v_{\mathrm{M}}(a(h(\hat{\beta}_w), \hat{\beta}_w), \hat{\beta}_w)}_{\leq 0}.
\]  
When $v_{\mathrm{M}} = 0$, $m$ out-earns $w$. If $v_{\mathrm{M}}$ is increasing in $a$ and $\beta$, then it is possible for $w$ to earn more than $m$, but only due to higher perceived productivity, more intensive assessment, and overestimated effort. 

\item The objective welfare of agent $i$ is $\hat{\mu}_i+v_{\mathrm{M}}(a(h(\hat{\beta}_i), \hat{\beta}_i),\hat{\beta}_i)-c(a(h(\hat{\beta}_i), \hat{\beta}_i)).$ The difference between $m$ and $w$ can be either positive or negative, for the same reason as in Part (iv). When $v_{\mathrm{M}}= 0$, the difference is positive and equals  
\[
\underbrace{\hat{\mu}_m-\hat{\mu}_w}_{>0}  
- \underbrace{\left[c(a(h(\hat{\beta}_m), \beta^*)) - c(a(h(\hat{\beta}_w), \beta^*))\right]}_{<0}.
\] 
\end{enumerate}
\end{proposition}

\begin{proposition}\label{prop_disparity2}
Disparities identified in Parts (i)–(iii) of Proposition \ref{prop_disparity}, and in Part (v) when $v_{\mathrm{M}}=0$, decrease in the weak set order among locally asymptotically stable SCEs under either of the following:
\begin{enumerate}[(i)]
    \item ability misspecifications $|\Delta_i|$, $i \in \{m,w\}$, decrease;
    \item $\zeta$ increases, so that assessment $h(\beta,\zeta)$ increases for all $\beta>0$. 
\end{enumerate}

\end{proposition}

Parts (i)-(iii) of Proposition \ref{prop_disparity} characterize the cross-domain reinforcement of stereotypes, showing how stereotypes of ability can bleed into the domain of effort, both in productivity and in choice. The misperception of effort choice is most novel: it arises from society's use of erroneous beliefs about effort productivity to infer effort, with both theoretical and practical implications.\footnote{Certainly, misbeliefs about effort choice feed back into the learning about effort productivity, influencing both the equilibria and how they are reached.  Section \ref{sec_discussion} compares equilibrium outcomes with and without this channel, and Appendix \ref{sec_learning} examines its implications for learning dynamics.}

Part (iv) of Proposition \ref{prop_disparity} highlights the welfare implications of distortions, emphasizing the role of the market's reward to effort, elaborated as follows: 
\[v_{\mathrm{M}}(a(\underbrace{h(\hat{\beta}_m}_{(2)}), \underbrace{\hat{\beta}_m}_{(3)}), \underbrace{\hat{\beta}_m}_{(1)})  
- v_{\mathrm{M}}(a(\underbrace{h(\hat{\beta}_w)}_{(3)}, \underbrace{\hat{\beta}_w}_{(3)}), \underbrace{\hat{\beta}_w}_{(1)}). 
\] 
When $v_{\mathrm{M}}=0$, this difference vanishes. When $v_{\mathrm{M}}(a,\beta)$ is increasing in $a$ and $\beta$, the difference becomes positive, and the change can be decomposed as follows:  
(1) holding effort fixed, effort is perceived as more productive for $w$ than for $m$, and this is valued by the market;
(2) due to inflated perceptions of effort productivity, $w$ is subject to more intensive assessment in reality; and  
(3) holding assessment fixed, $w$ is perceived to exert greater effort than $m$. 
Each of these distortions contributes to a more negative bracketed term, which when added to  $\Delta_{m}-\Delta_w$, narrows the overall reward gap and may even reverse it.  However, this reflects an entrenchment of ability-based stereotypes that spill over into the effort domain, rather than genuine equity. 

The objective welfare gap between agents is ambiguous for the same reason. When the market rewards only ability, $w$ earns less but works harder than $m$, achieving a lower welfare.

Proposition \ref{prop_disparity2} identifies policy levers that can either mitigate or exacerbate distortions and disparities, focusing on interventions that operate through assessment. In the linear–quadratic example, increasing any parameter in $\zeta \in \{\lambda_{\mathrm{E}}, \delta, -c, -\kappa\}$ raises $h(\cdot)$ pointwise, thereby mitigating the distortions in perceived effort productivity, assessment, and effort choice (both real and perceived). Intuitively, increasing $\lambda_{\mathrm{E}}$ makes the evaluator value effort more. Raising $\delta$ increases the market's reward to effort; in the case where the evaluator cares about the agent's payoff, this indirectly raises the evaluator's valuation of effort. Finally, lowering $c$ or $\kappa$ operates through the cost channels. All of these interventions serve to intensify assessment.

We now examine the implications of Propositions \ref{prop_disparity} and \ref{prop_disparity2} for education and labor market discrimination.

\vspace{-10pt}

\paragraph{Education.} Proposition~\ref{prop_disparity} informs gender- and race-stereotypes in education. 

In the context of gender, girls are seen as less brilliant than boys in STEM fields, where raw talent is emphasized for successes \citep{leslie2015expectations,bian2017gender}. Society---including parents and teachers---assumes that girls must compensate through diligence, attributing their academic successes to greater study aptitude \citep{wang2018systematic}. This misconception, known in educational psychology as the \emph{talent-effort trade-off}, has several consequences. Notably, it subjects girls to more intensive assessment: \cite{tiedemann2000gender} and \cite{copur2020teachers} find that teachers provide more detailed, critical feedback to mathematics work done by girls than by boys, flagging small errors that are often overlooked in boys' work. The resulting pressure contributes to socially prescribed perfectionism and burnout among girls \citep{rice2015perfectionism}, suggesting that girls work harder than boys in reality. 

A similar pattern appears in the context of race, where Asian students are viewed as hard-working but uncreative \citep{lu2024creativity}. Recent evidence by \cite{arcidiacono2022asian} reveals that Asian applicants receive lower personal ratings in college admissions that measure leadership potential and social skills. In response, Asian parents invest heavily in ``Ivy coaching,'' often spending considerable amounts on private counselors to polish their children's applications \citep{moscatello2024inventing}. Asian students, in turn, are pressured to ``out-white'' their white peers through extra hard work \citep{hsin2014explaining}.

The above findings are consistent with the predictions in Proposition \ref{prop_disparity} concerning perceived effort productivity, assessment, and real effort. Part (iii) of the proposition suggests an even more gloomy prospect: ultimately, society perceives girls and Asian students as working harder than they actually do. This reinforcement of the talent–effort trade-off does not arise because girls and Asian students face stronger scrutiny and work harder (which they do in reality). Rather, it stems from inflated beliefs about their effort productivity, which lead to incorrect inferences about their effort choices.

When effort is valued by employers or universities (e.g., through curriculum alignment), the above mechanisms are shown to narrow or even reverse the reward gap between minorities and majorities---consistent with the closing of gender gaps in STEM and the strong admissions outcomes of Asian students \citep{ceci2014women,lau2024aa}. Part (iv) of Proposition \ref{prop_disparity} warns that these trends should not be interpreted as definitive evidence that discrimination has disappeared, as they may reflect an entrenchment of ability-based stereotypes that spill over into the effort domain.

Proposition \ref{prop_disparity2} identifies distortion-mitigating policies that operate through intensifying assessment. In education contexts, increasing $\lambda_{\mathrm{E}}$ can be achieved through pay-for-performance schemes for teachers \citep{deming2022four}; increasing $\delta$ through curriculum reforms and alignment \citep{cullen2013can}; and reducing $\kappa$ through data-driven assessment innovations \citep{mandinach2021misconceptions}. By contrast, policies that raise the cost $c$ of studying or test preparation are bias-amplifying. These include China’s ban on after-school tutoring \citep{hall2024tutoring}, allowing multiple standardized-test retakes \citep{proulx2024should}, and admissions practices that favor affluent applicants \citep{chettydiversifying}. By disproportionately burdening less-affluent students, such policies widen the disparity between rich and poor.

\vspace{-10pt}
\paragraph{Labor markets.} Due to space constraints, we focus on gender discrimination in labor markets,\footnote{See, however, \cite{hyun2005bamboo} and \cite{lu2024creativity} for how the model minority stereotype hinders East Asians’ advancement to upper management positions, a phenomenon known as the \emph{bamboo ceiling}.}  and highlight insights that are not immediate from the education application.

Our prediction that women face heightened scrutiny is supported by empirical evidence regarding the academic labor market. In psychology, \cite{correll2017sws} find that professors evaluating fictitious candidates write four times more doubt-raising comments for women than men. For example: ``I would need to see evidence that she had gotten these grants and publications on her own'' and ``It is impossible to make such a judgment without teaching evaluations.'' In economics, \cite{hengel2022publishing} finds that female-authored papers take longer to pass peer review than comparable male-authored papers. \cite{eberhardt2023gender} find that recommendation letters for entry-level faculty positions are more likely to describe women using ``grindstone'' terms (e.g., hardworking, diligent) and make more frequent comparisons to their peers. \cite{dupas2021gender} document that female presenters receive more probing seminar questions. 

Regarding effort, our prediction that women exert more effort than men in reality is supported by Hengel's (2022) finding that female authors invest more in writing quality.  When effort produces general human capital valued by the broader market, this distortion may paradoxically improve women's compensation, causing a narrowing or even reversal of the gender pay gap. We thus caution that empirical signs of cracks in the glass ceiling---as documented in \cite{gayle2012gender} and \cite{bertrand2018coase, bertrand2020gender} among others---may reflect the entrenchment rather than a decline of stereotypes, with ambiguous welfare consequences.

Our comparative statics reveal additional policy levers to address labor market discrimination. Traditional approaches to combating stereotypes are studied largely in isolation by psychologists, and common solutions involve bias-correcting programs such as longitudinal training and reinforcement \citep{devine2012long}. These programs are widely implemented across sectors, yet recent meta-analyses suggest that they are costly to scale and underscore the need for more cost-effective alternatives \citep{greenwald2020implicit, bohnet2016works}. Our result on the cross-domain interaction of stereotypes reveals a cost-saving opportunity. As shown in Proposition \ref{prop_disparity2}, correcting misspecified beliefs about ability reduces effort distortions through the mechanism we propose, thereby ``killing two birds with one stone.''

 In labor contexts, intensifying assessment can be achieved through general human capital subsidies, which raise $\delta$; rater training, which lowers assessment costs $\kappa$ \citep{bohnet2016works}; and leveraging AI–human complementarities \citep{brynjolfsson2025generative}, which reduces effort costs $c$.

\section{Affirmative action}\label{sec_aa}
In this section, we extend the baseline framework and use our results to analyze the consequences of several affirmative action policies that are used in practice to address discrimination. 

To this end, we consider a population of agents belonging to different groups or collectives. Suppose there are \( J \geq 2 \) groups of agents, each with a population of \( \alpha_j > 0 \) such that \( \sum_{j=1}^J \alpha_j = 1 \). The true mean innate ability and effort productivity of group $j$ are $\mu_j^* \in \mathbb{R}$ and $\beta_j^* \in (\ul\beta,\bar\beta)$, respectively. These jointly determine the group's outcome $X_j=\mu^*_j+r(a_j,\beta^*_j)+\epsilon_j$, which is independently distributed across groups. 

Society misspecifies the mean ability of each group $j$ by \(\Delta_j \neq 0, \) and learns about their effort productivity based on observed outcomes, as in the baseline  model. For analytical interest, suppose that  both sets $\{j: \Delta_j<0\}$ and $\{j:\Delta_j>0\}$ are nonempty. The term $\bar\Delta \coloneqq \sum_{j=1}^J \alpha_j\Delta_j$ represents the average ability misspecification across groups.

We examine two types of affirmative action policies: \emph{color-blind assessment} and \emph{color-sighted, impartial assessment}.

\vspace{-10pt}

\paragraph{Color-blind assessment.}\label{sec_aa_colorblind}
In this case, group identities are unobservable to the evaluator. Consequently, all agents are treated equally as average members of the whole population.

\begin{proposition}\label{prop_colorblind}
The set of Berk--NE under color-blind assessment coincides with that in Section \ref{sec_result} when the agent's mean ability is misspecified by \( \bar{\Delta} \).  
\end{proposition}

\vspace{-15pt}
\paragraph{Color-sighted impartial assessment.} Under the color-sighted policy, group identities are observable, but all groups must be assessed at the same intensity. Society forms distinct beliefs about each group's productivity, shaping and being shaped by the strategic choices of assessment and efforts.   As we shall see, analyzing this multi-dimensional problem requires new methodological ideas. It is not a straightforward application of the analysis thus far developed for a single-dimensional belief space.

\vspace{-10pt}

\subparagraph{Preliminaries.} We first update the equations defining a Berk--NE. To begin, observe that the problem facing any individual agent remains the same as in the baseline model. The problem of the evaluator is different. Given any vector \(\bm{\beta} \coloneqq (\beta_j)_{j=1}^J\) of effort productivities, the optimal assessment solves
\begin{equation}\label{eqn_eproblem_aa}
    \max_{h \in [0,1]} \sum_{j=1}^J \alpha_i v_{\mathrm{E}}\left(a\left(h,\beta_j\right), \beta_j\right) - \kappa(h).
\end{equation}  
Under the same conditions underpinning Assumption \ref{assm_regularity1}, this problem admits a unique solution, denoted by \(h(\bm{\beta})\).

Given assessment intensity $h$, the perceived outcome distribution is Gaussian, with a mean vector \((\hat{\mu}_j + R(h, \beta_j))_{j=1}^J\) and a variance-covariance matrix \(h^{-1} Id_{J}\). The true outcome distribution has the same variance-covariance matrix but a different mean vector: \((\mu_j^* + R(h, \beta_j^*))_{j=1}^J\). The KL minimization problem is 
\begin{equation}\label{eqn_dkl_aa}
   \min_{\bm{\beta} \in [\ul\beta,\bar\beta]^{J}} \frac{1}{2h^{-1}} \sum_{j=1}^J \left(\Delta_j + R(h, \beta_j) - R(h, \beta_j^*)\right)^2, 
\end{equation} 
where the objective is the KL divergence. Since the objective is separable across $j=1,\cdots, J$, the KL minimizer is given by $\tilde{\bm{\beta}}(h) \coloneqq (\tilde{\beta}_j(h))_{j=1}^J$. 

Berk–Nash equilibrium beliefs regarding effort productivity correspond to the fixed points of 
\[\tilde{\bm{\psi}}: [\ul\beta,\bar\beta]^J \rightarrow [\ul\beta,\bar\beta]^J, \quad \bm{\beta} \mapsto (\tilde{\bm{\beta}} \circ h)(\bm\beta)\coloneqq \tilde{\bm{\psi}}(\bm{\beta}).\] 
Fixed points in the interior of $[\ul\beta,\bar\beta]^J$ correspond to SCE beliefs. 
A Berk–Nash equilibrium with belief vector $\hat{\bm{\beta}}$ is locally asymptotically stable in the Lyapunov sense if all eigenvalues of the Jacobian matrix of $\tilde{\bm{\psi}}(\bm{\beta})-\bm{\beta}$, evaluated at $\hat{\bm{\beta}}$, have negative real parts. 

\vspace{-10pt}

\subparagraph{Results.}
We focus on the case where each group's ability is slightly misspecified to draw the sharpest distinction between the color-sighted and color-blind policies. 

\begin{theorem}\label{thm_colorsighted}
Suppose that Assumptions \ref{assm_primitives} and \ref{assm_primitives2} hold. Under color-sighted, impartial assessment, the following are true when $|\Delta_j|$, $j=1,\cdots, J$ are sufficiently small.
    \begin{enumerate}[(i)]
        \item There is a unique Berk--NE, which is self-confirming and locally asymptotically stable. 
        \item  In this equilibrium, $|\hat{\beta}_j-\beta_j^*| = O(|\Delta_j|)$, with $\hat{\beta}_j<\beta_j^*$ if $\Delta_j>0$ and $\hat{\beta}_j>\beta_j^*$ if $\Delta_j<0$.
        \item A small increase in $|\Delta_j|$ increases $
        |\hat{\beta}_j-\beta_j^*|$ but does not affect $|\hat{\beta}_k-\beta_k^*|$, $k \neq j$. 
        \item A small increase in $\zeta$ increases all $|\hat{\beta}_j-\beta_j^*|$, $j=1,\cdots, J$. 
    \end{enumerate}
\end{theorem}

Comparing Proposition \ref{prop_colorblind} and Theorem \ref{thm_colorsighted}, we can see that when the evaluator overestimates agents' abilities on average, i.e., $\bar\Delta>0$, color-blind assessment may generate multiple equilibria, some of which exhibit large distortions (as in Part (iii) of Theorem \ref{thm_direction}). In contrast, under color-sighted policies, the equilibrium must be unique and exhibit distortions on the order of $O(|\Delta_j|)$ for each group $j$. 

To illustrate this difference, consider the linear-quadratic example in Section \ref{sec_example}. Suppose two groups of agents, $m$ and $w$, with $\Delta_m>0>\Delta_w$. Under a color-blind assessment, the aggregate outcome exhibits a deficit that cannot be attributed to specific groups. Therefore, the evaluator lowers his belief about effort productivity for the average agent. This, in turn, lowers the assessment given by 
\[h(\beta)=\frac{\lambda_1 \beta^2}{\lambda_2 \beta^2+\kappa c}. \]
Since the right-hand side decreases without bound on $[0,1]$ as $\beta$ declines, a belief cascade arises when the belief space is not overly constrained (that is, when $\ul{\beta}$ is sufficiently close to zero). The resulting equilibria may exhibit large distortions in $\beta$ even when the misspecification of ability is arbitrarily small.

Under color-sighted, impartial assessment, the evaluator forms distinct beliefs for different groups, overestimating $w$'s effort productivity and underestimating $m$'s. The relevant domain for equilibrium beliefs is thus $[\ul\beta,\beta_m^*]\times [\beta_w^*, \overline{\beta}]$. Within this domain, solving the optimal assessment yields  
\[
h(\bm{\beta}) = \frac{\lambda_1 \sum_{j \in \{m,w\}} \alpha_j \beta_j^2}{\lambda_2 \sum_{j \in \{m,w\}} \alpha_i \beta_i^2 + \kappa c} 
> \ul{\ul h}\coloneqq \frac{\lambda_1  \alpha_w (\beta_w^*)^2}{\lambda_2  \alpha_w (\beta_w^*)^2 + \kappa c}.
\]  
Crucially, the lower bound $
\ul{\ul h}$ on the right-hand side depends only on $w$'s characteristics: their true productivity and population.   It arises from the presence of $w$ --- whose productivity must be inflated to explain the outcome data --- and the requirement that all groups be assessed at equal intensity.  Its universality prevents a belief cascade. 

Parts (iii) and (iv) analyze comparative statics of equilibrium distortions under color-sighted, impartial assessment. Results are qualitatively similar to those in the baseline model, although the proofs are more challenging; see Appendix~\ref{sec_op}.

In Appendix \ref{sec_learning_aa},  we examine the misspecified Bayesian learning foundation for the equilibrium under color-sighted, impartial assessment. We show that, near the equilibrium, the dynamic system governing the evolution of asymptotic posterior beliefs has eigenvalues that are all negative, hence the equilibrium is a sink. By Bena{\"{\i}}m's (1999) ``attractor theorem,'' posterior beliefs converge to the equilibrium with positive probability, and with arbitrarily high probability given a sufficiently concentrated prior.

\section{Discussion}\label{sec_discussion}
\paragraph{Strategic uncertainty.} In our model, society knows the agent's effort strategy $a(\cdot,\cdot)$  and perceives their effort choice as $a(h,\beta)$, where $\beta$ is (random) perceived productivity updated from realized outcomes and concentrates on $\hat{\beta}\neq \beta^*$ in equilibrium. To better understand the role of such strategic uncertainty,  we compare our framework with the \emph{first-order misspecification} and \emph{double-misspecification} scenarios studied in \citet{murooka2021misspecified, murooka2023higher}.
 
\vspace{5pt}
\emph{First-order misspecification.} In this case, players' first-order beliefs about the fundamental are common knowledge, although they may differ across players. Society accurately perceives the agent's effort choice as $a(h,\beta^*)$, while holding a belief $\beta$ about their  productivity.  The effective effort perceived by society is thus 
\[
   R^{\mathrm{F}}(h,\beta) \coloneqq r\left(a(h,\beta^*),\beta\right),
\]
and the KL  minimization problem becomes 
\[
\min_{\beta \in [\ul\beta,\bar\beta]} |\Delta_{\mu} + R^{\mathrm{F}}(h,\beta)-R^{\mathrm{F}}(h,\beta^*)|. 
\]

One can verify that all results from Section \ref{sec_result} remain qualitatively valid under first-order misspecification (details are available upon request). Given this, we compare the magnitudes of belief distortions, focusing on small misspecifications of ability. 

The next proposition shows that, perhaps surprisingly,  relative to our model, first-order misspecification amplifies the belief distortion in the least distorted SCE, one with the lowest  $\hat{\beta}$ when $\Delta_{\mu}<0$ and the highest $\hat{\beta}$ when $\Delta_{\mu}>0$:

\begin{proposition}\label{prop_decomposition}  
When \( |\Delta_{\mu}| \) is sufficiently small, holding all else constant, the gap \( |\hat{\beta} - \beta^*| \) in the least distorted SCE is smaller in our model than under first-order misspecification. 
\end{proposition}  

To develop intuitions, we decompose \( R(h,\beta) - R(h,\beta^*) \) as
\[
R(h,\beta) - R(h,\beta^*) = \underbrace{r\left(a(h,\beta),\beta\right) - r\left(a(h,\beta^*),\beta\right)}_{\text{(I)}} + \underbrace{r\left(a(h,\beta^*),\beta\right) - r\left(a(h,\beta^*),\beta^*\right)}_{\text{(II)} = R^{\mathrm{F}}(h,\beta)-R^{\mathrm{F}}(h,\beta^*)}. 
\]
On the right-hand side, term (I) captures the effect of misperceiving effort choice alone, while term (II) captures the effect of misperceiving effort productivity alone. The total effect is on the left-hand side.

Consider the case \(\Delta_{\mu} > 0\), in which equilibrium belief of effort productivity must belong to $[\ul\beta, \beta^*)\). Since (I) is negative over this domain, adding it to (II) lowers the left-hand side of the SCE-defining equation: $R^{\mathrm{F}}(h(\beta), x)-R^{\mathrm{F}}(h(\beta), \beta^*)=-\Delta_{\mu}$. To restore balance, society must increase $x$, thereby mitigating the distortion in perceived effort productivity. The argument for \(\Delta_{\mu} < 0\) is analogous.

The above reasoning fixes assessment across scenarios. In Appendix \ref{sec_proof}, we show that allowing assessment to vary has a second-order  effect on the least distorted SCE when ability misspecification is small, for reasons akin to the envelope theorem.

The implications of Proposition \ref{prop_decomposition} for other equilibrium outcomes are ambiguous. Consider, for instance, assessment intensity. Let $\delta_x$ denote society's first-order belief about effort productivity, and $\delta_{\delta_y}$ denote its second-order belief about the agent's first-order belief. Given these beliefs, optimal assessment $h(x,y)$ solves
\[\max_{h \in [0,1]} v_{\mathrm{E}}(a(h,y),x)-\kappa(h). \]
In our model, \( y = x \), hence the optimal assessment is $h(x,x)$. Under first-order misspecification, \( y = \beta^* \), hence the optimal assessment is $h(x,\beta^*)$. While \( x \) has a greater influence on \( h \) in our case, it is less also distorted by Proposition \ref{prop_decomposition}. The overall effect is  ambiguous. The same ambiguity applies to perceived and actual efforts, and to the agent’s market reward and objective welfare.

\vspace{5pt}

\emph{Double misspecification.} In this case, society mistakenly believes that the agent shares its first-order belief about the fundamental. While this change leaves the choices of effort and assessment unaffected relative to our model, it alters the KL minimization problem: given an initial belief \(\hat{\beta}\) about effort productivity, society continues to believe that the effort choice is \(a(h(\hat{\beta}), \hat{\beta})\), when updating its belief about the fundamental based on the observed outcomes. The KL minimization problem becomes  
\[
\min_{x \in [\ul\beta, \bar\beta]} \left| \Delta_{\mu} + \underbrace{r(a(h(\hat{\beta}), \hat{\beta}), x)}_{\text{equals } r(a(h(\hat{\beta}), x), x) \text{ in our model}} - r(a(h(\hat{\beta}), \beta^*), \beta^*) \right|.
\]

Despite this difference, double misspecification generates the same set of Berk--NE as our model, since \(x = \hat{\beta}\) in equilibrium. However, the learning foundations of these equilibria differ --- a distinction we explore in Appendix \ref{sec_learning}.

 \vspace{-10pt}
\paragraph{Heterogeneous priors.}  So far, we have assumed that E and M share the same prior, in order to capture misbeliefs broadly held at the societal level. Relaxing this assumption doesn't fundamentally change the analysis, provided that the problem of incentive provision is separated from the rest. 

 As already shown,  A's incentives depend solely on the assessment intensity they face; recall \eqref{eqn_aproblem}. Given $h$, effort strategy is $a(h,\cdot)$, and this is common knowledge. In the case where society members hold different beliefs about $\beta$, their inferences about A's effort choice differ. Now, suppose that E doesn't care about M's belief (in the linear-quadratic example when $\lambda_{\mathrm{A}}=0$). Then E's  problem remains as in \eqref{eqn_eproblem}. Let $\hat{\beta}_{\mathrm{E}}$ denote an equilibrium belief held by E. Equilibrium assessment (publicly observed) is then $h(\hat{\beta}_{\mathrm{E}})$, and equilibrium effort strategy $a(h(\hat{\beta}_{\mathrm{E}}),\cdot)$. For each party $i \in \{\mathrm{E}, \mathrm{M}\}$, their equilibrium belief $\hat{\beta}_i$ solves the  KL minimization problem: 
 \[
 \min_{\beta\in[\ul\beta,\bar\beta]}\; \bigl|\Delta_{i}+R(h(\hat{\beta}_{\mathrm{E}}),\beta)-R(h(\hat{\beta}_{\mathrm{E}}),\beta^*)\bigr|, \quad\text{where }\Delta_{i}\coloneqq\hat{\mu}_i-\mu^*.\]
 M's belief affects A's reward but not their incentives. To evaluate its impact on A, substitute $\hat{\beta}_{\mathrm{M}}$ for the M's belief about $\beta$ in Proposition \ref{prop_disparity}. 

 Turning to the learning dynamics, E's beliefs follow the exact same trajectories as before, whereas M's beliefs are passively updated from its prior based on the public history (of assessments and outcomes). The stability of an equilibrium is determined solely by that of E's  belief.

\vspace{-10pt} 

\paragraph{Explicit incentives.} 
Our focus on implicit incentive provision is motivated by applications. That said, the analysis extends naturally to settings in which the evaluator directly offers linear incentive contracts of the form $hX+w$, where $h$ represents the strength of explicit incentives and $w$ the base wage. In such environments, $\mu$ represents factors such as luck, teammate influence, spillover across tasks, etc., that cannot be directly controlled by the agent. Misspecified  learning is shaped by central considerations in contract theory such as the agent's risk attitude. The equilibrium consequences in low- and high-dimensional environments (with multiple contractible signals, agents, or tasks) and their learning foundations are subjects of future research. 

\vspace{-10pt}
\paragraph{Further application.} Our model has applications beyond discrimination. In Appendix \ref{sec_oversight}, we study a setting in which the evaluator is an oversight authority, the agent is a politician, and the market consists of voters. Society’s beliefs about the politician’s ability are distorted by partisanship, and learning about effort productivity interacts with voters’ forward-looking electoral behavior. The framework captures a range of phenomena in selective oversight and identifies new levers for promoting equality and efficiency.

\section{Related literature}\label{sec_literature}

\paragraph{Misspecified Bayesian learning.} 
Misspecified Bayesian learning has long been studied in statistical decision theory \citep{berk1966limiting}, with recent economic research emphasizing its interaction with strategic behavior; the seminal paper is  \cite{esponda2016berk}. See \cite{bohren2024misspecified} for a recent survey. We make several contributions to this literature.

First, we are the first to study misspecified learning in principal–agent relationships. The literature on behavioral incentive contracting, surveyed by \cite{kHoszegi2014behavioral}, typically examines how profit-maximizing principals exploit agents’ exogenously given biases, with no learning on either side. Misspecified learning has only recently been considered by \cite{niu2023motivated}, where a principal determines an agent’s overconfidence level once and for all, after which the agent acts alone as in \cite{heidhues2018unrealistic}. In our framework, it is the principal (the evaluator and the market), not the agent, who holds misspecified beliefs about the agent's innate ability. Misspecified learning about effort productivity then interacts with strategic incentive provision over time.

Second, we contribute to the comparative statics of Berk–Nash equilibria. We focus on Lyapunov locally asymptotically stable equilibria and use the implied out-of-equilibrium dynamics to obtain robust comparative statics for all stable equilibria. Conceptually, this is an application of the correspondence principle \citep{echen02,echenique2006correspondence}, although our stability notion and learning dynamics are new. Recently, \citet{ghosh2024robust} derive robust comparative statics for Berk–Nash equilibria in general Markov decision problems, but the mechanisms are unrelated to ours.  \cite{heidhues2018unrealistic} study how over- and under-confidence affect single-agent decision making with misspecified Bayesian learning. We draw a similar distinction, but unlike \cite{heidhues2018unrealistic}, we do not assume a unique equilibrium and instead examine the robust comparative statics of stable equilibria. 

Third, we analyze a setting with strategic uncertainty about the agent’s effort choice. Strategic uncertainty is, by definition, absent from misspecified-learning models with a single decision maker \citep{heidhues2018unrealistic, heidhues2021convergence, esponda2021asymptotic, fudenberg2021limit}; from multi-player models with publicly observable behavior \citep{ba2023multi}, and from \citet{murooka2021misspecified, murooka2023higher}, where players—whether correctly or incorrectly—believe that they know their opponents’ posterior beliefs about the fundamental and can therefore fully infer their behavior from commonly known  equilibrium strategies. By contrast, social learning models such as \citet{frick2020misinterpreting} and \citet{bohren2021learning} stress misbeliefs about opponents’ strategies as the main obstacle to information aggregation, while assuming correct beliefs about the fundamental. In our case, the agent’s equilibrium effort strategy is common knowledge, yet misbeliefs of the fundamental still lead to inaccurate inferences about the agent's behavior.

Fourth, we provide learning foundations for Lyapunov locally asymptotically stable equilibria, emphasizing the challenges posed by strategic uncertainty. Our analysis combines Bayesian tools—such as Laplace’s approximation—with stochastic approximation methods. The key mathematics references are \citet{benaim1999dynamics} and \citet{pemantle1990nonconvergence}. In economics, stochastic approximation has been used to study misspecified Bayesian learning in single-agent decision problems \citep{esponda2021asymptotic, heidhues2021convergence}, and in multi-player games without strategic uncertainty \citep{murooka2021misspecified, murooka2023higher}.\footnote{Beyond misspecified Bayesian learning, stochastic approximation has also been used to analyze long-run dynamics in stochastic fictitious play, evolutionary games, and social learning; see, e.g., \citet{benaim2006stochastic}, \cite{sandholm2011stochastic}, \cite{dasaratha2022learning}, and \cite{arieli2025hazards}.} The last three papers assume Gaussian outcomes that, without strategic uncertainty, imply (truncated) normal posteriors. We also use this parametric structure, but must additionally confront strategic uncertainty. Alternative approaches exist. \citet{fudenberg2021limit} establish convergence to uniformly strict Berk–Nash equilibria, a refinement that relies on finite actions and outcomes and does not narrow the equilibrium set in our setting with a continuum of actions, outcomes, and states. \citet{frick2023belief} also adopt Lyapunov asymptotic stability but assume finite states (or closely related structures), and their analyses are essentially conducted in single-agent environments.

\vspace{-10pt}
\paragraph{Implicit incentives.} The literature on career concerns, or implicit incentive provision through market beliefs is too vast to describe here. Our equilibrium analysis adopts the classical additive–normal structure pioneered by \cite{holmstrom1999managerial}, but collapses the environment to a single period following \citet{dewatripont1999economics,dewatripont1999economics2}. As already evident in our analysis, additivity allows us to capture the ability–effort trade-off central to the discrimination literature, while the Gaussian structure implies that the strength of implicit incentives is governed by the precision of the outcome as a signal of ability. Precision is chosen endogenously by the evaluator, yielding a simple parameterization of the incentive provision problem.

In the dynamic learning model presented in Appendix \ref{sec_learning}, the differences from existing career-concerns models are more pronounced. Traditional models assume that the agent’s ability is fixed over time and that the agent is forward-looking. Because of moral hazard, one-shot deviations from equilibrium effort may distort the market’s beliefs about ability, which in turn fuels further deviations by the agent. In the context of career concerns, such dynamic incentives have been studied by \cite{prat2014dynamic}, \cite{bonatti2017career}, and \cite{cisternas2018two}, among others. Such dynamic incentives are absent from our model, where the agent is myopic, and their innate ability is drawn independently each period. What is persistent is effort productivity, and learning interacts with the provision of implicit incentives by an evaluator who may be forward-looking or myopic. Misspecifications of ability yield persistent misbeliefs about effort productivity even in the long run.

\vspace{-10pt}
\paragraph{Discrimination.}  
The theoretical literature on discrimination has largely focused on statistical discrimination, exemplified by the Phelpsian \citep{aigner1977statistical} and Arrovian \citep{coateloury1993} traditions; see \cite{FANG2011133} for a survey. While recent work by \cite{echenique2025rationally} seeks to bridge these  via rationally inattentive assessment, both frameworks require agents to hold correct beliefs. Our analysis concerns a distinct source of discrimination: misspecifications hard-wired into agents' prior beliefs. This is  conceptually closer to what \cite{becker2010economics} calls taste-based discrimination, although it is belief-based and interacts with learning and incentive provision. The comparative statics, welfare predictions, and policy levers we derive cannot be captured by taste-based mechanisms alone.

Several recent papers examine discrimination that stem from inaccurate beliefs. \cite{bohrensystemic2025} study how discrimination in one domain spills into other domains via mechanisms that are distinct from ours. However, the kind of ``systemic'' discrimination documented in \cite{bohrensystemic2025} is conceptually related to the ability-effort trade-off studied in our model. Several recent studies reviewed below analyze discrimination models featuring passive misspecified learning. Unlike our active learning model, these decision makers do not adjust their behavior as they update their beliefs about the fundamental.

\cite{bohren2019dynamics} propose a method to disentangle discrimination based on erroneous beliefs from those of alternative forms. In \cite{bohren2019dynamics}, there is an evaluator of agents' performance, but no agency relationship or learning about agents' effort productivity. The interaction between an evaluator and an agent undertaking a costly effort, which is the core of our model, is absent. The authors establish, theoretically and experimentally, a reversal of discriminatory beliefs over time that is novel relative to traditional taste-based or statistical discrimination.

\cite{heidhues2025overconfidence} examine discriminatory outcomes that arise from in-group overconfidence about their own abilities. Their main findings concern how such overconfidence shapes beliefs about other traits of both in-groups and out-groups, emphasizing the interaction between learning and group identity.

\cite{he2023learning} study how misspecified learning about the data-generating process distorts predictions of future outcomes, with implications for discrimination, polarization, etc. Some of the methods used in Section \ref{sec_aa} of this paper parallel those in that work.

\appendix
\section{Proofs}\label{sec_proof}

\subsection{Sufficient conditions for Assumption \ref{assm_regularity1}}\label{sec_assumption1}
In the main text, we imposed assumptions on players' optimal behavior instead of on the primitive elements of the model. This was done for expediency. In this appendix, we spell out the assumptions on primitives that ensure Assumption \ref{assm_regularity1}. 

These primitives include the effort cost function, the 
assessment cost function, the effective effort function, and the evaluator’s valuation of effort:
\[
\begin{array}{rcl}
c : \mathbb{R}_+ &\to& \mathbb{R}, \quad a \mapsto c(a), \\[6pt]
\kappa : [0,1] &\to& \mathbb{R}, \quad h \mapsto \kappa(h), \\[6pt]
r : \mathbb{R}_+^2 &\to& \mathbb{R}, \quad (a,\beta) \mapsto r(a,\beta), \\[6pt]
v_{\mathrm{E}} : \mathbb{R}_+^2 &\to& \mathbb{R}, \quad (a,\beta) \mapsto v_{\mathrm{E}}(a,\beta).
\end{array}
\] 
In the last two cases, we extend the $\beta$-argument of the function to $\mathbb{R}_+$. 

\begin{assumption}\label{assm_primitives}
The above functions are all smooth:
\begin{enumerate}[(i)]
 \item $c(\cdot)$ is strictly increasing and strictly convex, with $c'(0)=0$ and $\lim_{a \to +\infty} c'(a) = +\infty$.

 \item $\kappa(\cdot)$ is strictly increasing and strictly convex, with $\kappa'(0)=0$ and $\lim_{h \to 1} \kappa'(h) = +\infty$.

 \item $r(\cdot,\cdot)$ is weakly concave in $a$ and has strictly increasing differences in $(a,\beta)$. Moreover, $r(0,\beta) \equiv r(a,0) \equiv 0$, $r_a(0,\beta) > 0$, and $\lim_{a \to +\infty} r_a(a,\beta)/c'(a) = 0$ for all $\beta>0$.

 \item $v_{\mathrm{E}}(\cdot,\cdot)$ is strictly increasing in both arguments, with $v_{\mathrm{E},a}(0,\beta) > 0$ for all $\beta>0$. It is also weakly concave in $a$ and has strictly increasing differences in $(a,\beta)$.
\end{enumerate}
\end{assumption}

\begin{proposition}\label{obs1}
Under Assumption \ref{assm_primitives}, the effort, effective effort, and assessment maps (with their $\beta$-argument extended to $\mathbb{R}_+$) are smooth functions: 
\begin{enumerate}[(i)]
 \item  $a(h,0)\equiv a(0,\beta)\equiv 0$, $a(\cdot,\cdot)$ is strictly increasing in $h$ for all $\beta>0$ and strictly increasing in $\beta$ for all $h\in(0,1]$.
 \item $R(\cdot,\cdot)$ satisfies the properties in Part (i).
 \item $h(0)=0$ and $h(\beta)\in(0,1)$ for all $\beta>0$. 
\end{enumerate}
\end{proposition}

\begin{proof}
Part (i): For given $(h,\beta)$, $a(h,\beta)$ solves 
\begin{equation}\label{eqn_aproblemprime}
\tag{1'}\max_{a \geq 0} \  hr(a,\beta)-c(a).
\end{equation}
Since $hr(a,\beta)=0$ when $h=0$ or $\beta=0$, and $c(a)$ is strictly increasing in $a$,  it follows that $a(0,\beta)= a(h,0)\equiv 0$. 

The objective is strictly concave in $a$, with derivative $h r_a(a,\beta)-c'(a)$ with respect to (w.r.t.) $a$. So the solution is unique. 
When $h, \beta>0$, $h r_a(0,\beta)-c'(0)>0$ and $\lim_{a \rightarrow +\infty} h r_a(a,\beta)-c'(a)<0$ hold by Assumption \ref{assm_primitives}. It follows that $a(h,\beta)>0$ and is determined by the first-order condition
\begin{equation}\label{eqn_afoc}
hr_a(a,\beta)=c'(a).
\end{equation}
Since $c''-h r_{aa}>0$, we can determine $a(\cdot,\cdot)$ from the implicit function theorem, which also implies that $a(\cdot,\cdot)$ is smooth with partial derivatives: 
\begin{equation}\label{eqn_aprime}
    \frac{\partial a}{\partial h}=\left.\frac{r_a}{c''-h r_{aa}}\right\vert_{a=a(h,\beta)} \quad \text{ and } \quad  \frac{\partial a}{\partial \beta}=\left.\frac{hr_{a\beta}}{c''-h r_{aa}}\right\vert_{a=a(h,\beta)}.
    \end{equation}
These are strictly positive by Assumption \ref{assm_primitives} and the fact that $a(h,\beta)>0$.

\vspace{10pt}
\noindent Part (ii): Recall that $R(h,\beta)=r\left(a\left(h,\beta\right),\beta\right)$, where  $r(0,\beta)\equiv 0$. Combining this with $a(0,\beta), a(h,0)\equiv 0$ from Part (i) gives $R(0,\beta), R(h,0)\equiv 0$. 

When $h,\beta>0$,  
\begin{equation}\label{eqn_rprime}
\frac{\partial R}{\partial h}=\left.\frac{\partial r}{\partial a}\right\vert_{a=a(h,\beta)}\frac{\partial a}{\partial h}>0,  \quad  \frac{\partial R}{\partial \beta}=\left.\frac{\partial r}{\partial a}\right\vert_{a=a(h,\beta)}\frac{\partial a}{\partial \beta}+\left.\frac{\partial^2 r}{\partial a\partial \beta}\right\vert_{a=a(h,\beta)}>0,
\end{equation}
where the inequalities follow again from Assumption \ref{assm_primitives} and $a(h,\beta)>0$. 

\vspace{10pt}

\noindent Part (iii): For given $\beta$, $h(\beta)$ solves
\[
\max_{h \in [0,1]} \; V_{\mathrm{E}}(h,\beta)-\kappa(h),
\quad 
\text{where } V_{\mathrm{E}}(h,\beta)=v_{\mathrm{E}}(a(h,\beta),\beta).
\]
Since $a(h,0)\equiv 0$ by Part~(i), we have $V_{\mathrm{E}}(h,0)=v_{\mathrm{E}}(0,0)$, which is independent of $h$, while $\kappa(\cdot)$ is strictly increasing in $h$. Consequently, $h(0)=0$.

The derivative of the objective w.r.t. $h$ is 
\[v_{\mathrm{E},a}\left(a\left(h,\beta\right),\beta\right)\frac{\partial a}{\partial h}-\kappa'(h).\]
Evaluating this expression at $h=0$ and as $h \uparrow 1$, holding any $\beta>0$ fixed, yields
\[v_{\mathrm{E},a}(0,\beta)\frac{r_a(0,\beta)}{c''(0)}-0>0 \ \text{ and } \  v_{\mathrm{E},a}(a,\beta)\left.\frac{r_a(a,\beta)}{c''(a)}\right\vert_{a=a(1,\beta)}-\infty=-\infty.\]
Consequently, $h(\beta) \in (0,1)$ for all $\beta>0$. 
\end{proof}

Next, we impose more nuanced assumptions to ensure that $R(\cdot,\cdot)$ has strict increasing differences and that $h(\cdot)$ is strictly increasing. All hold in the setting of Section \ref{sec_example}. 

\begin{assumption}\label{assm_primitives2}
\begin{enumerate}[(i)]
\item $r(a,\beta)=\beta a$ and $c''' \in \bigl[0, \tfrac{(c'')^2}{c'}\bigr]$.\footnote{The condition on $c$ is satisfied by all power functions with exponent at least two, as well as the exponential function.}
\item  $v_{\mathrm{E}}(a,\beta)=\lambda_1 r(a,\beta)-\lambda_2 c(a)$ with $\lambda_1>0$ and $\lambda_2 \geq 0$.
\end{enumerate}
\end{assumption}

\begin{proposition}\label{prop_a2}
Suppose that, in addition to Assumption \ref{assm_primitives}, Part (i) of Assumption \ref{assm_primitives2} holds. Then $a(h,\beta)$ has weak increasing differences in $(h,\beta)$ and is weakly concave in $h$. Consequently, 
\begin{enumerate}[(i)]
    \item $R(h,\beta)$ has strict increasing differences in $(h,\beta)$.
    \item $V_{\mathrm{E}}(h,\beta)$ is weakly concave in $h$, hence $h(\beta)$ is unique for all $\beta \geq 0$. 
\end{enumerate}
If Part (ii) of Assumption \ref{assm_primitives2} also holds, then $V_{\mathrm{E}}(h,\beta)$ has strict increasing differences in $(h,\beta)$, so $h(\beta)$ is strictly increasing in $\beta$.\footnote{As will become clear, the last assumption aligns the objectives between A and E, so the increasing differences of $a(h,\beta)$ in $(h,\beta)$ extend to $V_{\mathrm{E}}(h,\beta)$. Without this assumption, we would in general obtain
\begin{align*}
\frac{\partial^2 V_{\mathrm{E}}}{\partial h \partial \beta}& =\left.\left[\left(\frac{\partial^2 v_{\mathrm{E}}}{\partial a^2}\frac{\partial a}{\partial \beta}+\frac{\partial^2 v_{\mathrm{E}}}{\partial a \partial \beta}\right)\frac{\partial a}{\partial h}+\frac{\partial v_{\mathrm{E}}}{\partial a} \frac{\partial^2 a}{\partial h\partial \beta}\right]\right\vert_{a=a(h,\beta)}\\
&=\left.\left[\frac{\partial^2 v_{\mathrm{E}}}{\partial a^2}\frac{c'}{(c'')^2}+\frac{\partial^2 v_{\mathrm{E}}}{\partial a \partial \beta}\frac{\beta}{c''}+\frac{\partial v_{\mathrm{E}}}{\partial a} \left(\frac{1-\frac{c'''c'}{(c'')^2}}{c''}\right)\right] \right\vert_{a=a(h,\beta)},
\end{align*}
which is strictly positive if and only if 
\[\left.\dfrac{\partial^2 v_{\mathrm{E}}}{\partial a^2}\right\vert_{a=a(h,\beta)}>-\left[\dfrac{c''}{c'}\left(\left.\beta\dfrac{\partial^2 v_{\mathrm{E}}}{\partial a \partial \beta}+\dfrac{\partial v_{\mathrm{E}}}{\partial a} \left(1-\dfrac{c'''c'}{(c'')^2}\right)\right)\right] \right\vert_{a=a(h,\beta)}.\]
This holds so long as $v_{\mathrm{E}}$ is not overly concave in $a$.} 
\end{proposition}

\begin{proof}
When $r(a,\beta)=\beta a$, we have $a(h,\beta)=(c')^{-1}(h\beta)$, 
with partial derivatives
\[\frac{\partial a}{\partial h}=\frac{\beta}{c''}, \ \frac{\partial a}{\partial \beta}=\frac{h}{c''}, \ \frac{\partial^2 a}{\partial h^2}=-\frac{c'''}{c''}\left(\frac{\partial a}{\partial h}\right)^2, \ \frac{\partial^2 a}{\partial h \partial \beta}=\frac{1-\frac{c'''c'}{(c'')^2}}{c''}.\]
Among these, the second derivative is weakly negative if $c''' \geq 0$, whereas the cross-partial is weakly positive if $c''' \leq (c'')^2/c'$. When both conditions are met, $\partial^2 a/\partial h^2 \leq 0$ and $\partial a^2/\partial h\partial \beta \geq 0$ as desired.

\vspace{5pt}
\noindent Part (i): The expression for $\partial R/\partial h$ is given in \eqref{eqn_rprime}. Simplifying it  using A's first-order condition \eqref{eqn_afoc}, gives \[\frac{\partial R}{\partial h}=r_a(a(h,\beta),\beta)\frac{\partial a}{\partial h} = \frac{c'(a(h,\beta))}{h}\frac{\partial a}{\partial h}.\] Differentiating w.r.t. $\beta$ yields 
\[
\frac{\partial^2 R}{\partial h\partial \beta }=
\frac{1}{h}\left[c''(a(h,\beta))\frac{\partial a}{\partial \beta}\frac{\partial a}{\partial h}+c'(a(h,\beta))\frac{\partial^2 a}{\partial h\partial \beta} \right].\]
Since $\partial a/\partial h, \partial a/\partial \beta>0$ and  $\partial^2 a/\partial h\partial \beta \geq 0$, the right-hand side is strictly positive whenever $h, \beta>0$, so that $a(h,\beta)>0$ and hence $c''(a(h,\beta))>0$. 

\vspace{5pt}

\noindent Part (ii): Straightforward algebra shows that 
\[\frac{\partial^2 V_{\mathrm{E}}}{\partial h^2}=\left.\frac{\partial^2 v_{\mathrm{E}}}{\partial a^2}\right\vert_{a=a(h,\beta)}\left(\frac{\partial a}{\partial h}\right)^2+\left.\frac{\partial v_{\mathrm{E}}}{\partial a}\right\vert_{a=a(h,\beta)}\frac{\partial^2 a}{\partial h^2}.\]
Under the assumptions $\partial v_{\mathrm{E}}/\partial a>0$ and $\partial^2 v_{\mathrm{E}}/\partial a^2 \leq 0$,  the right-hand side is weakly negative if $\partial^2 a/\partial h^2 \leq 0$.

\vspace{5pt}
Finally, suppose that $v_{\mathrm{E}}(a,\beta)=\lambda_1 r(a,\beta)-\lambda_2 c(a)$, with $\lambda_1>0$ and $\lambda_2 \geq 0$. We verify that
$V_{\mathrm{E}}(h,\beta)$ has strict increasing differences in $(h,\beta)$, provided that  $a(h,\beta)$ has weak increasing differences in $(h,\beta)$ (as is already assumed). Our proof exploits the alignment between A and E's objectives. Specifically,  
\begin{align*}
\frac{\partial V_{\mathrm{E}}}{\partial h} =\left[\lambda_1 r_a\left(a\left(h,\beta\right),\beta\right)-\lambda_2 c'(a(h,\beta)\right]\frac{\partial a}{\partial h}=
\left(\frac{\lambda_1}{h}-\lambda_2 \right)c'(a(h,\beta))\frac{\partial a}{\partial h}, 
\end{align*}
where the second equality uses A's first-order condition  \eqref{eqn_afoc}. When $h \geq \lambda_1/\lambda_2$, we have $\lambda_1/h-\lambda_2 \leq 0$, so $\frac{\partial V_{\mathrm{E}}}{\partial h}(h,\beta)-\kappa'(h)<0$. Therefore, it is without loss of generality (w.l.o.g.) to restrict the domain of $h$ to $(0, \min\{1, \lambda_1/\lambda_2\})$ in E's problem \eqref{eqn_eproblem}. Over this constrained domain,   
\[\frac{\partial^2 V_{\mathrm{E}}}{\partial h \partial \beta}=\left(\frac{\lambda_1}{h}-\lambda_2 \right)\left[c''(a(h,\beta))\frac{\partial a}{\partial \beta}\frac{\partial a}{\partial h}+c'(a(h,\beta))\frac{\partial^2 a}{\partial h\partial \beta} \right]>0,\]
so $h(\beta)$ is strictly increasing in $\beta$. 
\end{proof}

\subsection{Genericity of Assumption \ref{assm_regularity2}}\label{sec_assumption2}
\begin{proposition}\label{prop_genericity}
Assumption \ref{assm_regularity2} holds generically. 
\end{proposition}

\begin{proof}

Define
\begin{equation}\label{eqn_f}
    f(\beta) \coloneqq R(h(\beta),\beta)-R(h(\beta),\beta^*)\ \text{ for } \beta \geq 0, 
\end{equation}
with derivative w.r.t. $\beta$ 
\[f'(\beta)=\left[R_h(h(\beta),\beta)-R_h(h(\beta),\beta^*)\right]h'(\beta)+R_{\beta}(h(\beta),\beta).\]
By Sard's lemma, there exists a set $C$ of measure zero such that for all $\Delta_{\mu} \notin C$,  $f'(\beta)\neq 0$ for all $\beta \in f^{-1}(-\Delta_{\mu})$.

Fix any $\Delta_{\mu} \notin C$, and let $\hat{\beta}$ be a Berk--NE belief. 
If $\hat{\beta} \in (\underline{\beta},\bar{\beta})$, then $\hat{\beta}$ is an SCE belief. 
It follows that on a small neighborhood of $\hat{\beta}$, Equation \eqref{eqn_perfectperception}: 
\[
R\bigl(h(\beta),x\bigr)-R\bigl(h(\beta), \beta^*\bigr)=-\Delta_{\mu},
\]
has a root $\psi(\beta)$ satisfying $\psi(\beta)=\tilde{\psi}(\beta)$ and $\psi(\hat{\beta})=\hat{\beta}$. 
By the implicit function theorem,
\[
\psi'(\beta)
= 
\frac{\bigl[R_h\bigl(h(\beta),\psi(\beta)\bigr)-R_h\bigl(h(\beta),\beta^*\bigr)\bigr]\,h'(\beta)}
{R_{\beta}\bigl(h(\beta),\psi(\beta)\bigr)}
\]
on the neighborhood. Evaluating at $\beta=\hat{\beta}$ using $\psi(\hat{\beta})=\hat{\beta}$ yields 
\[\psi'(\hat{\beta})=\frac{\bigl[R_h\bigl(h(\hat{\beta}),\hat{\beta}\bigr)-R_h\bigl(h(\hat{\beta}),\beta^*\bigr)\bigr]\,h'(\hat{\beta})}
{R_{\beta}\bigl(h(\hat{\beta}),\hat{\beta}\bigr)}\neq 1,\] where the inequality follows from $f'(\hat{\beta})\neq 0$.

Now suppose $\hat{\beta}=\bar{\beta}$. Since the set 
\[
\{\beta\in[\underline{\beta},\bar{\beta}]: \tilde{\psi}(\beta)=\beta\}
\]
is compact and therefore has a finite open cover, either  $\tilde{\psi}(\beta)\neq \beta$ on a small neighborhood $(\bar\beta-\epsilon,\bar\beta)$, 
or $\tilde{\psi}(\beta)= \beta$ on a small neighborhood $(\bar\beta-\epsilon,\bar\beta)$. 

In the first case, $\tilde{\psi}(\beta)=\bar{\beta}$ over $(\bar\beta-\epsilon,\bar\beta]$, 
so $\tilde{\psi}$ crosses the $45^\circ$ line strictly from above at $\bar{\beta}$. 

In the second case, $\tilde{\psi}(\beta)= \psi(\beta)=\beta$ and so $\psi'(\beta)=1$ over $(\bar\beta-\epsilon,\bar\beta)$. However, this is impossible given what has been show above.

Finally, the case $\hat{\beta}=\underline{\beta}$ is analogous.
\end{proof}

\subsection{Omitted proofs}\label{sec_op}
\paragraph{Proof of Theorem \ref{thm_direction}.} In the main text, we assumed that the function $\psi:[\underline{\beta},\bar{\beta}]\rightarrow \mathbb{R}_+$ is well defined, where each $\psi(\beta)$ solves equation \eqref{eqn_perfectperception} in $x$:
\[
R\bigl(h(\beta),x\bigr)-R\bigl(h(\beta),\beta^*\bigr)=-\Delta_{\mu}
\]
for $\beta \in [\underline{\beta},\bar{\beta}]$. For a fixed interval $[\underline{\beta},\bar{\beta}]\subset \mathbb{R}_{++}$, this assumption holds in either of the following situations:
\begin{itemize}
    \item $-\left(\lim_{x \rightarrow +\infty} R(h(\underline{\beta}),x)-R(h(\underline{\beta}), \beta^*)\right)<\Delta_{\mu}<0$.
    \item $0<\Delta_{\mu}<R(h(\underline{\beta}), \beta^*)-R(h(\underline{\beta}),\underline{\beta})$.
\end{itemize}

We now complete the proof for cases where the conditions above may fail.

Consider Case A: $\Delta_{\mu}<0$, holding $[\underline{\beta},\bar{\beta}]$ fixed. First, note that $\tilde{\psi}(\beta)$---defined as the choice of $x\in[\underline{\beta},\bar{\beta}]$ that minimizes the gap between the two sides of \eqref{eqn_perfectperception}---must satisfy $\tilde{\psi}(\beta)>\beta^*$. Second, consider an increase of $\beta$ from $\beta'$ to $\beta''$, both in $[\underline{\beta},\bar{\beta}]$. Since $h(\cdot)$ is strictly increasing and $R(\cdot,\cdot)$ has strictly increasing differences, the left-hand side of \eqref{eqn_perfectperception} strictly increases in $\beta$ for any fixed $x>\beta^*$. To minimize the gap between the two sides of \eqref{eqn_perfectperception}, $\tilde{\psi}(\beta)$ must decrease, with the decrease being strict whenever $\tilde{\psi}(\beta'')<\bar{\beta}$. The remainder of the proof proceeds exactly as in the main text.

The proof for Case B: $\Delta_{\mu}>0$, Parts (i) and (ii), follows analogously.

Finally, consider Case B Part (iii). We wish to demonstrate that for fixed $\bar{\beta}>\beta^*$ and all $\Delta_{\mu}>0$, there exists $\underline{\underline{\beta}} \in (0,\beta^*)$ such that for all $\underline{\beta}\in(0,\underline{\underline{\beta}})$, $\underline{\beta}$ is a Berk--NE belief when the belief space is $[\underline{\beta},\bar{\beta}]$.

Recall the function defined in \eqref{eqn_f}:
\[
f: \mathbb{R}_+ \rightarrow \mathbb{R},\quad  f(\beta)= R\bigl(h(\beta),\beta\bigr)-R\bigl(h(\beta),\beta^*\bigr).
\]
Note that $f(0)=f(\beta^*)=0$, and that $f<0$ and is nonmonotone on $(0,\beta^*)$. Therefore, for any given $\Delta_{\mu}>0$, there exists $\underline{\underline{\beta}} \in (0,\beta^*)$ such that $f(\beta)>-\Delta_{\mu}$ for all $\beta\in[0,\underline{\underline{\beta}})$. Pick any $\underline{\beta}\in(0,\underline{\underline{\beta}})$. If $\psi(\underline{\beta})$ is well defined, the proof proceeds exactly as in the main text. If $\psi(\underline{\beta})$ is not defined, then
\[
R\bigl(h(\underline{\beta}),0\bigr) - R\bigl(h(\underline{\beta}),\beta^*\bigr) = 0 - R\bigl(h(\underline{\beta}),\beta^*\bigr) > -\Delta_{\mu},
\]
implying that $\tilde{\psi}(\underline{\beta})=\underline{\beta}$ minimizes the gap between the two sides of \eqref{eqn_perfectperception}, as desired. \qed

\vspace{-10pt}

\paragraph{Proof of Proposition \ref{prop_disparity}.} We derive only the expression for the expected market reward.

Fix an SCE, and let $\hat{\beta}$ denote society's belief about  
the productivity of agent $i\in\{m,w\}$.
The assessment applied to the agent is $h(\hat{\beta})$, and their effort is perceived as 
$
\hat{a}=a\bigl(h(\hat{\beta}),\hat{\beta}\bigr).
$ 
Conditional on observing $X$, the market rewards the agent by
\[
\mathbb{E}_{\delta_{(\hat{\mu},\hat{\beta},\hat{a})}}
\bigl[\eta + v_{\mathrm{M}}(a,\beta)\mid X\bigr]
=
\hat{\mu}
+ h(\hat{\beta})
\left[X - \hat{\mu} - R\bigl(h(\hat{\beta}),\hat{\beta}\bigr)\right]
+ v_{\mathrm{M}}(\hat{a},\hat{\beta}).
\]
By the definition of SCE, the true mean of $X$,
$
\mu^* + R\bigl(h(\hat{\beta}),\beta^*\bigr),
$
equals the perceived mean,
$
\hat{\mu} + R\bigl(h(\hat{\beta}),\hat{\beta}\bigr).
$
Hence the second term on the right-hand side has a true mean of zero, and the entire expression has a true mean of $\hat{\mu}+v_{\mathrm{M}}(\hat{a},\hat{\beta})$. \qed 

\vspace{-10pt}

\paragraph{Proof of Proposition \ref{prop_decomposition}.} In Section \ref{sec_discussion}, we used  $\delta_x$ to denote E's first-order belief about A's productivity, and $\delta_{\delta_y}$ to denote E's second-order belief about A's first-order belief. Letting $y=x$ and $y=\beta^*$ gives rise to our model and the first-order misspecification scenario in \cite{murooka2021misspecified}, respectively. 

The optimal assessment for given $(x,y) \in \mathbb{R}_{+}^2$ solves  
\[\max_{h\in [0,1]}\  v_{\mathrm{E}}(a(h,y),x)-\kappa(h).\]
Under Assumption \ref{assm_primitives}, the solution $h(x,y)$ is unique for all $(x,y)\in\mathbb{R}_+^2$, belongs to $(0,1)$ for all $(x,y)\in \mathbb{R}_{++}^2$, and is strictly increasing in $x$ for all $y>0$  (details are available upon request). 

Define  
\[
f(x,y) 
:= r\bigl(a(h(x,y),y),x\bigr) - r\bigl(a(h(x,y),\beta^*),\beta^*\bigr) \quad \text{for} \quad  (x,y)\in\mathbb{R}_+^2,
\]
noting that (i) 
$f(\beta^*,\beta^*)=0$; (ii) $f(x,\beta^*), f(x,x)<0$ on $(0,\beta^*)$; (iii) $f(x,\beta^*), f(x,x)>0$ on $(\beta^*,+\infty)$. Therefore both functions $x \mapsto f(x,\beta^*)$ and $x \mapsto f(x,x)$ are increasing in $x$ near $x=\beta^*$. For small $|\Delta_{\mu}|$, the roots of $f(x,x)=-\Delta_{\mu}$ and $f(x,\beta^*)=-\Delta_{\mu}$ that are closest to 
$\beta^*$ correspond to the least distorted SCE beliefs in our model and that in MY21, respectively. 

To compare $f(x,x)$ and $f(x,\beta^*)$, differentiate $f(x,y)$ with respect to $y$:
\begin{align*}
\frac{\partial f}{\partial y} = & \left[r_a\left(a(h(x,y),y), x\right) \frac{\partial a}{\partial h} (h(x, y), y) -  r_a\left(a(h(x,y),\beta^*), \beta^*\right)\frac{\partial a}{\partial h} (h(x, y), \beta^*)\right] \frac{\partial h}{\partial y}  \\
&+ r_a\left(a(h(x,y),y), x\right)\frac{\partial a}{\partial \beta}\left(h(x,y),y\right). 
\end{align*}
When $(x,y) \approx (\beta^*,\beta^*)$, the first line on the right-hand side becomes 
negligible, whereas the second line is strictly positive. 
Therefore, $f(x,x)$ increases faster than $f(x,\beta^*)$ in a small neighborhood 
of $x=\beta^*$, hence the desired results. \qed 

\vspace{-10pt}

\paragraph{Proof of Theorem \ref{thm_colorsighted}.} W.l.o.g. write 
\[
\{j:\Delta_j>0\}=\{1,\dots,K\} 
\quad \text{ and } \quad 
\{j:\Delta_j<0\}=\{K+1,\dots,J\}.
\]
Define
\[
D \coloneqq \prod_{j=1}^K [\beta_j^*,\bar\beta] 
        \times \prod_{j=K+1}^J [\ul\beta,\beta_j^*].
\]
An inspection of the KL minimization problem \eqref{eqn_dkl_aa} reveals that any solution must lie in $D$. Hence, w.l.o.g., we restrict the domain of $\bm{\beta}$ to $D$.
 
\vspace{8pt}

\noindent Parts (i) and (ii): Fix any $\bm{\beta} \in D$. The same argument used to establish Lemma \ref{obs1} implies that in the current setting, the optimal assessment for $\bm{\beta}$ --- which solves \eqref{eqn_eproblem_aa} --- is uniquely determined by the first-order condition 
\[\sum_{j=1}^J \alpha_j \frac{\partial V_{\mathrm{E}}}{\partial h}(h, \beta_j)=\kappa'(h).\]
Since $V_{\mathrm{E}}(h,\beta)$ has strict increasing differences in $(h,\beta)$, the left-hand side of the first-order condition is bounded below by 
\[\sum_{j=1}^J \alpha_j \frac{\partial V_\mathrm{E}}{\partial h}(h,\beta_j)>\sum_{j=K+1}^J \alpha_j   \frac{\partial V_\mathrm{E}}{\partial h}(h, \beta_j^*).\]
Therefore, $h(\bm{\beta})$ is greater than the unique root of 
\[\sum_{j=K+1}^J \alpha_j  \frac{\partial V_{\mathrm{E}}}{\partial h}(h,\beta_j^*)=\kappa'(h),\]
which we denote by $\ul{\ul h}$ to highlight its independence from the vector of ability misspecifications   $\bm{\Delta}= (\Delta_j)_{j=1}^J$. 

Let $|\Delta_j|$ be small for $j=1,\cdots J$. For each $j$, define $\psi_j (\bm{\beta})$ as the root of 
\[R\left(h(\bm{\beta}),x\right)-R\left(h(\bm{\beta}),\beta_j^*\right) = -\Delta_j,\] 
When $\Delta_j<0$, 
\[-\Delta_j=R\left(h(\bm{\beta}),x\right)-R\left(h(\bm{\beta}),\beta_j^*\right)>R(\ul{\ul h}, \psi_j(\bm{\beta}))-R(\ul{\ul h}, \beta_j^*)>0,\]
where the first inequality uses the assumption that $R(h,\beta)$ has strict increasing differences, and the second inequality the fact that $\psi_j(\bm{\beta})>\beta_j^*$. 
It follows that $\psi_j(\bm{\beta})-\beta_j^* = O(|\Delta_j|)$. The argument for $\Delta_j>0$ is analogous. 

By the implicit function theorem, 
\[\frac{\partial \psi_j}{\partial \beta_k}=g_j(\bm{\beta}; \Delta_j)\frac{\partial h}{\partial \beta_k}, \text{ where }\]
\[g_j(\bm{\beta};\Delta_j)= \frac{\frac{\partial R}{\partial h}(h(\bm{\beta}), \beta_j^*)-\frac{\partial R}{\partial h}(h(\bm{\beta}), \psi_j(\bm{\beta}))}{\frac{\partial R}{\partial \beta}(h(\bm{\beta}), \psi_j(\bm{\beta}))}.\]
It follows that $g_j(\bm{\beta};\Delta_j) = O(|\Delta_j|)$, with $\sgn(g_j(\bm{\beta};\Delta_j))=\sgn(\Delta_j)$.

Define $g(\bm{\beta};\bm\Delta)= (g_j(\bm{\beta}; \Delta_j))_{j=1}^J$ and $\bm{\psi}(\bm{\beta})= (\psi_j(\bm{\beta}))_{j=1}^J$. 
Straightforward algebra shows that the Jacobian $\nabla_{\bm{\beta}}\bm{\psi}$ takes the rank-one form:
\[
\nabla_{\bm{\beta}}\bm{\psi}=g(\bm{\beta}; \bm{\Delta})\,(\nabla_{\bm{\beta}} h)^{\top},
\]
Therefore, the Euclidean norm of the Jacobian equals the product of the factor norms: 
\[
\|\nabla_{\bm{\beta}}\bm{\psi}\|_2=\|g(\bm{\beta}; \bm{\Delta})\|_2 \; \| \nabla_{\bm{\beta}} h\|_2.
\]
Since $\|g(\bm{\beta}; \bm{\Delta})\|_2 = O(\|\bm{\Delta}\|_2)$, and $\sup_{\bm{\beta} \in D} \|\nabla_{\bm{\beta}}h\|_2$ is finite and independent of $\bm{\Delta}$, the right-hand side vanishes as $\|\bm\Delta\|_2 \to 0$. It follows that for sufficiently small $\|\bm{\Delta}\|_2$, the function $\bm{\psi}$ restricts to a contraction on $D$. The contraction mapping theorem guarantees the existence of a unique fixed point, which is globally asymptotically stable (in the Lyapunov sense) relative to $D$. It follows that the fixed point is locally asymptotically stable.\footnote{Global asymptotic stability may fail on the larger unrestricted domain $[\underline{\beta},\bar{\beta}]^J$ during dynamic learning studied in Appendix \ref{sec_learning_aa}; hence  the theorem establishes only local asymptotic stability.}

\vspace{8pt}

\noindent Parts (iii) and (iv): Let $\hat{\bm{\beta}}$ denote the fixed point derived above. We examine its comparative statics with respect to a parameter $\rho \in \mathbb{R}$, whose interpretation depends on the context. Defining $\Psi(\bm{\beta}, \rho)=\bm{\psi}(\bm{\beta}, \rho)-\bm{\beta}$, the fixed point satisfies $\Psi(\hat{\bm{\beta}}(\rho),\rho)=\bm{0}$. Applying the implicit function theorem yields
\[
\frac{\partial \hat{\bm{\beta}}}{\partial \rho}=-\left.(\nabla_{\bm\beta} \Psi)^{-1}\; \frac{\partial \Psi}{\partial \rho}\right\vert_{\bm{\beta}=\hat{\bm{\beta}}(\rho)}.
\]
Observing that $\nabla_{\bm{\beta}}\Psi = -Id_J + \nabla_{\bm{\beta}} \bm{\psi}$ is a rank-one perturbation of the identity matrix, we apply the Sherman--Morrison formula to compute its  inverse:

\begin{fact}[\citealp{sherman1950adjustment}]\label{fact_shermanmorrison}
Let $A$ be an $J \times J$ invertible matrix, and let $u$ and $v$ be $J$-vectors. $A+uv^{\top}$ is invertible if and only if $1+v^{\top}A^{-1}u \neq 0$, in which case
\[
(A+u v^{\top})^{-1}=A^{-1}-\frac{A^{-1}u v^{\top}A^{-1}}{1+v^{\top}A^{-1}u}.
\]
\end{fact}
Letting $A=-Id_{J}$, $u=g(\bm{\beta};\bm{\Delta})$, and $v=\nabla_{\bm{\beta}} h$ in the Sherman-Morrison formula, straightforward algebra shows that
\[
\frac{\partial \hat{\bm{\beta}}}{\partial \rho} = \left.\left(Id_J + \frac{g(\bm{\beta}; \bm{\Delta}) (\nabla_{\bm{\beta}} h)^{\top}}{1 - (\nabla_{\bm{\beta}} h)^{\top} g(\bm{\beta}; \bm{\Delta})}\right)\; \frac{\partial \Psi}{\partial \rho} \right\vert_{\bm{\beta} = \hat{\bm{\beta}}(\rho)}.
\]

Consider first the effect of perturbing $\rho=\Delta_1$. We observe that $\tfrac{\partial \Psi}{\partial \Delta_1} = \tfrac{\partial \bm{\psi}}{\partial \Delta_1}$, where
\[
\frac{\partial \psi_{j}}{\partial \Delta_1} \begin{cases}
<0 & \text{if } j=1, \\[6pt]
=0 & \text{if } j \neq 1,
\end{cases}
\]
where the inequality follows from the same reasoning as in Part~(i) of Theorem~\ref{thm_direction}, while the zero derivative is evident from the definition of $\psi_j$. Substituting this into the expression for $\tfrac{\partial \hat{\bm{\beta}}}{\partial \Delta_1}$ gives
\[
\frac{\partial \hat{\beta}_j}{\partial \Delta_1} \begin{cases}
<0 & \text{if } j=1, \\[6pt]
\approx 0 & \text{if } j \neq 1,
\end{cases}
\qquad \text{for } \bm{\Delta}\approx \bm{0}.
\]

Next, consider the effect of perturbing $\rho=\zeta$. Straightforward algebra shows that
\[
\frac{\partial \Psi}{\partial \zeta}
= \frac{\partial h}{\partial \zeta} \ g(\bm{\beta};\bm{\Delta}).
\]
Hence
\begin{align*}
\frac{\partial \hat{\bm{\beta}}}{\partial \zeta}
&= \left.\frac{\partial h}{\partial \zeta}\,
\left(Id_J+\frac{g(\bm{\beta};\bm{\Delta})(\nabla_{\bm{\beta}} h)^{\top}}
{1-(\nabla_{\bm{\beta}} h)^{\top} g(\bm{\beta};\bm{\Delta})}\right)\, g(\bm{\beta};\bm{\Delta})
\right|_{\bm{\beta}=\hat{\bm{\beta}}} \\[6pt]
&= \frac{\partial h}{\partial \zeta} \,
\left.\frac{1}{1-(\nabla_{\bm{\beta}} h)^{\top} g(\bm{\beta};\bm{\Delta})}
\, g(\bm{\beta};\bm{\Delta})\right|_{\bm{\beta}=\hat{\bm{\beta}}}.
\end{align*}
In the last line, $\tfrac{\partial h}{\partial \zeta}>0$ because the objective in \eqref{eqn_eproblem_aa} has strictly increasing differences in $(h,\zeta)$. The fraction is close to one for $\bm{\Delta}\approx \bm{0}$, and as shown earlier, $\sgn\bigl(g_j(\bm{\beta};\Delta_j)\bigr)=\sgn(\Delta_j)$. Combining these observations yields $\sgn\!\left(\tfrac{\partial \hat{\beta}_j}{\partial \zeta}\right) = \sgn(\Delta_j).$ \qed

\cleardoublepage

\vspace*{16em}
\begin{center}
    \Huge{
    Online Appendix for \\ ``Implicit Incentive Provision with Misspecified Learning''\\ by Federico Echenique and Anqi Li}
\bigbreak
\end{center}

\thispagestyle{empty}
\cleardoublepage

\section{Learning foundation}\label{sec_learning}
This appendix develops the misspecified Bayesian learning foundation of the Berk–Nash in the steady-state model.

Time is discrete and infinite, and all players are Bayesian. At the outset, the true fundamental \(\omega^*=(\mu^*,\beta^*)\) is revealed once and for all to A. S---comprising E and M---knows this structure but not the exact fundamental value $\omega$.  The prior held by S is degenerate at \(\hat{\mu} \neq \mu^*\) along the \(\mu\)-argument. Along the 
\(\beta\)-argument, the prior admits a smooth, full-support density 
\(p_0\) on \([\ul\beta,\bar\beta] \subset \mathbb{R}_{++}\), with $\beta^* \in (\ul\beta,\bar\beta)$.

In every period $n \in \mathbb{N}$, the following events unfold: 
\begin{enumerate}
    \item E publicly commits to an assessment intensity $h_n \in [0,1]$.
    \item A privately chooses an effort level $a_n \geq 0$.
    \item The period-$n$ ability $\eta_n \sim N(\mu^*,1)$ and noise $\varepsilon_n \sim N(0,h_n^{-1}-1)$ are realized independently of each other and over time.  These determine the outcome
    \[X_n=\eta_n+r(a_n,\beta^*)+\varepsilon_n, \]
    which is publicly realized. 
    \item M forms an expectation of A's period-$n$ value $\eta_n+v_{\mathrm{M}}(a_n,\beta)$ and rewards A  accordingly. 
    \item S updates its beliefs about $\beta$.
\end{enumerate}

Much of the analysis assumes that players are myopic, maximizing expected payoffs period by period. This assumption is standard in the misspecified learning literature, and we view it as a reasonable approximation of reality:  while society accumulates institutional knowledge about the fundamental, individual members---whether students, teachers, workers, or employers---are short-lived. The case of a forward-looking evaluator can also be studied by within our framework, as discussed toward the end of this appendix.

Let $p_n(\cdot\mid H_n)$ denote the posterior density of $\beta$ in Stage 5 of period $n$, formed based on the public history $H_{n} \coloneqq  (X_i,h_i)_{i=1}^{n}$. For notation simplicity, we omit the dependence of $p_n$ on $H_n$ whenever the context is clear. 

Consider the stage game in period $n$. Let $\hat{a}(h_n,\beta)$ denote society's belief of A's  effort in Stage 4 (as a random variable of $\beta$), where $\beta \sim p_{n-1}(\cdot \mid H_{n-1})$, irrespective of the true effort in period $n$. 
Conditional on observing $X_n$, M's expectation of A's period-$n$ value (as a random variable of $\beta$) is 
\begin{align*}
\mathbb{E}\!\left[\eta_n 
    + v_{\mathrm{M}}\bigl(\hat{a}(h_n,\beta),\beta\bigr) 
    \,\middle|\, h_n, X_n\right]
=& \hat{\mu} 
  + \frac{\sigma^2_{\eta}}{\sigma^2_{X_n}}
    \left[X_n - \hat{\mu} - r\bigl(\hat{a}(h_n,\beta),\beta\bigr)\right]
  \\
  & + v_{\mathrm{M}}\bigl(\hat{a}(h_n,\beta),\beta\bigr).
\end{align*}
Taking expectation over $\beta \sim p_{n-1}$ and using the fact that $\sigma^2_{\eta}/\sigma^2_{X_n}= h_n$ simplifies the expression to
\[h_n X_n+\text{terms unaffected by the true effort in period $n$}.\]
A can influence the mean of $X_n$ through adjusting their actual choice of effort. Combined with the derivation in Section \ref{sec_setup}, we obtain $a(h_n, \beta^*)$ as the true effort, and $a(h_n,\beta)$ as the perceived effort. E's problem in period $n$ is therefore
\begin{equation}\label{eqn_eproblem_learning}
    \max_{h \in [0,1]} \int_{\ul\beta}^{\bar\beta} V_{\mathrm{E}}(h,\beta) \; p_{n-1}(\beta \mid H_{n-1})\; d\beta-\kappa(h).
\end{equation}
The unique solution, denoted by $h_n$, depends on only $p_{n-1}(\cdot \mid H_{n-1})$, and so is adapted to $H_{n-1}$. Since $p_{n-1}$ first-order stochastically dominates $\delta_{\ul\beta}$ and is in turn dominated by $\delta_{\bar\beta}$, $h_n \in [h(\ul\beta), h(\bar\beta)]$ always holds. For notational simplicity, we write  $[\ul h,\bar h]$ for $[h(\ul\beta), h(\bar\beta)]$, and note that $[\ul h, \bar h] \subset (0,1)$. 

Consider next the belief updating process. Applying Bayes' rule, we obtain 
\begin{equation}\label{eqn_posteriordensity}
p_{n}(\beta \mid H_{n}) = 
\frac{
    \exp\left(
        \displaystyle -\sum_{i=1}^{n} 
        \frac{(X_i - \hat{\mu} - R(h_i, \beta))^2}{2 h_i^{-1}}
    \right) 
    p_0(\beta)
}{
    \displaystyle \int_{\ul\beta}^{\bar\beta} 
    \exp\left(
        \displaystyle -\sum_{i=1}^n 
        \frac{(X_i - \hat{\mu} - R(h_i, \tilde{\beta}))^2}{2 h_i^{-1}}
    \right) 
    p_0(\tilde{\beta}) \, d\tilde{\beta}}.
\end{equation}
as the posterior density of $\beta$ given $H_n$, formed in Stage 5 of period $n$.

We wish to characterize society's limiting posterior beliefs as time approaches infinity. The analysis makes the following assumption, whose role will soon become clear: 

\begin{assumption}\label{assm_linearr}
    $R(h,\beta)=g_1(\beta)g_2(h)+g_3(h)$, where both $g_1:\mathbb{R}_+\rightarrow \mathbb{R}$ and $g_2:[0,1]\rightarrow \mathbb{R}$ 
    are strictly increasing and strictly positive on the interiors of their domains.
\end{assumption}

\begin{theorem}\label{thm_learning}
Suppose that Assumptions \ref{assm_regularity1}, \ref{assm_regularity2}, and \ref{assm_linearr} hold.  As $n\rightarrow \infty$:
\begin{enumerate}[(i)]
    \item Society's posterior beliefs about $\beta$ almost surely converge to the set of the $\hat{\beta}$ values that emerge in the locally asymptotically stable Berk--NE of the steady-state model.
    \item For every such $\hat{\beta}$, convergence occurs with strictly positive probability, and with probability arbitrarily close to one given a sufficiently concentrated prior around $\hat{\beta}$.
\end{enumerate}
\end{theorem}

The remainder of this appendix proves Theorem \ref{thm_learning}.
\vspace{-10pt}
\paragraph{Discussion of challenges.} Suppose that, as in Section \ref{sec_example}, \( r(a, \beta) = \beta a \), so that outcome is linear in effort productivity for any given effort level. A close inspection of \eqref{eqn_posteriordensity} reveals that even in this simple setting, posteriors of \( \beta \) do \emph{not} follow truncated normal distributions, even under a uniform or truncated normal prior. 

This departure does not stem from the choice of prior, but the persistent asymmetric information between S and A about the latter's productivity. Under common knowledge about A's strategy $a(\cdot, \cdot)$, such information asymmetry yields strategic uncertainties about A's effort choices. From society's perspective, conditional on the public history $H_{n-1}$ (to which $h_n$ is adapted), the effort choice $a(h_n, \beta)$ in period $n$ remains random.    Substituting $a(h_n,\beta)$ into \( r(a, \beta) \) gives \( \beta a(h_n, \beta) \), which is no longer linear in \( \beta \). Consequently, the log density of the outcome: 
\[
-\frac{(X_n - \hat{\mu} - \beta a(h_n, \beta))^2}{2h_n^{-1}},
\]
is not quadratic in \( \beta \), so the posterior cannot be truncated normal. The proof strategy developed below exploits the parametric structure of the Gaussian family. Outside this family, tracking the belief updating process and characterizing asymptotic beliefs becomes much more challenging. 

By contrast, truncated normality is preserved under specific priors (e.g., uniform, truncated normal), absent strategic uncertainties. That is, if S believes---whether correctly or mistakenly---that A's effort choice is adapted to the public history and so is fully predictable.  In the log-density function above, replace $a(h_n,\beta)$ with some $\hat{a}_n$ adapted to $H_{n-1}$. The resulting expression is quadratic in $\beta$. 

The assumption that players view others' behavior as fully predictable is standard in the economic literature on misspecified Bayesian learning. For example, 
\begin{itemize}
    \item It is vacuously true in the single-agent models \citep{heidhues2018unrealistic, heidhues2021convergence, esponda2021asymptotic, fudenberg2021limit}, and in multi-player models where behavior is publicly observable \citep{ba2023multi}.
    \item It holds in the first-order misspecification scenario studied by  \citet{murooka2021misspecified,murooka2023higher} (MY21, 23), where  beliefs about the fundamental may differ across players yet remain common knowledge. Recasting this in our setting, it is common knowledge that A's effort choice is $a(h_n,\beta^*)$, which is adapted to $H_{n-1}$.
    \item It also holds in the double-misspecification scenario of MY21 and MY 23, where each player mistakenly believes that others share their own beliefs about the fundamental. Translated to our context, S believes A solves
    \[
\max_{a \in [0,1]} \;
\int_{\ul\beta}^{\bar\beta} 
    h_n \bigl(\hat{\mu} + r(a,\beta)\bigr) 
    \, p_{n-1}(\beta \mid H_{n-1}) 
\, d\beta 
- c(a),
\]
   where $\hat{\mu}+ r(a,\beta)$ is the perceived mean of $X_n$ given $a$ and $\beta$. The solution is again adapted to $H_{n-1}$.
\end{itemize}
In Section~\ref{sec_discussion}, we already compared our setting with MY21 and MY23 in terms of equilibrium. Now it is clear that these differences also matter for the learning dynamics.  

Strategic uncertainty is intrinsic to our model and, more broadly, to games of incomplete information. Addressing the challenges it poses on misspecified Bayesian learning remains an open question. 

\vspace{-10pt}
\paragraph{Proof overview.} Assumption~\ref{assm_linearr} is introduced to address the above  challenges. It stipulates that the optimal effective effort is affine in \( \breve{\beta}= g_1(\beta) \), thereby enabling a change of variable from \( \beta \) to \( \breve{\beta} \). Since \( g_1 \) is strictly increasing, \( g_1^{-1}(\breve{\beta}) \) is unique for all \( \breve{\beta} \in g_1([\ul\beta, \bar\beta]) \). When players' problems are expressed in terms of \( \breve{\beta} \), all results from the steady-state model carry over directly. 

In the learning model, this transformation allows us to express the log-density of the outcome as a quadratic function of \( \breve{\beta} \) (details to follow), thus bringing us closer to approximating  posteriors using truncated normals. The drawback, however, is that the prior on \( \breve{\beta} \) can be arbitrary, even if the original prior on \( \beta \) is uniform or truncated normal. Still, the objective is clear: to reduce the challenges posed by strategic uncertainty to a more primitive question, namely, how does a general prior affect posteriors? 

In the learning model, the prior influences posteriors directly through Bayesian updating and indirectly via players’ strategic behavior. We show that the overall impact decays sufficiently fast over time, using Laplace’s approximation and the Chernoff bound for normal random variables. 

The remainder of this appendix proceeds as follows:  
\begin{itemize}
    \item Appendix \ref{sec_laplace} provides background on Laplace's approximation. 

    \item Appendix \ref{sec_learning_withoutprior} proves Theorem \ref{thm_learning} under specific priors for which posteriors are exactly truncated normal, using stochastic approximation techniques. Key references are \cite{benaim1999dynamics} (B99), \cite{pemantle1990nonconvergence} (P90), and MY23. 
    
    \item Appendix \ref{sec_learning_withprior} demonstrates the asymptotic validity of the results of  Appendix \ref{sec_learning_withoutprior} under general priors. 
\end{itemize} 

In Appendix \ref{sec_learning_discussion}, we give interpretations of Assumption \ref{assm_linearr} and propose an alternative approach that dispenses with it. Other model variants, such as a forward-looking evaluator, are also discussed. Proofs are in Appendix \ref{sec_learning_proof}.

For notatio simplicity, we relabel \( \breve{\beta} \) as \( \beta \) and the prior of $\breve{\beta}$ as \( p_0 \), with support $[\ul\beta,\bar\beta]$. We adopt standard matrix notation: $\|\cdot\|$ denotes the Euclidean norm, $v_i$ denotes the $i^{\text{th}}$ coordinate of a vector $v$, and $Id$ denotes the identity matrix.

\subsection{Laplace's approximation}\label{sec_laplace}
Let $\bar p_0$ be any smooth density function with support on $\mathbb{R}$ such that 
\[p_0(\beta)=\frac{\bar p_0(\beta)\mathbf{1}_{\beta \in [\ul\beta, \bar\beta]}}{\int \bar p_0(\tilde{\beta})\mathbf{1}_{\tilde{\beta} \in [\ul\beta, \bar\beta]} d\tilde{\beta}} \ \text{ for } \beta \in [\ul\beta,\bar\beta]. \]
Under Assumption \ref{assm_linearr}, 
\begin{align*}
\log \bar p_{n}(\beta \mid H_{n}) = &\displaystyle \sum_{i=1}^{n} -\dfrac{\left(X_i - \hat{\mu} - \beta g_2(h_i)-g_3(h_i)\right)^2}{2h_i^{-1}}+ \log \bar p_0(\beta)  \\
&+ \text{terms independent of } \beta, 
\end{align*}
with the following derivatives w.r.t. $\beta$: 
\[\frac{d}{d\beta} \log \bar p_{n}(\beta \mid H_{n})= \sum_{i=1}^{n}\frac{X_i-\hat{\mu}-\beta g_2(h_i)-g_3(h_i)}{h_i^{-1}}g_2(h_i)+\frac{d\log \bar p_0}{d\beta}, \]
\[\frac{d^2}{d\beta^2} \log \bar p_{n}(\beta \mid H_{n})= \sum_{i=1}^{n} g^2_2(h_i)h_i+\frac{d^2 \log \bar p_0}{d\beta^2}, \text{ and }\]
\[\frac{d^3}{d\beta^3} \log \bar p_{n}(\beta \mid H_{n})= \frac{d^3 \log \bar p_0}{d\beta^3}.\] 
Notably, the third derivative is independent of $n$, indicating a negligible impact of derivatives beyond the second order as $n$ grows to infinity. 

The term $g^2_2(h_n)h_n \coloneqq I(h_n)$---or simply $I_n$---represents the Fisher information of the period-$n$ outcome $X_n$ as a signal of $\beta$. It depends only on $h_n$ and is strictly positive since $h_n\geq \ul h>0$ and $g_2(h)>0$ for all $h>0$. Intuitively, since assessment is always positive and induces positive effort, the outcome always reveals new information about $\beta$. In what follows, we write $[\ul I, \bar I] \coloneqq [I(\ul h), I (\bar h)]$, noting that $\ul I>0$.

Let $m_n$ denote the mode of $\bar p_n$. For large $n$, $m_n$ is determined by the first-order condition
\[\frac{d}{d\beta} \log \bar p_{n}(\beta \mid H_{n})= 0,
\]
and can be explicitly solved as
\begin{equation}\label{eqn_mtgeneral}
m_{n} = \frac{\displaystyle \sum_{i=1}^{n} (X_i - \hat{\mu}-g_3(h_i)) h_ig_2(h_i) + \left.\dfrac{d \log \bar p_0}{d\beta}\right\vert_{\beta=m_{n}}}{\displaystyle \sum_{i=1}^{n} I_i}.
\end{equation}
Evaluating the second derivative of $\log \bar p_n(\beta \mid H_{n})$ at $m_n$ then gives
\begin{equation}\label{eqn_vtgeneral}
\left.\frac{d^2}{d\beta^2} \log \bar p_{n}(\beta \mid H_{n})\right\vert_{\beta=m_n} = -\sum_{i=1}^{n} I_i + \left.\frac{d^2 \log \bar p_0}{d\beta^2} \right\vert_{\beta=m_n} \coloneqq -\frac{1}{v_{n}}.
\end{equation}
Based on \eqref{eqn_mtgeneral} and \eqref{eqn_vtgeneral}, expand $\log \bar p_{n}(\beta\mid H_{n})$ around $m_n$ as
\begin{align}
\log \bar p_{n}(\beta \mid H_{n}) 
  &= \log \bar p_{n}(m_n \mid H_{n})
     + 0  \; (\beta - m_n) \nonumber \\
  &\quad + \frac{1}{2} 
     \left.\frac{d^2}{d\beta^2}\log \bar p_n(\beta \mid H_{n})
     \right|_{\beta=m_n} \, (\beta - m_n)^2 \nonumber \\
  &\quad + \frac{1}{6} 
     \left.\frac{d^3}{d\beta^3}\log \bar p_0(\beta)
     \right|_{\beta=z(\beta, m_n)} \, (\beta - m_n)^3 \nonumber \\
  &= -\frac{(\beta - m_n)^2}{2v_n}
     + K(\beta,m_n)(\beta - m_n)^3
     + \text{terms independent of } \beta, \nonumber
\end{align}
where
\[
K(\beta, m_n) \coloneqq \frac{1}{6}
   \left.\frac{d^3}{d\beta^3}\log \bar p_0(\beta)
   \right|_{\beta=z(\beta,m_n)}.
\]
Exponentiating both sides yields
\[\bar p_{n}(\beta \mid H_{n}) \propto \exp\left(-\frac{(\beta - m_{n})^2}{2v_{n}}\right) \exp\left(K(\beta,m_n)(\beta - m_{n})^3\right).\]

On $[\ul\beta,\bar\beta]$, we have $p_0(\beta)\propto \overline{p}_0(\beta)$, so again 
\[p_n(\beta \mid H_{n}) \propto \exp\left(-\frac{(\beta - m_{n})^2}{2v_{n}}\right) \exp\left(K(\beta,m_n)(\beta - m_{n})^3\right).\]
On the right-hand side, the first exponent corresponds to a normal density with mean $m_{n}$ and variance $v_{n}$; the second exponent captures the impact of the third derivative on the posterior. In  Appendix \ref{sec_learning_proof_withprior}, we show that the process $\{m_n\}$ is almost surely (a.s.) bounded. It follows that $v_n = \Theta(1/n)$ and $K(\beta,m_n) = O(1)$ a.s. for large $n$, hence the second exponent diminishes exponentially fast relative to the first. Asymptotically, the posterior distribution of $\beta$ can be approximated by a truncated normal distribution, denoted by $\tilde{N}(m_n, v_n)$.

 \subsection{Stochastic approximation: Specific priors}\label{sec_learning_withoutprior}
This appendix establishes Theorem \ref{thm_learning} under the assumption that the prior is uniform on $[\ul\beta,\bar\beta]$. The analysis extends directly to truncated normal priors. 

Under uniform prior, $p_0$ cancels from the numerator and denominator of \eqref{eqn_posteriordensity}. The period-$n$ posterior---with density $p_n$ ---follows exactly a truncated normal distribution $\tilde{N}(m_n, v_n)$, where $m_n$ and $v_n$ can be obtained from omitting $p_0$ from \eqref{eqn_mtgeneral} and \eqref{eqn_vtgeneral}. The corresponding untruncated normal distribution is $N(m_n, v_n)$, for which we define  
\[
\xi_n \coloneqq \frac{1}{nv_n}
\]
as the inverse of the time-scaled variance. The vector $\theta_n\coloneqq (m_n, \xi_n)$ is  referred to as the period-$n$ \emph{state} and is shown to belong to $M\coloneqq \mathbb{R} \times [\ul I, \bar I]$. It follows that $p_n$ is fully parameterized by $n$ and $\theta_n$. Inspecting \eqref{eqn_eproblem_learning} reveals that $h_n$, and hence $I_n$, depends only on $n$ and $\theta_{n-1}$. 

We capture the evolution of $\theta_n$ through a system of stochastic difference equations. To begin, rewrite \eqref{eqn_vtgeneral} as 
\begin{equation}
\xi_{n+1}-\xi_n=\frac{1}{n+1}\left(I_{n+1}-\xi_n\right). \label{eqn_xiupdate}
\end{equation}
Since $I_{n}\in [\ul I, \bar I]$ always holds, the process $\{\xi_n\}$ is uniformly bounded in $[\ul I, \bar I]$. Meanwhile, recall that $
X_n - \hat{\mu} = R(h_n, \beta(h_n)) + \epsilon_n/\sqrt{h_n}
$, where $\beta(h_n) \in \mathbb{R}_+$ denotes the unconstrained best-fit parameter of the outcome distribution induced by $h_n$. Under the conditions prescribed in the proof of Theorem \ref{thm_direction}, in Appendix \ref{sec_op},  $\beta(h_n)$ is well defined for all $h_n \in [\ul h, \bar h]$. Projecting $\beta(h_n)$ onto $[\ul\beta,\bar\beta]$ yields the constrained best-fit parameter $\tilde{\beta}(h_n)$ in the steady-state model.

Substituting the above into \eqref{eqn_mtgeneral}, straightforward algebra yields 
\begin{equation}
m_{n+1}-m_n
= \frac{1}{n+1}\left(
   \frac{ I_{n+1}\,\bigl(\beta(h_{n+1})-m_n\bigr)
          + \sqrt{I_{n+1}}\,\epsilon_{n+1} }
        { \tfrac{n}{n+1}\,\xi_{n} + \tfrac{1}{n+1}\, I_{n+1} }
   \right). \label{eqn_mupdate}
\end{equation}
In \eqref{eqn_xiupdate} and \eqref{eqn_mupdate}, let $h_{n+1}=h_{n+1}(\theta_n)$ and $I_{n+1}=(I \circ h_{n+1})(\theta_n)$. In vector form, 
\begin{equation}\label{eqn_de}
\theta_{n+1}-\theta_n=\frac{1}{n+1}\Bigl(F_{n+1}(\theta_n)+b_{n+1}(\theta_n)\epsilon_{n+1}\Bigr),
\end{equation}
where 
\[
F_{n+1}(\theta) \coloneqq 
\begin{bmatrix}
\dfrac{(I\circ h_{n+1})(\theta) \left[(\beta\circ h_{n+1})(\theta) - m\right]}{\frac{n}{n+1}\xi + \frac{1}{n+1} (I\circ h_{n+1})(\theta)} \\[1em]
(I\circ h_{n+1})(\theta) - \xi
\end{bmatrix}, \quad 
b_{n+1}(\theta) \coloneqq 
\begin{bmatrix}
\dfrac{\sqrt{(I\circ h_{n+1})(\theta)}}{\frac{n}{n+1}\xi + \frac{1}{n+1} (I\circ h_{n+1})(\theta)} \\[1em]
0
\end{bmatrix}.
\]

Next, we approximate the asymptotic behavior of $\{\theta_n\}$ by a time-invariant ODE system. When $n$ is large, $v_n = \Theta(1/n)$ and so is negligible. Consequently, we can approximate the truncated normal distribution $\tilde{N}(m_n,v_n)$ by the degenerate distribution $\delta_{\tilde{m}_n}$, where $\tilde{m}_n$ denotes the projection of $m_n$ onto $[\ul\beta,\bar\beta]$. The optimal assessment given $\delta_{\tilde{m}_n}$ is  $h(\tilde{m}_n)$, where $h(\cdot)$ denotes  the assessment map in the steady-state model.  In (\ref{eqn_de}), replace $h_{n+1}(\theta_n)$ with $h(\tilde{m}_n)$, drop the noise term $\epsilon_n$, and take  $n \rightarrow \infty$. This gives an ODE system 
\begin{equation}\label{eqn_ode}
\dot{\theta}=F(\theta)\coloneqq \begin{bmatrix}
\dfrac{(I\circ h)(\tilde{m})\; \left(\psi(\tilde{m}) - m\right)}{\xi} \\
(I\circ h)(\tilde{m}) - \xi
\end{bmatrix},
\end{equation}
where $\psi(\tilde{m})=(\beta \circ h) (\tilde{m})$ because both sides represent the unconstrained best-fit productivity parameter of the outcome distribution induced by $h(\tilde{m})$. 

There is a one-to-one correspondence between the steady states of the ODE and the Berk--NE of the steady-state model. This is because a steady state $(\hat{m},\hat{\xi})$ of the ODE---$\hat{\xi}$ fully determined by $\hat{m}$---must satisfy $\psi(\widetilde{\hat{m}})=\hat{m}$. If $\hat{m} \in (\ul \beta,\bar\beta)$, then  $\psi(\hat{m})=\hat{m}$. In the steady-state model,  $\hat{m}$ corresponds to an SCE belief of effort productivity. If $\hat{m}\leq \ul\beta$, then $\widetilde{\hat{m}}=\ul\beta$. Near the steady state, the $m$-dimension of the ODE changes at rate $\psi(\ul\beta)-m$, until the system settles at  $\hat{m}=\psi(\ul\beta)$. It follows that $\psi(\ul\beta) \leq \ul\beta$ and so $\tilde{\psi}(\ul\beta)=\ul\beta$. Therefore, $\ul\beta$ is a Berk--NE belief in the steady-state model. The argument for $\hat{m} \geq \bar\beta$ is analogous.

The stability properties of Berk–NE parallel those of the ODE's steady states. Since the second dimension of the ODE always converges, stability is determined solely by the first dimension. If $\hat{m}\in(\ul\beta,\bar\beta)$ and $\psi$ crosses the 45-degree line strictly from above at $\hat{m}$, then the steady state $(\hat{m},\hat{\xi})$ is a sink and so is locally asymptotically stable in the Lyapunov sense. If the crossing occurs strictly from below, the steady state is a saddle point and therefore Lyapunov unstable. If $\hat{m}\leq \ul\beta$, the steady state is again a sink, following the description in the above paragraph. The case $\hat{m}\geq \bar\beta$ is analogous.

By Theorem \ref{thm_direction}, steady states, ordered by their $m$-dimension from largest to smallest, must alternate in type: sink $\to$ saddle $\to$ sink $\to \cdots$. Moreover, the $m$-coordinate of any saddle point must lie in $(\ul\beta,\bar\beta)$. The attraction basin of a saddle point is the vertical line passing  through it. Every other point in $M$ lies in the attraction basin of a sink. 

Figure \ref{fig:phasediagram} illustrates the above discussions in the linear–quadratic example.

\begin{figure}[htbp]
    \centering
    \includegraphics[width=.9\textwidth]{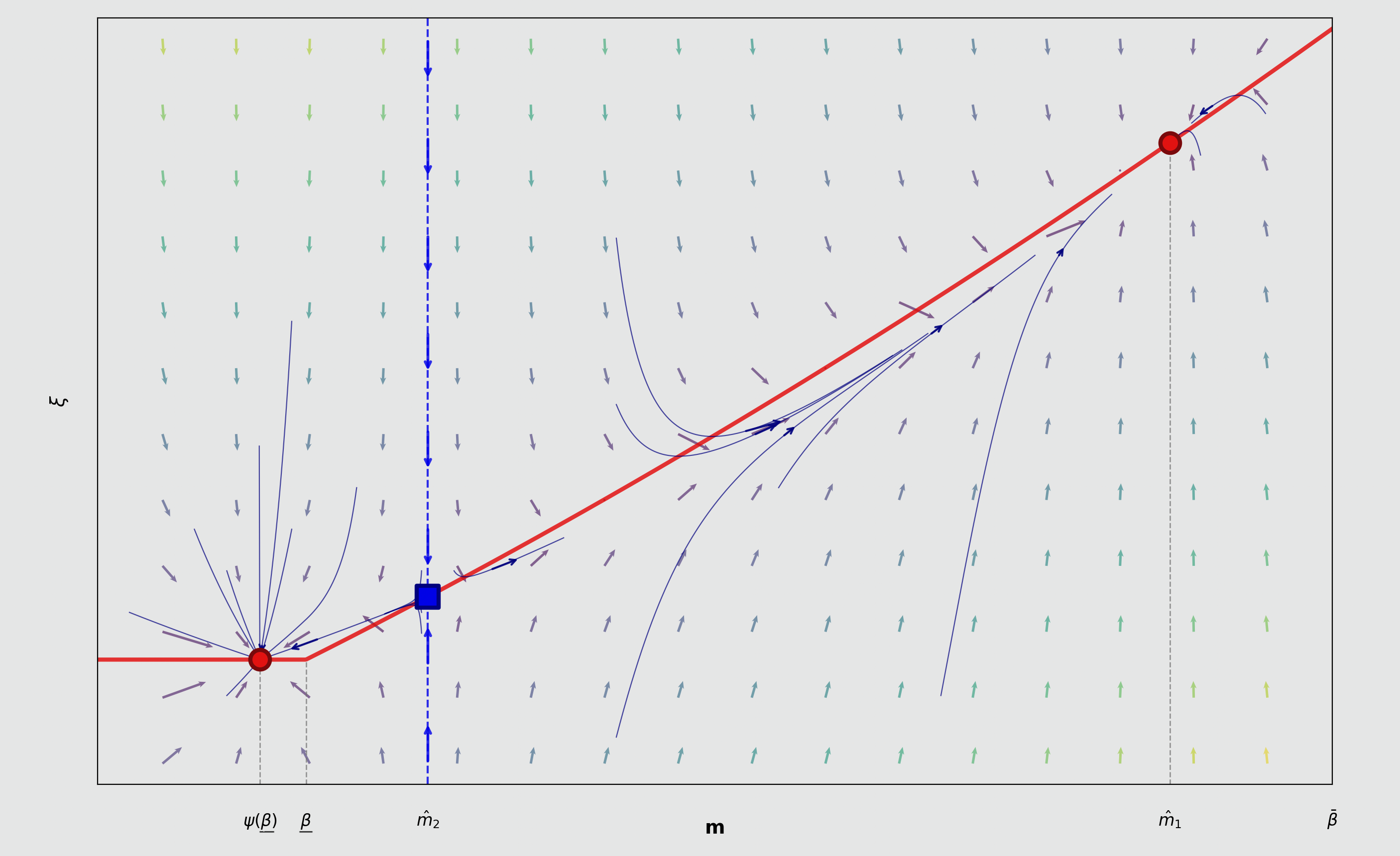} 
    \caption{Linear-quadratic example for a small $\Delta_{\mu}>0$: the red curve separates the areas in which $\dot{\xi}$ is positive and negative, respectively. The blue square is a saddle point, whose attraction basin is the blue dashed line. Every other point lies in the attraction basin of a sink, marked as a red dot. }
    \label{fig:phasediagram}
\end{figure}

We now state the main convergence result of this appendix: 

\begin{theorem}\label{thm_learning_withoutprior}
The process in \eqref{eqn_de} converges a.s. to  the sinks of \eqref{eqn_ode}, and to every sink with positive probability. 
\end{theorem}

Here, we outline the proof, deferring the details to Appendix \ref{sec_learning_proof_withoutprior}. 

We first demonstrate that for large $n$, the sample path $\{\theta_{n+i}\}_{i=1}^k$ over the next $k$ periods---$k$ small relative to $n$---must stay arbitrarily close to the trajectory of the ODE starting from $\theta_n$. This property, referred to as ``shadowing'' in stochastic approximation, ensures that the limiting behavior of $\{\theta_n\}$ can be approximated by the ODE's trajectories.

Given this, we apply Theorem 5.7 of B99 (the ``limit set theorem''), which states that the ODE must be ``chain recurrent'' on the limit set of $\{\theta_n\}$. Roughly, this means that starting from any point in the limit set, we can return to it through a chain of arbitrarily small intermediate steps guided by the ODE. In our setting---which is effectively one-dimensional---chain recurrent points of the ODE coincide with its steady states. It follows that $\{\theta_n\}$ converges a.s. to the OED's steady states. 

We next demonstrate that convergence to any saddle point occurs with zero probability. Roughly speaking, this is because the attraction basin of a saddle point is a vertical line, while shocks to the state---arising from  the randomness in the outcome---come in the horizontal direction. After many shocks, the state escapes a small neighborhood around the saddle point and fails to return with positive probability. Formally, we apply Proposition 5 of MY23, verifying that its conditions are satisfied in our setting.

Combining the above arguments yields a.s. convergence to sinks. To establish positive-probability convergence to any given sink, we apply Theorem 7.3 of B99, the ``attractor theorem.'' Under the full-support assumption, $\theta_n$ belongs to a small neighborhood of the sink with positive probability for every $n$. Starting from there, the trajectory of the ODE is attracted to the sink by definition. Now, since the sample path is shadowed by the ODE's trajectory for large $n$, it, too, converges to the sink.

\subsection{Stochastic approximation: General priors}\label{sec_learning_withprior}
This appendix extends the previous analysis to general priors. 

We first quantify the prior's influence on the updating process, ignoring its impact on players' strategic behavior.  Rewrite \eqref{eqn_mtgeneral} as
\begin{align*}
m_{n+1} - m_n & = \dfrac{1}{n+1} \Bigl(F_{n+1,1}(\theta_n)+b_{n+1,1}(\theta_n)\epsilon_n\Bigr)\\
&\quad +\dfrac{1}{\sum_{i=1}^{n+1} I_i}\left.\dfrac{d\log \bar{p}_0}{d\beta}\right\vert_{\beta=m_{n+1}}-\dfrac{1}{\sum_{i=1}^{n} I_i}\left.\dfrac{d\log \bar{p}_0}{d\beta}\right\vert_{\beta=m_{n}}, 
\end{align*}
where the first line on the right-hand side also appeared in \eqref{eqn_mupdate}, while the second line is new and represents the influence of  the prior.
In Appendix \ref{sec_learning_proof_withprior}, we show that $\{\theta_n\}$ is a.s. bounded.  It follows that  for large $n$, the second line is a.s. bounded by 
\begin{align*}
&\left\vert\dfrac{1}{\sum_{i=1}^{n+1} I_i}\left.\dfrac{d\log \bar{p}_0}{d\beta}\right\vert_{\beta=m_{n+1}}-\dfrac{1}{\sum_{i=1}^{n} I_i}\left.\dfrac{d\log \bar{p}_0}{d\beta}\right\vert_{\beta=m_{n}}\right\vert\\
&=\dfrac{1}{\sum_{i=1}^{n+1} I_i}\left\vert\left.\dfrac{d\log \bar{p}_0}{d\beta}\right\vert_{\beta=m_{n+1}}- \left.\dfrac{d\log \bar{p}_0}{d\beta}\right\vert_{\beta=m_{n}}-\dfrac{I_{n+1}}{\sum_{i=1}^n I_i} \left.\dfrac{d\log \bar{p}_0}{d\beta}\right\vert_{\beta=m_{n}}\right\vert\\
&\leq \frac{1}{(n+1)}\left\vert C \cdot (m_{n+1}-m_n)  +O\left(\frac{1}{n}\right)\right\vert,
\end{align*}
where the constant $C$ depends on $m_n$, $m_{n+1}$, $\overline{p}_0$, and $\ul I$. In addition, we have a.s.
\[|\epsilon_n| = O\left(\sqrt{\log n}\right),\footnote{For each \( n \), \( \epsilon_n \sim {N}(0, 1) \). Standard Gaussian concentration states that 
\[
\mathbb{P}(|\epsilon_n| > \sqrt{2\log n}) \leq \frac{1}{n^2}.
\]
Since the right-hand side is summable, it follows from the Borel–Cantelli lemma that 
\[
|\epsilon_n| \leq \sqrt{2\log n} \quad \text{eventually, a.s.}
\]} \]
and $|F_{n+1,1}(\theta_n)|, |b_{n+1,1}(\theta_n)| = O(1)$. Combining these yields a.s. 
\begin{align*}
|m_{n+1} - m_n| 
&\leq \frac{1}{n+1} \left(\, |F_{n+1,1}(\theta_n)| + |b_{n+1,1}(\theta_n) \epsilon_n|  
+ C \cdot |m_{n+1} - m_n| + O\left(\frac{1}{n}\right) \,\right) \\
&\phantom{=} \hspace{3cm} \Longrightarrow  |m_{n+1} - m_n| = O\left( \frac{\sqrt{\log n}}{n} \right),
\end{align*}
so 
\[m_{n+1}-m_n=\frac{1}{n+1}\left(F_{n+1,1}(\theta_n)+b_{n+1}(\theta_n)\epsilon_n+O\left(\frac{\sqrt{\log n}}{n}\right)\right).\]

A similar derivation applies to $\xi$. After accounting for the prior's influence, \eqref{eqn_vtgeneral} a.s. becomes

\begin{align*}
\xi_{n+1} - \xi_n 
= \frac{1}{n+1} \left( I_{n+1} - \xi_n + O\left(\frac{\sqrt{\log n}}{n}\right) \right).
\end{align*}
In vector form, we have a.s. 
\begin{equation}\label{eqn_de_generalprior}
\theta_{n+1} - \theta_n = \frac{1}{n+1} \left( F_{n+1}(\theta_n)+ b_{n+1}(\theta_n)\epsilon_n+ O\left(\frac{\sqrt{\log n}}{n}\right) \right),
\end{equation}
where $F_{n+1}$ and $b_{n+1}$ are defined in \eqref{eqn_de}. Compared to the uniform-prior case, now there is an approximation error of $O(\sqrt{\log n}/n)$ in the bracketed term. As it turns out, this error vanishes sufficiently fast as $n$ grows to infinity, so that the convergence results established under uniform prior continue to hold.

Certainly, the prior also influences players' strategic behavior and in turn belief updating. However, the total effect is still too small to overturn the results established under uniform prior. Indeed, now with the freedom to choose the prior, we can concentrate it around a sink of the ODE. By Bena{\"\i}m's (1999) ``attractor theorem,'' converge to this sink occurs with probability arbitrarily close to one. 

\subsection{Discussion}\label{sec_learning_discussion}
\paragraph{Assumption \ref{assm_linearr}.} Assumption \ref{assm_linearr} is not innocuous when viewed globally. When \( r(a, \beta) = \beta a \) (as in Section \ref{sec_example}), the assumption holds if, say,  \( c(a) = a^\gamma/\gamma \) with \( \gamma >1 \), so that 
\[
a(h, \beta) = (h\beta)^{\frac{1}{\gamma-1}} \Longrightarrow R(h, \beta) = h^{\frac{1}{\gamma-1}} \beta^{\frac{\gamma}{\gamma-1}}.
\]

However, Assumption \ref{assm_linearr} always holds locally. This motivates an alternative approach that dispense with the assumption altogether: begin with a prior concentrated around a locally asymptotically stable steady state, and approximate \( R(h, \beta) \) linearly in \( \beta \) near that steady state. Over time, the posterior diverges from the prior. However, using posterior concentration results in Bayesian statistics (e.g., \citealt{vandervaart2008gp}), we establish that with high probability, posterior beliefs remain within an \( O(1/\sqrt{n}) \) neighborhood of the steady state as \( n \to \infty \). Within this neighborhood, the linear approximation remains valid, which suffices to establish positive-probability convergence to the steady state. 

Compared to the approach developed in this paper, this alternative approach yields only positive-probability convergence rather than almost-sure convergence. Its advantage is that it avoids imposing additional assumptions on the primitives.  Given the scope of this paper, we leave the full development of the alternative approach to another paper.

\vspace{-10pt}

\paragraph{Forward-looking evaluator.} Now suppose that E is forward-looking and maximizes the net present value of the subjective payoffs from conducting assessments, discounted at rate $\delta \in (0,1)$. A and M  remain myopic. Under this model variant, A's effort choice problem remains unchanged, as does the update equation \eqref{eqn_de}, given assessment levels. E's problem changes: compared to the myopic case, today's assessment now affects future payoffs, introducing an experimentation motive. 

We impose two restrictions on the evaluator. First, the assessment space is a compact interval $[\ul h, \bar h]$ satisfying $[h(\ul\beta), h(\bar\beta)]\subseteq [\ul h, \bar h] \subset (0,1)$ (again $h(\cdot)$ denotes the assessment map in the steady-state model). This assumption ensures that even with an experimentation motive, assessment, and hence Fisher information, is always positive.  Second, we restrict E to using Markovian strategies, which depend only on the calendar time $n+1$ and the state variable $\theta_n$. Formally, $\rho: \mathbb{N} \times M \rightarrow \Delta ([\ul h, \bar h])$, where for each $n \in \mathbb{N}_0$, $\rho(n+1,\theta_n)$ specifies the distribution over assessment in period $n+1$, conditional on the period-$n$ state $\theta_n$.

Initially, the experimentation motive can complicate the analysis. For example, E's objective may not be concave in $h$ when $n$ is small. As \( n \) grows large, however, the problem becomes more regular. Under the Markovian assumption, the envelope theorem implies that the period-$n$ assessment, $h_n$, influences future payoffs only through  the state variable \( \theta_n \). By the update equation \eqref{eqn_de}, this influence vanishes as \( n \to \infty \), given that 
assessment and hence Fisher information is always positive. For large \( n \), E's objective becomes dominated by the current-period payoff, given in \eqref{eqn_eproblem_learning}. Given this, all desired properties---such as the uniqueness of optimal assessment---follow immediately. Compared to the myopic case, the only change lies in the exact assessment intensities before the system settles, which affects neither the convergence results nor the  steady states.

\section{Learning under color-sighted impartial assessment}\label{sec_learning_aa}
In this appendix, we examine the misspecified Bayesian learning foundation for the Berk--NE under the color-sighted, impartial assessment studied in Section \ref{sec_aa}. 

There are $J \geq 2$ groups of agents. At the outset, $\omega_j^*=(\mu_j^*, \beta_j^*)$ is revealed once and for all to group $j$, $j=1,\cdots, J$. Society knows this structure but not the exact value of $\bm{\omega}=(\bm{\mu}, \bm{\beta})$, where $\bm{\mu}=(\mu_j)_{j=1}^J$ and $\bm{\beta}=(\beta_j)_{j=1}^J$. Society's prior is degenerate at $\hat{\mu}_j \neq \mu_j^*$ along each $\mu_j$-argument. Along $\bm{\beta}$, the prior admits a density $p_0(\bm{\beta})=\prod_{j=1}^J p_{0,j}(\beta_j)$, where each $p_{0,j}$ is smooth with full support on $[\ul\beta,\bar\beta]$.

In every period $n \in \mathbb{N}$, the following events unfold: 
\begin{enumerate}
    \item E commits to assess all groups at the same intensity $h_n \in [0,1]$. 
    \item Groups make hidden effort choices $a_{j, n} \geq 0$, $j=1,\cdots, J$ simultaneously. 
    \item The outcome of group $j$ is
    \[X_{j,n} = \eta_{j,n}+r(a_{j,n}, \beta_j^{*})+\varepsilon_{j,n},\]
    where $\eta_{j,n} \sim N(\mu_j^*, 1)$ is  ability and $\varepsilon_{j,n} \sim N(0,h_n^{-1}-1)$ is noise, drawn independently of each other, across groups, and over time. Only outcomes are publicly realized. 
    \item M rewards each group $j$ for the expected value of $\eta_{j,n}+v_{\mathrm{M}}(a_{j,n}, \beta_j)$. 
    \item Society updates beliefs about $\bm{\beta}$. 
\end{enumerate}

All players myopically maximize their expected payoffs period by period. In period $n$, group $j$’s true effort is $a(h_n,\beta_j^*)$, while their perceived effort (as a random variable of $\beta_j$) is $a(h_n,\beta_j)$. Let $H_n=(h_i,(X_{j,i})_{i=1}^n)$ denote the public history up to period $n$, and let $p_{j,n}(\cdot \mid H_n)$ denote the posterior density of $\beta_j$ formed in Stage 5 of period $n$, conditional on $H_n$. In Stage 1 of period $n$, the evaluator chooses the optimal assessment $h_n$ by solving 
\begin{equation}\label{eqn_eproblem_learning_aa}
\max_{h \in [0,1]} \; \sum_{j=1}^J \alpha_j \int_{\ul\beta}^{\bar\beta} V_{\mathrm{E}}(h,\beta_j)\, p_{j,n-1}(\beta_j \mid H_{n-1}) \, d\beta_j - \kappa(h).
\end{equation}
The objective represents the expected payoff from assessing all groups equally, weighted by their population, net of assessment cost.

Turning to the belief updating process: in Stage 5 of period $n$, the posterior density of $\beta_j$ conditional on $H_n$ is
\[
p_{j,n}(\beta_j \mid H_n) = 
\frac{
    \exp\left(
        \displaystyle -\sum_{i=1}^n 
        \frac{(X_{j,i} - \hat{\mu}_j - R(h_i,\beta_j))^2}{2 h_i^{-1}}
    \right) 
    p_{j,0}(\beta_j)
}{
    \displaystyle \int_{\ul\beta}^{\bar\beta} 
    \exp\left(
        \displaystyle -\sum_{i=1}^n 
        \frac{(X_{j, i} - \hat{\mu}_j - R(h_i,\tilde{\beta}_j))^2}{2 h_i^{-1}}
    \right) 
    p_{j,0}(\tilde{\beta}_j) \, d\tilde{\beta}_j}.
\]
The joint density $p_n(\bm{\beta}\mid H_n)$ is the product of the marginal densities across groups.

\begin{theorem}\label{thm_learning_aa}
Let $p_0$ be as above, and suppose that Assumptions \ref{assm_primitives}--\ref{assm_linearr} hold. As $n \rightarrow \infty$, society's posterior beliefs about $\bm{\beta}$ converge to the Berk--NE belief vector $\hat{\bm{\beta}}$ in Theorem \ref{thm_colorsighted} with positive probability, and with probability arbitrarily close to one given a sufficiently concentrated prior around $\hat{\bm{\beta}}$.  
\end{theorem}

In contrast to Theorem \ref{thm_learning}, here we cannot establish a.s. convergence to the Berk--NE belief vector, due to the multidimensionality of the belief space. Beyond the one-dimensional case, chain recurrent sets of the limiting ODE may include more complex objects---such as limit cycles---and are more difficult to characterize. Given this inherent challenge, positive-probability convergence to each locally asymptotically stable equilibrium---as established in Theorem \ref{thm_learning}---is the best one can hope for.

\section{Proofs for Appendices \ref{sec_learning} and \ref{sec_learning_aa}}\label{sec_learning_proof}

\subsection{Proof of Theorem \ref{thm_learning_withoutprior}}\label{sec_learning_proof_withoutprior}
We follow the notation convention of B99 as closely as possible. In particular, we use $n$ to  denote a generic nonnegative integer in $\mathbb{N}_0=\mathbb{N}\cup\{0\}$. Set $H_0=\emptyset$, $\theta_0=0$, and $h_1(\theta_0)=h_1$, which determine the initial values of $F_1(\theta_0)$ and $F(\theta_0)$. Equation~\eqref{eqn_de} can then be rewritten as  
\[
\theta_{n+1}-\theta_n
= \underbrace{\gamma_{n+1}}_{=\frac{1}{n+1}}
\left(
   F(\theta_n)
   + \underbrace{U_{n+1}}_{=b_{n+1}(\theta_n)\epsilon_{n+1}}
   + \underbrace{d_{n+1}}_{=F_{n+1}(\theta_n)-F(\theta_n)}
\right),
\quad n\in\mathbb{N}_0.
\]
In stochastic approximation, $\gamma_{n+1}$ is referred to as the \emph{step size}; $U_{n+1}$ is a \emph{martingale difference noise}, adapted to $H_n$ (hence $\mathbb{E}[U_{n+1}\mid H_n]=0$); $d_{n+1}$ is an \emph{approximation error}, again adapted to $H_n$.

The following conditions are required to apply existing results in stochastic approximation to our setting.

\begin{condition}\label{lem_stepsize}
$\sum_n \gamma_{n+1}=\infty$ and $\sum_n \gamma_{n+1}^2<\infty$. 
\end{condition}

\begin{proof}
    This follows from the definition of $\gamma_{n+1}$.
\end{proof}

\begin{condition}\label{lem_flipschitz}
    $F: M \rightarrow \mathbb{R}^2$, $\theta\mapsto F(\theta)$ is Lipschitz.
\end{condition}

\begin{proof}
    This follows from the definition of $F$. 
\end{proof}

\begin{condition}\label{lem_variance}
$\{U_{n+1}\}$ is subgaussian with $\mathbb{E}[\|U_{n+1}\|^2] \in [\ul \Gamma, \bar\Gamma]$, where $\ul\Gamma\coloneqq \ul I/(\bar I)^2$ and $\bar\Gamma\coloneqq  \bar I/(\ul I)^2$. 
\end{condition}

\begin{proof}
Because $U_{n+1}=[b_{n+1, 1}(\theta_n), 0]^{\top} \epsilon_{n+1}$, where $\epsilon_{n+1} \sim N(0,1)$, and 
$b_{n+1, 1}(\theta) \in [\sqrt{\ul\Gamma},  \sqrt{\bar\Gamma}]$ for all $n \in \mathbb{N}_0$ and $\theta \in M$. 
\end{proof}

\begin{condition}\label{lem_mbound}
$\{\theta_{n+1}\}$ is a.s. bounded. 
\end{condition}

\begin{condition}\label{lem_fdifference}
$d_{n+1}=F_{n+1}(\theta_n)-F(\theta_n) \rightarrow 0$ a.s. 
\end{condition}

\begin{condition}\label{lem_localbound}
For any saddle point $\hat{\theta}$ of the ODE, there exists a neighborhood $U$ of $\hat{\theta}$ in $M$ such that $F_{n+1}(\theta)-F(\theta)= O(1/n)$ across $\theta \in U$. 
\end{condition}

Proofs of Conditions \ref{lem_mbound}, \ref{lem_fdifference}, and \ref{lem_localbound} build on the next lemma. 

\begin{lemma}\label{lem_hbound}
    For sufficiently large $n$:
    \begin{enumerate}[(i)]
        \item $h_{n+1}(\theta)-h(\tilde{m})= O(1/\sqrt{n})$ uniformly over $\theta \in M$.
        \item For any $m \in (\ul\beta,\bar\beta)$, there exists a neighborhood of $U$ of $m$ such that the bound in Part (i) can be improved to  $O(1/n)$ across $\theta \in U \times [\ul I, \bar I]$. 
    \end{enumerate}
\end{lemma}

\begin{proof}
Below is a useful fact about truncated normal distributions (see, e.g., Lemma 1 of MY23). 

\begin{fact}\label{fact_mbound}
Let $\nu_n(\theta)$ denote the mean of $\tilde{N}(m, 1/(n\xi))$. There exists a constant \( C > 0 \) such that  
\[
\left| \nu_{n}(\theta) - \tilde{m} \right| < \frac{C}{\sqrt{n}} \ \text{ for all } n \in \mathbb{N} \text{ and } \theta \in M. 
\]
Furthermore, for any \( m' \in (\ul{\beta}, \bar{\beta}) \), there exists a neighborhood \( U' \) of \( m' \) and a constant \( C' > 0 \) (depending on $U'$)  such that 
\[
\left| \nu_{n}(\theta) - m \right| < \frac{C'}{n} \ \text{ for all }  n \in \mathbb{N} \text{ and }\theta \in U' \times [\ul I, \bar I]. 
\]
\end{fact}

First, we fix any $\theta=(m,\xi)\in M$ and bound 
\[
\left\lvert \, 
\mathbb{E}_{\beta \sim \tilde{N}\!\left(m,\,1/(n\xi)\right)}\!\left[\frac{\partial V_{\mathrm{E}}}{\partial h}(h,\beta)\right] 
- \frac{\partial V_{\mathrm{E}}}{\partial h}(h,\tilde{m}_n) 
\right\rvert
\quad \text{for large } n.
\]
A second-order Taylor expansion of $\tfrac{\partial V_{\mathrm{E}}}{\partial h}(h,\beta)$ around $\beta=\nu_n(\theta)$ gives
\begin{align*}
\frac{\partial V_{\mathrm{E}}}{\partial h}(h,\beta)
&= \frac{\partial V_{\mathrm{E}}}{\partial h}(h,\nu_n(\theta)) 
  + \frac{\partial^2 V_{\mathrm{E}}}{\partial h \,\partial \beta}(h,\nu_n(\theta)) \, (\beta-\nu_n(\theta)) \\
&\quad + \frac{1}{2} \frac{\partial^3 V_{\mathrm{E}}}{\partial h \,\partial \beta^2}(h,z(\beta)) \, (\beta-\nu_n(\theta))^2,
\end{align*}
where $z(\beta)$ lies between $\beta$ and $\nu_n(\theta)$ and so belongs to $[\ul\beta,\bar\beta]$.
Taking expectation over $\beta \sim \tilde{N}(m, 1/(n\xi))$ yields 
\begin{align*}
\mathbb{E}_{\beta \sim \tilde{N}(m, 1/(n\xi))}\left[\frac{\partial V_{\mathrm{E}}}{\partial h}\right]
&= \frac{\partial V_{\mathrm{E}}}{\partial h}(h, \nu_n(\theta)) + 0 
+ \frac{1}{2} \sup_{z \in [\ul\beta, \bar\beta],\, h \in [0,1]} \left\vert\frac{\partial^3 V_{\mathrm{E}}}{\partial h \partial \beta^2}(h, z)\right\vert \; \mathbb{V}_n(\beta) \\
&= \frac{\partial V_{\mathrm{E}}}{\partial h}(h, \nu_n(\theta)) + O\left(\frac{1}{n}\right),
\end{align*}
where the second equality uses the variance bound $\mathbb{V}_n(\beta)\leq 1/(n\ul\xi)= O(1/n)$. It follows that
\begin{align*}
&\left\vert\mathbb{E}_{\beta \sim \tilde{N}\left(m, 1/(n\xi)\right)}\left[\frac{\partial V_{\mathrm{E}}}{\partial h}\right] - \frac{\partial V_{\mathrm{E}}}{\partial h} (h, \tilde{m})\right\vert \\
\leq & \left\vert\mathbb{E}_{\beta \sim \tilde{N}\left(m, 1/(n\xi)\right)}\left[\frac{\partial V_{\mathrm{E}}}{\partial h}\right] - \frac{\partial V_{\mathrm{E}}}{\partial h} (h, \nu_n(\theta))\right\vert+ \left\vert \frac{\partial V_{\mathrm{E}}}{\partial h} (h, \nu_n(\theta)) -\frac{\partial V_{\mathrm{E}}}{\partial h} (h, \tilde{m}) \right\vert\\
\leq & \ O\left(\frac{1}{n}\right)+ \sup_{\beta \in [\ul\beta, \bar\beta],\, h \in [0,1]} \left| \frac{\partial^2 V_{\mathrm{E}}}{\partial h \partial \beta} \right| \; \left| \nu_n(\theta) - \tilde{m} \right|.
\end{align*}
By Fact \ref{fact_mbound}, the last line is $O\left(1/\sqrt{n}\right)$ uniformly over $M$, and $O(1/n)$ in a neighborhood of $\theta$ if $m \in (\ul\beta,\bar\beta)$. 

Next, recall that $h_{n+1}(\theta)$ and $h(\tilde{m})$ are determined by the first-order conditions 
\[\mathbb{E}_{\beta \sim \tilde{N}\left(m, 1/(n\xi)\right)}\left[\frac{\partial V_{\mathrm{E}}}{\partial h}\right]=\kappa'(h)\quad \text{and} \quad   \frac{\partial V_{\mathrm{E}}}{\partial h}(h,\tilde{m})=\kappa'(h). \]
Since both $V_{\mathrm{E}}$ and $\kappa$ are smooth, the bounds established above extend to $h_{n+1}(\theta)-h(\tilde{m})$. 
\end{proof}

\vspace{-10pt}

\paragraph{Proof of Condition \ref{lem_mbound}.} Since $\{\xi_n\}$ is bounded in $[\ul I, \bar I]$, it suffices to show that $\{m_n\}$ is a.s. bounded. Denote $C_{n+1}=n\xi_{n}/I_{n+1}$ and rewrite \eqref{eqn_mupdate} as 
\begin{align*}
m_{n+1}=\frac{C_{n+1} m_n + \psi(\tilde{m}_n)}{1+C_{n+1}}+\frac{U_{n+1}}{n+1}+\frac{(\beta\circ h_{n+1})(\theta_n)-(\beta \circ h)(\tilde{m}_n)}{1+C_{n+1}}.
\end{align*}
Since $C_{n+1}= \Theta(n)$ and $h_{n+1}(\theta_n)-h(\tilde{m}_n)= O(1/\sqrt{n})$ when $n$ is large, the third term on the right-hand side is of $o(1/n)$. Omitting this term, we obtain 
\[
m_{n+1}\leq \frac{C_{n+1} m_n +\bar \psi}{1+C_{n+1}}   +\frac{U_{n+1}}{n+1}, \quad \text{where } \bar\psi=\sup_{\beta \in [\ul\beta,\bar\beta]}\psi. 
\]
Consider a new process $\{m_n'\}$ defined by
\[
     m'_{n+1}= \frac{C_{n+1} m_n' +\bar\psi}{1+C_{n+1}}   +\frac{U_{n+1}}{n+1}.
\]
If $m_n \leq m'_n$, then
\[
m_{n+1}  \leq  \frac{C_{n+1} m_n +\bar\psi}{1+C_{n+1}}   +\frac{U_{n+1}}{n+1} \leq \frac{C_{n+1} m'_n +\bar\psi}{1+C_{n+1}}   +\frac{U_{n+1}}{n+1}=m'_{n+1}. 
\]
Applying this argument iteratively shows that for all $k \geq 1$, 
\begin{align*}    
 m_{n+k}  & \leq   m'_{n+k}  \\
& = \left(\prod_{j=1}^{k} \frac{C_{n+j}}{1+C_{n+j}}\right)  m'_n +\left(1- \left(\prod_{j=1}^{k} \frac{C_{n+j}}{1+C_{n+j}}\right) \right)\bar\psi + O_p\left( \sqrt{\sum_{j=1}^{k} \frac{1}{(n+j)^2}}\right)\\
  & \leq \max\{m'_n, \bar\psi \} + O_p\left( \sqrt{\sum_{j=1}^{k} \frac{1}{(n+j)^2}}\right), 
\end{align*}
where the equality holds up to $O_p\left( \sqrt{\sum_{j=1}^{k} \frac{1}{(n+j)^2}}\right)$ by Condition  \ref{lem_variance}. Since the term $\sqrt{\cdot}$ remains finite as $k \rightarrow \infty$, the last line is a.s. bounded from above by 
Kolmogorov’s Strong Law of Large Numbers. A lower bound can be derived analogously. \qed 

\vspace{-10pt}

\paragraph{Proof of Condition \ref{lem_fdifference}.} Evaluate $d_{n+1}$ using Lemma \ref{lem_hbound} and Condition \ref{lem_mbound}. \qed
 \vspace{-10pt} 

 \paragraph{Proof of Condition \ref{lem_localbound}.} Recall that any saddle point $\hat{\theta}=(\hat{m}, \hat{\xi})$ satisfies $\hat{m} \in (\ul\beta,\bar\beta)$. Then apply Lemma \ref{lem_hbound} to bound $F_{n+1}(\theta)-F(\theta)$.   \qed
 \vspace{-10pt}
 
\paragraph{Proof of Theorem \ref{thm_learning_withoutprior}.} 
The proof has been outlined in Appendix \ref{sec_learning_withoutprior}. Here, we provide the omitted details. 
\vspace{-10pt}

\subparagraph{Preliminaries.} The following are taken from B99.

Define $\tau_0=0$ and $\tau_n=\sum_{i=1}^n \gamma_i$ for $n \in \mathbb{N}$. The inverse of the map 
$n \rightarrow \tau_n$ is defined by 
\[\tau^{-1}(t)=\sup\{k \geq 0: t \geq \tau_{k}\}, \quad  t \geq 0. \]
By Condition \ref{lem_stepsize}, $\lim_{t \rightarrow \infty}\tau^{-1}(t)=\infty$. 

Define the continuous time affine interpolated process $\theta:  \mathbb{R}_+\rightarrow M$ by 
\[\theta(\tau_n+s)=\theta_n+s \frac{\theta_{n+1}-\theta_n}{\tau_{n+1}-\tau_n}, \quad  n \in \mathbb{N}_0 \text{ and } 0 \leq s < \gamma_{n+1}. \]
Define the \emph{limit set} of $\theta: \mathbb{R}_+\rightarrow M$ by
\[L(\theta)=\bigcap_{t \geq 0}\bar{\theta([t,\infty))}.\]
A point $\theta \in M$ is \emph{attainable} by $\theta: \mathbb{R}_+\rightarrow M$ if for all $t>0$ and neighborhood $U$ of $\theta$,  
\[\mathbb{P}(\exists s \geq t: \theta(s) \in U)>0.\]
The set of attainable points by $\theta: \mathbb{R}_+ \rightarrow M$ is denoted by $Att(\theta)$.

Let $\Phi: \mathbb{R}_+ \times M \rightarrow M$, $(t,\theta)\mapsto \Phi(t,\theta)=\Phi_t(\theta)$,  denote the \emph{semiflow} induced by $F$, such that  \[\Phi_l(\theta(t))=\theta(t)+\int_t^{t+l} F(\theta(s)) ds\] 
for all $(t, l) \in \mathbb{R}_+ \times \mathbb{R}_{+}$. 
$\theta: \mathbb{R}_+ \rightarrow M$ is an \emph{asymptotic pseudotrajectory} for $\Phi$ if 
\[\lim_{t \rightarrow \infty} \sup_{0 \leq l \leq T} \|\theta(t+l), \Phi_{l}(\theta(t))\|=0 \ \text{ for all }  T>0.\]
This means that for sufficiently large $t$, the curve $[0, T] \rightarrow M$, $l\mapsto \theta(t+l)$ shadows the $\Phi$-trajectory of the point $\theta(t)$ over the interval $[0,T]$ with arbitrary accuracy. 

 We wish to demonstrate that $\theta: \mathbb{R}_+\rightarrow M$ is an asymptotic pseudotrajectory of $\Phi$. To this end,  define two continuous time processes $\bar U: \mathbb{R}_+ \rightarrow \mathbb{R}^2$ and $\bar\gamma: \mathbb{R}_+ \rightarrow \mathbb{R}$ such that 
\[\bar U(\tau_n+s)=U_{n+1}, \quad   \bar\gamma(\tau_n+s)=\gamma_{n+1}\]
for all $n \in \mathbb{N}\cup \{0\}$ and $0 \leq s < \gamma_{n+1}$. Define 
\[\Delta(t,T)=\|\sup_{0\leq l \leq T}\int_{t}^{t+l} \bar{U}(s)ds \|\]
for $t, T \geq 0.$ Intuitively, $\Delta(t,T)$ represents the cumulative martingale difference noises, weighted by step sizes, over $[t, t+T]$. 

A point $a \in M$ is \emph{chain recurrent} if for all $\delta, T>0$, there exists a \emph{$(\delta, T)$-pseudo-orbit} from $a$ to $a$. That is, a finite sequence of partial trajectories 
\begin{align*}
\{\Phi_t(y_i): 0 \leq t \leq t_i\}; & \ i=0, \cdots, k-1; \ t_i \geq T 
\end{align*}
\text{such that } 
\begin{align*}
d(y_0, a)&<\delta;\\
d(\Phi_{t_j}(y_j), y_{j+1})& <\delta, \ j=0,\cdots, k-1;\\
y_k& =a.
\end{align*}
For a nonempty invariant set $M'$ of $\Phi$, we say that $\Phi$ is \emph{chain recurrent on $M'$} if every point in $M'$ is chain recurrent for $\Phi$ restricted to $M'$. A compact invariant set on which $\Phi$ is chain recurrent is called \emph{internally chain recurrent}. In our setting, every invariant set is contained in the steady states. Since these are finite and chain recurrent (by definition), the three concepts are equivalent: invariant points, internally chain recurrent points, and steady states.

\vspace{-10pt}
\subparagraph{Step 1.} We first demonstrate that $\{\theta_n\}$ a.s. converges to the steady states of $\Phi$. 
The proof builds on Proposition 4.1 of B99, which requires three conditions: 
\begin{enumerate}
    \item[A1] $\lim_{t \rightarrow \infty} \Delta(t,T)=0$ for all  $T>0$.
    \item[A2] $\sup_n\|\theta_n\|<\infty$.
    \item[A2'] $F$ is Lipschitz and bounded on a neighborhood of $\{\theta_n: n \geq 0\}$. 
\end{enumerate}

Proposition 4.4 of B99 shows that under Conditions \ref{lem_stepsize}, \ref{lem_variance}, and \ref{lem_fdifference}, A1 holds a.s. Indeed, for all $\alpha, t, T>0$, 
\begin{equation}\label{eqn_deltabound}
\mathbb{P}(\Delta(t,T) \geq \alpha)\leq C\exp\left(-\frac{\alpha^2}{ C' T \bar\gamma(u)}\right)
\end{equation}
holds for some $t \leq u \leq t+T$ and some positive constants $C$ and $C'$ that depend on $\dim(M)$ ($=2$) and $\bar\Gamma$.  Meanwhile, A2 holds a.s. by Condition \ref{lem_mbound}, and A2' holds a.s. by Conditions \ref{lem_flipschitz} and \ref{lem_mbound}. Proposition 4.1 of B99 states that under A1 and either A2 or A2', $\theta: \mathbb{R}_+ \rightarrow M$ is an asymptotic pseudotrajectory of $\Phi$. Furthermore, A2' implies that for sufficiently large $t$: 
\begin{equation}\label{eqn_pseudo}
\sup_{0 \leq l \leq T} \|\theta(t+l)-\Phi_l(\theta(t))\| \leq C(T)[\Delta(t-1, T+1)+\sup_{t\leq s \leq t+l}\bar\gamma(s)],
\end{equation}
where $C(T)$ is a constant that depends only on $T$ and $F$. 

We characterize the limit set $L(\theta)$ by applying Theorem 5.7 of B99 (the ``limit set theorem''). Since $\theta: \mathbb{R}_+ \rightarrow M$ is an asymptotic pseudotrajectory of $\Phi$, $L(\theta)$ is internally chain recurrent if the image of $\theta$ has a compact closure in $M$ --- an assumtpion that holds a.s. under Condition \ref{lem_mbound}. Since internally chain recurrent points and steady states are equivalent in our setting (established in the preliminaries), it follows that $L(\theta)$ is contained in the steady states of $\Phi$.

\vspace{-10pt}
\subparagraph{Step 2.} We next demonstrate that $\{\theta_n\}$ converges to any saddle point $\hat{\theta}=(\hat{m}, \hat{\xi})$ of $\Phi$ with zero probability. The proof builds on Proposition 5 of MY23,\footnote{The classical reference of non-convergence is P90, which shows that a stochastic approximation process almost surely avoids linearly unstable steady states of the limiting ODE, assuming bounded martingale differences. Proposition 5 of MY23 establishes this result for Gaussian martingale differences. } which requires, in addition to Conditions  \ref{lem_stepsize}-\ref{lem_fdifference}, Condition~\ref{lem_localbound}: the approximation error not only vanishes but at rate $O(1/n)$. As we shall later demonstrate, this rate condition can be slightly weakened.

Additionally, the proposition requires that shocks to the stochastic approximation process come in directions orthogonal to the attraction basin of $\hat{\theta}$.  Formally, for the eigenvalues of the Jacobian matrix $\left.\nabla F\right\vert_{\theta=\hat{\theta}}$ with negative real parts, let $H$ be the affine space spanned by the corresponding eigenvectors, and $H^*$ be the set of all unit vectors orthogonal to $H$. There exist $k>0$, $n^* \in \mathbb{N}_0$, and a neighborhood $U$ of $\hat{\theta}$ such that $|b_{n+1}(\theta)\cdot h|\geq k$ for all $h \in H^*$, $n \geq n^*$, and $\theta \in U$. 

In our case, we have $\hat{m} \in (\ul\beta,\bar\beta)$ and hence $\psi(\hat{m})=\hat{m}$. Moreover, $\hat{\xi}=(I \circ h)(\hat{m})$.  Linearizing $F$ around $\hat{\theta}$ yields 
\begin{align*}
&\left.\nabla F \right \vert_{\theta=\hat{\theta}} \\
& =\left.\begin{bmatrix}
\dfrac{(I\circ h)'(m)\cdot (\psi(m)-m)+(I\circ h)(m)\cdot (\psi'(m)-1)}{\xi} & -\dfrac{(I\circ h)(m)\cdot(\psi(m)-m)}{\xi^2}\\
(I\circ h)'(m) & -1
\end{bmatrix}\right\vert_{\theta=\hat{\theta}}\\
& = \begin{bmatrix}
\psi'(\hat{m})-1 & 0\\
(I\circ h)'(\hat{m}) & -1
\end{bmatrix}.
\end{align*}
The eigenvalues of $\left.\nabla F\right\vert_{\theta=\hat{\theta}}$ are its diagonal values: $\psi'(\hat{m})-1$ and $-1$. The first is positive ($\psi$ crosses the 45 degree line from below at $\hat{m}$). The second is negative, and the corresponding eigenvectors take the form of $[0, y]^{\top}$. In the above condition, setting $H=\{[0,y]: y \in \mathbb{R}\}$, $H^*=\{[1,0]^{\top}\}$, and $k=\sqrt{\ul\Gamma}$ verifies that 
 \[ b_{n+1}(\theta)\cdot [1,0]^{\top}=b_{n+1,1}(\theta) \geq \sqrt{\ul\Gamma}=k  \ \text{ for all  } \theta \in M.\]
 
With the above, we can readily apply Proposition 5 of MY23 to establish zero-probability convergence to $\hat{\theta}$. 

 \vspace{-10pt}
 
\subparagraph{Step 3.} Lastly, we demonstrate that $\{\theta_n\}$ converges to every sink of $\Phi$ with positive probability. Let $\mathcal{F}_t$ denote the filtration generated by $\theta: \mathbb{R}_+ \rightarrow \mathbb{R}^2$. Fix any $\delta, T>0$. For sufficiently large $t$,  
\begin{align*}
&\mathbb{P}\left(\sup_{0\leq l\leq T} \|\theta(t+l)-\Phi_l(\theta(t)) \|\geq \delta \mid \mathcal{F}_{t}\right) \\
\tag{$\because$ \eqref{eqn_pseudo}}\leq &\quad \mathbb{P}\left(C(T)[\Delta(t-1,T+1)+\sup_{t \leq s \leq t+l}\bar\gamma(s)] \geq \delta\right)\\
\tag{$\because$ $\bar\gamma(t)$ is decreasing in $t$}=& \quad \mathbb{P}\left(\Delta(t-1,T+1) \geq \frac{\delta}{C(T)}-\bar\gamma(t) \right)\\
\tag{$\because$ \eqref{eqn_deltabound}}\leq & \quad \sup_{t \leq u \leq t+T}C\exp\left(-\frac{(\delta/C(T)-\bar\gamma(t))^2}{C' T \bar\gamma(u)}\right)\\
\tag{$\because$ $\bar\gamma(t)$ is decreasing in $t$}= & \quad C\exp\left(-\frac{(\delta/C(T)-\bar\gamma(t))^2}{C' T \bar\gamma(t)}\right)
\end{align*}
Denote the upper bound on the last line by $w(t,\delta, T)$ and note that it vanishes as $t \rightarrow \infty$. This verifies that Condition (24) of B99 is satisfied. 

Consider any sink $\hat{\theta}$ with attraction basin $B(\hat{\theta})$. Let $N_{\epsilon}(\hat{\theta})$ denote the $\epsilon$-neighborhood of $\hat{\theta}$ in $M$. For sufficiently small $\epsilon$, we have $N_{\epsilon}(\hat{\theta})\subset B(\hat{\theta})$. Under the full-support assumption, $\mathbb{P}(\theta_n \in N_{\epsilon/2}(\hat{\theta}))>0$ for any $n \in \mathbb{N}$. For sufficiently large $t=\tau^{-1}(n)$ such that $\theta(t)=\theta_n \in N_{\epsilon/2}(\hat{\theta}) \subset B(\hat{\theta})$ (the first equality follows from the definition of $\tau$), the $\Phi$-trajectory of the point $\theta(t)$ is attracted to $\hat{\theta}$. Hence, 
\[\|\Phi_l(\theta(t))-\hat{\theta}\| \leq \|\theta(t)-\hat{\theta}\|< \frac{\epsilon}{2} \ \text{ for all } l>0. \]
Meanwhile, \eqref{eqn_pseudo} implies that 
\[\sup_{0\leq l\leq T} \|\theta(t+l)-\Phi_l(\theta(t))\|<\frac{\epsilon}{2} \]
when $t$ is large. Combining these inequalities yields 
\[\sup_{0 \leq l \leq T}\|\theta_{t+l}-\hat{\theta}\|\leq \sup_{0 \leq l \leq T}\|\theta_{t+l}-\Phi_l(\theta(t))\|+\sup_{0 \leq l \leq T}\|\Phi_l(\theta(t))-\hat{\theta}\|\leq \frac{\epsilon}{2}+\frac{\epsilon}{2}=\epsilon, \]
hence
\[\mathbb{P}(\exists s \geq t: \theta(s) \in B_{\epsilon}(\hat{\theta}))>0.\]
This establish that $\hat{\theta
}$ is attainable by $\theta: \mathbb{R}_+ \rightarrow M$. 

Therefore, $\hat{\theta}$ is a sink (and hence a point attractor), satisfying $\hat{\theta} \in B(\hat{\theta})\cap Att(\theta)$. Theorem 7.3 of B99 (the ``attractor theorem'') states that under Condition (24) of B99 (verified above), if $Att(\theta) \cap B(\hat{\theta}) \neq \emptyset$, then $\mathbb{P}(L(\theta)\subseteq \{\hat{\theta}\}) >0$. This establishes that   $\{\theta_n\}$ converges to $\hat{\theta}$ with positive probability. \qed

\subsection{Proof of Theorem \ref{thm_learning}}\label{sec_learning_proof_withprior}
Define $d'_{n+1}=F_{n+1}(\theta_n)-F(\theta_n)+O(\sqrt{\log n}/n)$ as the new approximation error. The big $O$ term was derived in Appendix \ref{sec_learning_withprior}.

\vspace{-10pt}

\paragraph{Step 1.}  We re-establish Lemma \ref{lem_hbound}, Condition \ref{lem_mbound}, and Condition \ref{lem_fdifference} under general priors. 

Suppose that $\{m_n\}$ is a.s. bounded---an assumption verified later. Then, for large $n$, we have a.s. $m_n = O(1)$ and, from the derivation in Appendix \ref{sec_learning_withprior}, $v_n = \Theta(1)$. Based on these, we establish in the next paragraph that 
\[
\left\vert\mathbb{E}_{\beta \sim p_n}\left[\frac{\partial V_{\mathrm{E}}}{\partial h}\right]-\mathbb{E}_{\beta \sim \tilde{N}(m_n,v_n)}\left[\frac{\partial V_{\mathrm{E}}}{\partial h}\right]\right\vert = o\left(\frac{1}{n}\right).
\]  
Together with the proof of Lemma \ref{lem_hbound}, this implies
\[
\left\vert \mathbb{E}_{\beta \sim p_n}\left[\frac{\partial V_{\mathrm{E}}}{\partial h}\right]-\frac{\partial V_{\mathrm{E}}}{\partial h}(h,\tilde{m}_n)\right\vert = O\left(\frac{1}{\sqrt{n}}\right)  \ \text{ a.s. }
\]  
and therefore
\[
h_{n+1}(\theta_n)-h(\tilde{m}_n) = O\left(\frac{1}{\sqrt{n}}\right) \text{ a.s.}
\]

This analog of Lemma \ref{lem_hbound}(i) is precisely what was used in Appendix \ref{sec_learning_proof_withoutprior} to establish that $\{m_n\}$ is a.s. bounded. Hence, the assumption that $\{m_n\}$ is a.s. bounded is valid, and condition \ref{lem_mbound} --- which states that $\{\theta_n\}$ is a.s. bounded --- follows immediately. Moreover, $h_{n+1}(\theta_n)-h(\tilde{m}_n)= O(1/\sqrt{n})$ a.s. implies  $F_{n+1}(\theta_n)-F(\theta_n) = O(1/\sqrt{n})$ a.s. and hence 
\[
d_{n+1}'(\theta_n)  =F_{n+1}(\theta_n)-F(\theta_n)+O\left(\frac{\sqrt{\log n}}{n}\right) = O\left(\frac{1}{\sqrt{n}}\right) \text{ a.s.}
\]  
where the last equality uses the fact that $\log n$ is subpolynomial. It follows that Condition \ref{lem_fdifference} holds for $d'_{n+1}$.

The remainder of the proof distinguishes between two cases: 
\vspace{-10pt}
\subparagraph{Case 1: $m_n \in (\ul\beta,\bar\beta)$. } Fix any $\alpha \in (0,1/6)$. Define $S_n = [m_n-n^{-1/2+\alpha}, m_n+n^{-1/2+\alpha}]$. By Laplace's approximation,
\begin{align*}
\mathbb{E}_{\beta \sim p_n}\left[\dfrac{\partial V_{\mathrm{E}}}{\partial h}\right]
&=\dfrac{\displaystyle \int_{\ul\beta}^{\bar\beta} \dfrac{\partial V_{\mathrm{E}}}{\partial h}(h,\beta)\exp\left(-\dfrac{(\beta-m_n)^2}{2v_n}\right) \exp\left(K(\beta,m_n)(\beta-m_n)^3\right) d\beta}{\displaystyle \int_{\ul\beta}^{\bar\beta} \exp\left(-\dfrac{(\beta-m_n)^2}{2v_n}\right) \exp\left(K(\beta,m_n)(\beta-m_n)^3\right) d\beta}\\
& \leq \frac{\text{I}+\text{II}}{\text{III}+\text{IV},} \quad \text{ where }
\end{align*}
\begin{align*}
\mathrm{I} &= \left[\sup_{\beta \in S_n} \exp\left(K(\beta,m_n)(\beta-m_n)^3\right)\right] \;\left[\int_{\beta \in S_n} \dfrac{\partial V_{\mathrm{E}}}{\partial h}(h,\beta)
    \exp\left(-\dfrac{(\beta - m_n)^2}{2v_n}\right)  d\beta\right], \\
\mathrm{II}  & =  \left[\sup_{\beta \in [\ul\beta,\bar\beta]-S_n} \exp\left(K(\beta,m_n)(\beta-m_n)^3\right)\right]\;
    \left[\int_{\beta \in [\ul\beta,\bar\beta]-S_n} \dfrac{\partial V_{\mathrm{E}}}{\partial h}(h,\beta)
    \exp\left(-\dfrac{(\beta - m_n)^2}{2v_n}\right)  d\beta\right], \\
\mathrm{III} & = \left[\inf_{\beta \in S_n} \exp\left(K(\beta,m_n)(\beta-m_n)^3\right)\right]\; \left[\int_{\beta \in S_n} 
    \exp\left(-\dfrac{(\beta - m_n)^2}{2v_n}\right)  d\beta\right],\\
\mathrm{IV} & = \left[\inf_{\beta \in [\ul\beta,\bar\beta]-S_n} \exp\left(K(\beta,m_n)(\beta-m_n)^3\right)
\right] \; \left[\int_{\beta \in [\ul\beta,\bar\beta]-S_n}
    \exp\left(-\dfrac{(\beta - m_n)^2}{2v_n}\right)  d\beta\right].
\end{align*}
By comparison,
\[\mathbb{E}_{\beta \sim \tilde{N}(m_n,v_n)}\left[\dfrac{\partial V_{\mathrm{E}}}{\partial h}\right]=\left.\frac{\mathrm{I}+\mathrm{II}}{\mathrm{III}+\mathrm{IV}}\right\vert_{K(\beta,m_n) \equiv 0}. \]
Let $\Phi$ denote the c.d.f. of the standard normal distribution. Then 
 \begin{align*}
&\int_{\beta \in [\ul\beta,\bar\beta] \setminus S_n} 
\exp\left(-\frac{(\beta - m_n)^2}{2 v_n}\right) \, d\beta 
\leq 2 \sqrt{2\pi v_n} \left[ 1 - \Phi\left( \frac{\cancel{m_n} + n^{-\frac{1}{2} + \alpha} - \cancel{m_n}}{\sqrt{v_n}} \right) \right] \\
\tag{$\because$ Chernoff bound for normal distribution}&\leq 2 \sqrt{2\pi v_n}  \frac{\exp\left(-\frac{(n^{-\frac{1}{2} + \alpha})^2}{2 v_n}\right)}{\sqrt{2\pi} \cdot \frac{n^{-\frac{1}{2} + \alpha}}{\sqrt{v_n}}} \\
&\tag{$\because v_n = \Theta(1/n)$}= \frac{\exp\left(-\Theta\left(n^{2\alpha}\right)\right)}{\Theta\left(n^{\frac{1}{2} + \alpha}\right)}.
\end{align*}
Given that $\alpha \in (0,1/6)$, the upper bound in the last line vanishes exponentially in $n$ . It follows that for large $n$, II and IV become negligible, and so  
\begin{align*}
&\frac{\text{I}+\text{II}}{\text{III}+\text{IV}} - \left(\left.\frac{\mathrm{I}+\mathrm{II}}{\mathrm{III}+\mathrm{IV}}\right\vert_{K(\beta,m_n) \equiv 0}\right) \\
& \approx \frac{\text{I}}{\text{III}}-\left(\left.\frac{\text{I}}{\text{III}}\right\vert_{K(\beta,m_n) \equiv 0}\right)\\
&=\left(\dfrac{\sup_{\beta \in S_n} \exp\left(K(\beta,m_n)(\beta-m_n)^3\right)}{\inf_{\beta \in S_n} \exp\left(K(\beta,m_n)(\beta-m_n)^3\right)}-1\right) \; \left(\left.\dfrac{\text{I}}{\text{III}}\right\vert_{K(\beta,m_t)\equiv 0}\right)\\
&=O(n^{-3/2 + 3\alpha}) \cdot O(1) \\
&= o\left(\frac{1}{n}\right), 
\end{align*}
where the last two equalities use the facts that  \( K(\beta, m_n) = O(1) \) and \( \alpha \in (0, 1/6) \). A lower bound of the same order can be derived analogously.

\vspace{-10pt}
\subparagraph{Case 2: $m_n \notin (\ul\beta,\bar\beta)$.} We only prove the result for the case of $m_n<\ul\beta$. Fix any $\alpha \in (1/3,1)$, and redefine $S_n = [\ul\beta,\ul\beta+n^{-\alpha}]$. As in Case 1, we have 
\begin{align*}
\mathbb{E}_{\beta\sim p_n}\left[\dfrac{\partial V_{\mathrm{E}}}{\partial h}\right]\leq  \frac{\mathrm{I}+\mathrm{II}}{\mathrm{III}+\mathrm{IV}}, \  \mathbb{E}_{\beta \sim \tilde{N}(m_n,v_n)} \left[\frac{\partial V_{\mathrm{E}}}{\partial h}\right] =\left.\frac{\mathrm{I}+\mathrm{II}}{\mathrm{III}+\mathrm{IV}}\right\vert_{K(\beta,m_n) \equiv 0}. 
\end{align*}
Applying the Chernoff bound for normal random variables yields 
\[1-\Phi\left(\dfrac{\ul\beta-m_n}{\sqrt{v_n}}\right) \geq \frac{1}{\sqrt{2\pi}}\exp\left(-\frac{(\ul\beta-m_n)^2}{2 v_n}\right)\left(\frac{1}{\frac{\ul\beta-m_n}{\sqrt{v}_n}}- \frac{1}{\left(\frac{\ul\beta-m_n}{\sqrt{v}_t}\right)^3}\right)\]
and 
\[1- \Phi\left(\dfrac{\ul\beta-m_n+n^{-\alpha}}{\sqrt{v_n}}\right) \leq \frac{1}{\sqrt{2\pi}}\exp\left(-\frac{(\ul\beta-m_n+n^{-\alpha})^2}{2v_n}\right)\frac{1}{\frac{\ul\beta-m_n}{\sqrt{v}_n}}.\]
Taking ratio and simplifying using $m_n = O(1)$ and $v_n = \Theta(1/n)$ yields
\begin{align*}
\frac{1-\Phi\left(\dfrac{\ul\beta-m_n}{\sqrt{v_n}}\right)}{1- \Phi\left(\dfrac{\ul\beta-m_n+n^{-\alpha}}{\sqrt{v_n}}\right)}  
\geq \exp\left(\frac{2(\ul\beta-m_n)n^{-\alpha}+n^{-2\alpha}}{2v_n}\right)\left(1-O\left(\frac{1}{n}\right)\right) \\
= \exp\left(\Theta\left(n^{1-\alpha}\right)\right), \text{ which grows exponentially in } n \text{ given } \alpha<1.
\end{align*}
It follows that for large $n$, 
\begin{align*}
&\frac{\mathrm{I}+\mathrm{II}}{\mathrm{III}+\mathrm{IV}}-\left(\left.\frac{\mathrm{I}+\mathrm{II}}{\mathrm{III}+\mathrm{IV}}\right\vert_{K(\beta,m_n) \equiv 0}\right) \\
&\approx \frac{\mathrm{I}}{\mathrm{III}}-\left(\left.\frac{\mathrm{I}}{\mathrm{III}}\right\vert_{K(\beta,m_n)\equiv 0}\right) \\
&=\left(\frac{\sup_{\beta \in [\ul\beta,\ul\beta+n^{-\alpha}]} \exp\left(K(\beta,m_n)(\beta-m_n)^3\right)}{\inf_{\beta \in [\ul\beta,\ul\beta+n^{-\alpha}]} \exp\left(K(\beta,m_n)(\beta-m_n)^3\right)} -1\right)\; \left(\left.\frac{\mathrm{I}}{\mathrm{III}}\right\vert_{K(\beta,m_n)\equiv 0}\right)\\
&= o\left(\frac{1}{n}\right),
\end{align*}
where the last equality uses the fact that $m_n = O(1)$ and $\alpha>1/3$. 

\vspace{-10pt}

\paragraph{Step 2.} We weaken Condition \ref{lem_localbound} to the following: 

\begin{condition}\label{lem_relaxelocalbound}
Fix any saddle point $\hat{\theta}$ of the ODE, and let $U$ be the neighborhood of $\hat{\theta}$ in Lemma \ref{lem_hbound}(ii). For large $n$, we have $d_{n+1}'(\theta_n) = O(\sqrt{\log n}/n)$ if $\theta_n \in U$. 
\end{condition}

Condition \ref{lem_relaxelocalbound} holds a.s. in our setting. To verify, combine the proof of Lemma \ref{lem_hbound}(ii) and the derivation in Step 1. This yields $h_{n+1}(\theta_n)-h(\tilde{m}_n)= O(1/n)$ for $n$ large and $\theta_n \in U$. Substituting this into the definitions of $F_{n+1}$ and $F$ gives $F_{n+1}(\theta_n)-F(\theta_n)= O(1/n)$. Then 
\[d_{n+1}'(\theta_n)=F_{n+1}(\theta_n)-F(\theta_n)+O\left(\frac{\sqrt{\log n}}{n}\right) = O\left(\frac{\sqrt{\log n}}{n}\right).\]

\vspace{-10pt}
\paragraph{Step 3.} Since Conditions \ref{lem_stepsize}-\ref{lem_variance} are prior-independent, and  Conditions \ref{lem_mbound} and \ref{lem_fdifference} remain satisfied, applying the limit set theorem (Theorem 5.7 of B99) establishes a.s. convergence to the steady states of $\Phi$, as previously.

To establish zero-probability convergence to any saddle point $\hat{\theta}$, let $U$ be the neighborhood of $\hat{\theta}$ in Lemma \ref{lem_hbound}(ii). MY23's proof of their Proposition 5 uses Condition \ref{lem_localbound}, i.e., $F_{n+1}(\theta)-F(\theta)= O(1/n)$ across $\theta \in U$ for large $n$, on p.15 of their online appendix. Replacing their approximation error with $d'_{n+1} = O(\sqrt{\log n}/n)$, leads to the following change in MY23's proof: In Condition (32) of MY23, the term $O(S_n/n^2)$ on the right-hand side becomes $O(S_n \sqrt{\log n}/n^2)$. Everything else --- in particular Condition (33) of MY23 --- remains unchanged. Substituting Conditions (32) and (33) into p.709 of P90 and using the fact that $\log n$ is subpolynomial, we obtain the result on p.709 of P90, or the first display on p.16 of the online appendix of MY23.\footnote{In essense, conditions (22) and (23) of MY23 (after updating the approximation error) can be expressed as: $A_n \geq \frac{k_1 S_n^2}{n+1}+O(\frac{ S_n\sqrt{\log n}}{n^2})$ and $B_n \geq \frac{k_2}{n^2}$, where $S_n$ and the constants $k_1, k_2$ are all positive. Based on these, we  wish to demonstrate that $A_n+B_n \geq \frac{const}{n^2}$ for some positive constant. 

We modify the argument on p.709 of P90 using the fact that $\log n$ is subpolynomial. Specifically, fix any $\gamma \in (0,1)$. If $S_n>\frac{c}{n^{\gamma}}$, then the lower bound for $A_n$ is dominated by the term $\frac{k_1 S_n^2}{n+1}>0$, hence  $A_n+B_n\geq \frac{const}{n^2}$.  If $S_n \leq \frac{c}{n^{\gamma}}$, then the lower bound for $A_n$ may be dominated by the term $O(\frac{S_n\sqrt{\log n}}{n^2})$. In that case, $O(\frac{S_n\sqrt{\log n}}{n^2})=o(1/n^2)\ll B_n$, hence $A_n+B_n \geq \frac{const}{ n^2}$ remains valid.  }  Everything else remains unchanged.

\vspace{-10pt}
\paragraph{Step 4.} We demonstrate that for any given sink $\hat{\theta}$, convergence occurs with probability arbitrarily close to one under sufficiently concentrated priors. Again, we apply Theorem 7.3 of B99 (the ``attractor theorem''), which states that for any open set $U$ with compact closure $\bar U$ contained in $B(\hat{\theta})$, there exist numbers $\delta(U), T(U)>0$ depending on $U$ (these are derived in Lemma 6.8 of B99) such that
\[\mathbb{P}(L(\theta)\subseteq \{\hat{\theta}\}) \geq \left[ 1-w(t,\delta(U), T(U))\right]  \mathbb{P}\left(\exists s \geq  t: \theta(s) \in U\right).\]
As derived earlier, 
\[w(t,\delta, T)=C\exp\left(-\frac{(\delta/C(T)-\bar\gamma(t))^2}{C' T \bar\gamma(t)}\right),\]
where  $C$ and $C'$ depend only on $\dim(M)=2$ and $\bar\Gamma$, whereas $C(T)$ depends only on $T$ and $F$. The number $t$ on the right-hand side is taken to be sufficiently large so that $w(t, \delta(U), T(U))<1$. Now set $U=N_{\epsilon}(\hat{\theta})$, and let $\delta(U)$ and $T(U)$ be determined accordingly. As $\epsilon\rightarrow 0$ and the prior concentrates increasingly around $\theta$, the lower bound for $\mathbb{P}(L(\theta)\subseteq \{\hat{\theta}\})$ tends to $1$. \qed

\subsection{Proof of Theorem \ref{thm_learning_aa}}\label{sec_learning_proof_aa}
We prove the theorem under the uniform prior; the extension to general priors is analogous to Appendix \ref{sec_learning_withprior} and is therefore omitted.  

Under the uniform prior, posterior beliefs about $\beta_j$ at the end of period $n$ follow a truncated normal distribution $\tilde{N}(m_{j,n}, v_{j,n})$, where $m_{j,n}$ and $v_{j,n}$ are given in \eqref{eqn_mtgeneral} and \eqref{eqn_vtgeneral}, now indexed by $j$. Beliefs are independent across $j=1,\dots,J$. Define $\bm{\theta}_n = (\theta_{j,n})_{j=1}^J$ as the period-$n$ state, where $\theta_{j,n} \coloneqq (m_{j,n}, \xi_{j,n})$. The optimal assessment in period $n+1$, $h_{n+1}(\bm{\theta}_{n})$, solves \eqref{eqn_eproblem_learning_aa} and is abbreviated $h_{n+1}$. The Fisher information of $X_{j,n}$ as a signal of $\beta_j$ is $I_j(h_n)=g^2_{2,j}(h_n)h_n$.  

The evolution of $\bm{\theta}_n$ is governed by 
\begin{equation}\label{eqn_de_aa}
\bm{\theta}_{n+1}-\bm{\theta}_n=\frac{1}{n+1}\Bigl(F_{n+1}(\bm{\theta}_n)+ b_{n+1}(\bm{\theta}_n)\odot \bm{\epsilon}_{n+1}\Bigr), 
\end{equation}
where
\[
\renewcommand{\arraystretch}{1.1}
F_{n+1}(\bm{\theta}) \coloneqq
\begin{bmatrix}
\dfrac{(I_1 \circ h_{n+1})(\bm{\theta})\big[(\beta_1 \circ h_{n+1})(\bm{\theta}) - m_1\big]}
      {\tfrac{n}{n+1}\xi_1 + \tfrac{1}{n+1}(I_1 \circ h_{n+1})(\bm{\theta})}
\\
\vdots
\\
\dfrac{(I_J \circ h_{n+1})(\bm{\theta})\big[(\beta_J \circ h_{n+1})(\bm{\theta}) - m_J\big]}
      {\tfrac{n}{n+1}\xi_J + \tfrac{1}{n+1}(I_J \circ h_{n+1})(\bm{\theta})}
\\[4pt]
(I_1 \circ h_{n+1})(\bm{\theta}) - \xi_1
\\
\vdots
\\
(I_J \circ h_{n+1})(\bm{\theta}) - \xi_J
\end{bmatrix},
\]

\[
\renewcommand{\arraystretch}{1.1}
\begin{aligned}
b_{n+1}(\bm{\theta}) \coloneqq
\begin{bmatrix}
\dfrac{\sqrt{(I_1 \circ h_{n+1})(\bm{\theta})}}
      {\tfrac{n}{n+1}\xi_1 + \tfrac{1}{n+1}(I_1 \circ h_{n+1})(\bm{\theta})}
\\
\vdots
\\
\dfrac{\sqrt{(I_J \circ h_{n+1})(\bm{\theta})}}
      {\tfrac{n}{n+1}\xi_J + \tfrac{1}{n+1}(I_J \circ h_{n+1})(\bm{\theta})}
\\
\bm{0}_J
\end{bmatrix},
\qquad
\bm{\epsilon}_{n+1} \coloneqq
\begin{bmatrix}
\epsilon_{1,n+1}
\\
\vdots
\\
\epsilon_{J,n+1}
\\
\bm{0}_J
\end{bmatrix},
\end{aligned}
\]
and $\odot$ denotes the element-wise product. Derivations are analogous to those in Appendix \ref{sec_learning_withoutprior}.

For large $n$, $\bm{\theta}_n$ becomes almost degenerate at $\tilde{\bm{m}}_n= (\tilde{m}_{j,n})_{j=1}^J$. The optimal assessment given $\delta_{\tilde{\bm{m}}_n}$ is $h(\tilde{\bm{m}}_n)$, where $h(\cdot)$ solves \eqref{eqn_eproblem_aa}. In \eqref{eqn_de_aa}, dropping $\bm{\epsilon}_n$, replacing $h_{n+1}(\bm{\theta}_n)$ with $h(\tilde{\bm{m}}_n)$, and letting $n \rightarrow \infty$ yields the limiting system of ODEs
\begin{equation}\label{eqn_ode_aa}
\dot{\bm{\theta}}=F(\bm{\theta}),
\end{equation}
where 
\[
F(\bm{\theta}) \coloneqq  
\begin{bmatrix}  
\dfrac{(I_1 \circ h)(\tilde{\bm{m}}) \; \left[(\beta_1 \circ h)(\tilde{\bm{m}}) - m_1\right]}{\xi_1}  
\\[0.6em]  
\vdots  
\\[0.6em]  
\dfrac{(I_J \circ h)(\tilde{\bm{m}}) \; \left[(\beta_J \circ h)(\tilde{\bm{m}}) - m_J\right]}{\xi_J}  
\\[1.0em]  
(I_1 \circ h)(\tilde{\bm{m}}) - \xi_1  
\\[0.6em]  
\vdots  
\\[0.6em]  
(I_J \circ h)(\tilde{\bm{m}}) - \xi_J 
\end{bmatrix}.
\]

When ability misspecifications are small, $F$ admits a unique steady state $\hat{\bm{\theta}}= (\hat{\bm{\beta}}, \hat{\bm{\xi}})$, where $\hat{\bm{\beta}}$ was solved in the proof of Theorem \ref{thm_colorsighted} and was shown to satisfy $\hat{\bm{\beta}} \approx \bm{\beta}^*$. Thus $\widetilde{\hat{\bm{\beta}}}=\hat{\bm{\beta}}$. Linearizing $F$ around $\hat{\bm{\theta}}$ then yields 
\[
\left.\nabla F\right\vert_{\bm{\theta}=\hat{\bm{\theta}}}
=-Id_{2J} + 
\left.\begin{bmatrix}
\begin{array}{c|c}
\nabla_{\bm{m}}\bm{\psi} &  \bm{0}_{J \times J}\\ \hline 
(\tfrac{d}{dh}\bm{I}) (\nabla_{\bm{m}} h)^{\top} & \bm{0}_{J \times J}
\end{array}
\end{bmatrix}\right\vert_{\bm{\theta}=\hat{\bm{\theta}}}.
\]
The spectrum of this Jacobian  consists of the eigenvalue $-1$ (with multiplicity $J$) and the eigenvalues of the matrix $-Id_{J} + \left.\nabla_{\bm{m}} \bm{\psi}\right \vert_{\bm{m}=\hat{\bm{\beta}}}\coloneqq \left.\nabla_{\bm{m}}\Psi\right\vert_{\bm{m}=\hat{\bm{\beta}}}$. Below we demonstrate that all eigenvalues of $\left.\nabla_{\bm{m}}\Psi\right\vert_{\bm{m}=\hat{\bm{\beta}}}$---and consequently those of $\left.\nabla F\right\vert_{\bm{\theta}=\hat{\bm{\theta}}}$---have negative real parts. It follows that $\hat{\bm{\theta}}$ is a sink and therefore a point attractor. By the attractor theorem, the sequence $\{\bm{\theta}_n\}$ converges to $\hat{\bm{\theta}}$ with positive probability, and with probability arbitrarily close to one given a prior sufficiently concentrated around $\hat{\bm{\theta}}$. The details mirror those in Step 3 of Appendix \ref{sec_learning_proof_withoutprior} and Step 4 of Appendix \ref{sec_learning_proof_withprior}.

We  bound the eigenvalues of $\left.\nabla_{\bm{m}}\Psi\right\vert_{\bm{m}=\hat{\bm{\beta}}}$. The following matrix algebra results are useful:

\begin{fact}[\citealp{bauer1960norms}]\label{fact_bauerfike}
Let $A$ be a $J\times J$ diagonalizable matrix with decomposition $A=V\Lambda V^{-1}$, where $V$ is the matrix of eigenvectors and $\Lambda$ is the diagonal matrix of eigenvalues. For any $J \times J$ matrix $E$, every eigenvalue $\mu$ of $A+E$ satisfies:
\[|\mu-\lambda| \leq \|V\| \cdot \|V^{-1}\|\cdot \|E\|,\]
where $\lambda$ is some eigenvalue of $A$, and $\|\cdot\|$ is any matrix norm induced by a vector norm.
\end{fact}

\begin{fact}\label{fact_singularvalue}
The spectral norm (largest singular value) of a rank-one matrix $u v^{\top}$ is $\|u\|_2 \cdot \|v\|_2$.
\end{fact}

Applying the Bauer--Fike theorem, set $A= -Id_{J}$ (implying $\Lambda=\diag(-1,\dots, -1)$ and $V=Id_J$) and $E=\nabla_{\bm\beta}\bm{\psi}=g(\bm{\beta}; \bm{\Delta}) (\nabla_{\bm{\beta}}h)^{\top}$. Using the spectral norm $\|\cdot\|_2$, any eigenvalue $\mu$ of $\nabla_{\bm{\beta}}\Psi$ satisfies
\begin{align*}
|\mu+1|
&\;\leq\; \|V\|_2 \cdot \|V^{-1}\|_2 \cdot \|\nabla_{\bm{\beta}}\bm{\psi}\|_2 \\[6pt]
&\;\leq\; 1 \cdot 1 \cdot \|g(\bm{\beta};\bm{\Delta})\|_2 \cdot\|\nabla_{\bm{\beta}}h\|_2 \\[6pt]
&\;\leq\; \sup_{\bm{\beta}\in D}\|g(\bm{\beta};\bm{\Delta})\|_2 \cdot
            \sup_{\bm{\beta}\in D}\|\nabla_{\bm{\beta}}h\|_2,
\end{align*}
where the last term was shown to vanish as $\|\bm{\Delta}\|_2\rightarrow 0$ in Appendix \ref{sec_op}. Consequently, for small $\|\bm{\Delta}\|_2$, $\mu$ lies close to $-1$ and thus has a negative real part. \qed

\section{Selective oversight}\label{sec_oversight}
In this appendix, we apply the framework and results developed in the main text to study selective oversight in political economy.  

The agent is a politician (e.g., the President), the evaluator is an oversight authority (e.g., Congress), and the market consists of voters. The politician has innate ability $\eta$ and exerts costly effort $a$ to improve the policy outcome \(X=\eta+r(a,\beta)+\varepsilon\). The oversight authority can reduce noise in the outcome at cost $\kappa(h)$. Based on the realized outcome, voters form beliefs about the politician's ability. Electoral support depends on this belief and the expected value of $\delta r(a,\beta)$, where $\delta \in [0,1]$ captures the extent to which current policy benefits voters in the future. A standard assumption in political economy is that policy effort has only short-term impact, i.e., $\delta=0$. In such a case, $\beta$ is interpreted as the politician's aptitude for manipulating outcomes for electoral gain.  

Society—including both the oversight authority and voters—may hold stereotypical beliefs about the politician's ability, typically overestimating co-partisans and underestimating out-partisans. These misbeliefs are rooted in well-documented cognitive biases. For instance, \emph{familiarity bias} leads individuals to favor the familiar; \emph{false consensus} causes them to assume that their own views are widespread and to discredit dissenting perspectives; and \emph{in-group favoritism} fosters preferential treatment of ideological allies and hostility toward opponents. Such biases are pervasive among elites and even more so among the general public \citep{hertel2019legislative, druckman2023correcting, he2023learning}. Attempts to de-bias individuals through information or cross-cutting exposure often fail or even backfire \citep{nyhan2020facts}.  

Our analysis suggests that these biases may foster an \emph{ultimate attribution error}, whereby partisans attribute the success of co-partisan politicians to innate ability, but the success of out-partisans to manipulative effort, dismissed as insincere, opportunistic, or staged to win votes. Familiar claims by commentators include: “This is a stunt,” “They’re buying votes,” and “There are ulterior motives rather than genuine competence.”\footnote{A large literature in political science documents attribution errors in political accountability. Beyond the effect studied in this paper, ample evidence shows that voters often blame or credit politicians for factors beyond their control, such as economic cycles, natural disasters, or public health crises \citep{healy2013retrospective}.  \cite{glaeser2017fundamental} model how attribution errors can lead to democratic backsliding, e.g., reducing demand for institutional reforms and a free press, or by fostering a preference for dictatorship.} 

While this distortion may induce out-partisan politicians to exert greater effort, its implications for voter welfare are more nuanced and depend on whether effort yields long-term or only short-term effects. Under the true outcome distribution, voters' payoff net of their reward to the politician equals
$
-\Delta_{\mu} - \delta\,[R(h(\hat{\beta}),\hat{\beta})-R(h(\hat{\beta}),\beta^*)].$
When $\delta=0$, i.e., effort has no long-term effect, the second term vanishes. Voters then systematically over-reward co-partisans and under-reward out-partisans. When $\delta>0$, this prediction may reverse.  

Our analysis also suggests that out-partisan politicians face heightened scrutiny. A robust finding in the oversight literature is that Congress intensifies scrutiny of the executive branch under divided government, when the congressional majority opposes the President’s party. \cite{kriner2016investigating} document this pattern using data on nearly 13,000 congressional investigations from 1898–2014, measuring oversight intensity through the number of investigations and hearing days. Recent work by \cite{bellodi2024selective} analyzes congressional hearing transcripts using natural language processing and shows that co-partisan legislators are less probing about bureaucratic performance. Selective oversight also shapes the choice of targeted agencies: Republican-led committees often target the EPA, CFPB, and IRS, while Democratic-led committees scrutinize ICE and the DOJ’s civil rights division when a Republican is president. According to \cite{reynolds2023partisan}, this pattern has recently intensified amid rising polarization.  

Turning to comparative statics, a growing body of research on de-biasing political misperceptions has raised concerns about the cost-effectiveness and scalability of traditional interventions, such as information provision or cross-cutting exposure \citep{hartman2022interventions}. Our analysis suggests leveraging the ability–effort link to address this concern: correcting misbeliefs about a politician’s ability also reduces distortions in perceived effort. Moreover, our framework reveals new channels through which institutional factors and reforms may affect equilibrium distortions. One example is oversight budget and capacity, which influences the (shadow) cost of oversight $\kappa$ \citep{ban2023politicians}. Raising oversight budget and capacity reduces $\kappa$ and increases oversight intensity, thereby correcting distortions in perceived effort productivity and downstream outcomes in equilibrium.  

\newpage 
\begin{spacing}{.8}

\bibliographystyle{aer} 
\bibliography{misspecifiedcareerconcern.bib}
\end{spacing}

\end{document}